\documentclass[a4paper, notitlepage, 11pt]{article}
\usepackage[utf8]{inputenc}
\usepackage[english]{babel}
\usepackage[T1]{fontenc}
\usepackage{amsmath}

\usepackage{amsthm,amssymb,amscd, physics}
\usepackage{color}
\usepackage{savesym}
\usepackage{soul}
\usepackage{url}
\usepackage[normalem]{ulem}
\usepackage{enumitem}
\usepackage[left=2.5cm,right=2.5cm,bottom=2.5cm,top=2.5cm]{geometry}
\usepackage{tikz}
\usepackage{appendix}
\usepackage{csquotes}
\usepackage{float}   
\usepackage{fancyhdr}
\usepackage{mathtools}
\usepackage{stmaryrd}
\usepackage{thm-restate}
\usepackage{authblk}
\usepackage{amsmath}
\usepackage[sort&compress,square,numbers]{natbib}
\usepackage{diagbox}
\usepackage{bbm}
\savesymbol{div}
\usepackage{mathabx}
\usepackage{dsfont}
\usepackage[hidelinks,hypertexnames=false]{hyperref}

\restoresymbol{mathabx}{div}

\numberwithin{equation}{section}

\newcommand{\cket}[1]{\left|#1\right\rangle}

\newcommand{\mcl}[1]{\ensuremath{\mathcal{#1}}}

\newcommand{\C}{\ensuremath{\mathbb{C}}}
\newcommand{\R}{\ensuremath{\mathbb{R}}}
\newcommand{\CompFuncObs}[1]{\ensuremath{\mathcal{F}_{\C,#1}}}
\newcommand{\ReFuncObs}[1]{\ensuremath{\mathcal{F}_{\R,#1}}}
\NewDocumentCommand{\proj}{m}{|#1\rangle\langle #1|}
\newcommand{\id}[1]{\ensuremath{\mathbb{I}_{#1}}}
\newcommand{\VR}{\ensuremath{V_{\mathrm{KD,real}}}}
\newcommand{\SACondExp}[4]{\ensuremath{\mathbb{E}_{#1}^{#2,\mathrm{sa}}(#3|#4)}}

\newcommand{\Econd}{\mathcal{E}_{\mathrm{cond},Y}}

\usepackage{thmtools}
\usepackage{thm-restate}

\newtheorem{Theorem}{Theorem}[section]

\newtheorem{definition}[Theorem]{Definition}

\newtheorem{lemma}[Theorem]{Lemma}
\newtheorem{proposition}[Theorem]{Proposition}
\newtheorem{corollary}[Theorem]{Corollary}

\theoremstyle{definition}
\newtheorem{rmk}[Theorem]{Remark}

\newcommand{\IntEnt}[2]{\llbracket #1,#2\rrbracket}
\newcommand{\KDDist}[1]{\mathrm{Q}^{\mathrm{KD},#1}}
\newcommand{\KDSymb}[1]{\tilde{\mathrm{Q}}^{\mathrm{KD},#1}}
\newcommand{\QPEcond}[3]{\mathbb{E}^{#1}_{\hat{\rho}}(#2|#3)}
\newcommand{\QuantEcond}[3]{\mathbb{E}^{#1}_{\hat{\rho}}(#2|#3)}
\begin{document}

\title{What is special about the Kirkwood-Dirac distributions?\\
Only they produce natural conditional expectations.}
\author[1]{M. Spriet}
\author[1]{C. Langrenez}
\author[2]{R. Brummelhuis}
\author[1]{S. De Bièvre}

\affil[1]{Univ. Lille, CNRS, Inria, UMR 8524, Laboratoire Paul Painlevé, F-59000 Lille, France}
\affil[2]{Univ. Reims Champagne-Ardenne, CNRS, UMR 9008, Laboratoire Math\'ematique de Reims, F-51687 Reims, France}

\maketitle

\begin{abstract}
 Among the many quasiprobability representations of quantum mechanics,  the family of Kirkwood-Dirac (KD) representations  has come to the foreground in recent years. Each such KD  representation  is determined by the choice of two complementary complete sets of commuting observables $\hat A$ and $\hat B$ with respect to which it is Born-compatible, meaning that it correctly reproduces their Born probabilities for every state.    
 We identify in this paper what property uniquely characterizes the KD representations among all such $\hat A$ and $\hat B$ Born-compatible quasiprobability representations. For that purpose, we first define a natural notion of \emph{quantum conditional expectation} of an observable $\hat X$, given an observable $\hat Y$, in a state $\hat \rho$, as a best estimator and we show   that it has the basic properties generally expected of a conditional expectation.  We then show that only the KD representations provide a notion of conditional expectation, given $\hat B$ (or given $\hat A$) 
 that coincides with the above quantum conditional expectation. 
As a byproduct of our analysis, we show a state-dependent no-go theorem. We prove that, if the quantum conditional expectation of an observable $\hat X$, given an observable $\hat Y$ in a state $\hat \rho$ admits an anomalous value (meaning a value lying outside the interval $[x_{\min}, x_{\max}]$), then there cannot exist a Born-compatible joint probability distribution $\mu(x,y)$ for $\hat X$ and $\hat Y$ in the state $\hat \rho$ for which the associated conditional probability $\mu(x|y)$ 
yields a conditional expectation that coincides with the quantum conditional expectation.
 We further apply our findings to revisit a standard model for  phase estimation in quantum metrology. 
 We show in particular that, within the real sector of  a given KD representation, the classical Fisher information of this  phase estimation problem vanishes identically.  
\end{abstract}   

\medskip   
\tableofcontents
\section{\bf Introduction} 

Within quantum mechanics, it is not possible to properly define the joint probability distribution of two incompatible (\emph{i.e.} non-commuting) observables in a coherent manner for all states~\cite{ludwig1983,busch2003,ferrieetal2010, lostaglio2023kirkwood}.  This impossibility is one of the characteristic features of quantum mechanics, distinguishing it from classical mechanics. 
For the prototypical incompatible observables of position and momentum, Wigner famously found a way around this impossibility by associating to each quantum state a quasiprobability distribution for position and momentum, having the correct marginals, but that can take negative values~\cite{wigner1932}. A large variety of other quasiprobability distributions for position and momentum have subsequently been developed and have been proven useful in various fields of physics~\cite{weyl1927, Kirkwood33,Dirac45,moyal1949, Sudarshan63,Glauber63,cagl69a,cagl69b}, even though not all these constructions yield joint quasiprobability distributions: their marginals do not always coincide with the Born-rule probabilities for position and momentum. Extensions of these constructions to abelian groups other than the space translations, and their associated Heisenberg groups, have also been developed~\cite{wootters1987, 
cohendetetal1988,leonhardt1995,leonhardt1996, bouzdb96, FeichtingerKozek98,
wootters2004, gross2006,benyetal2025,debievreetal2025a,spriet2025, NicolaRiccardi2025}.    Other approaches include  coherent state constructions~\cite{klauderskagerstam1985,robertcombescure2021,gazeau2009,perelomov1986}. All these approaches  fit in a general unifying setup using frames and dual frames on the quantum Hilbert space $\mathcal{H}$. This setup  allows to define a quasiprobability distribution $Q(\hat{\rho})$ for all quantum states $\hat{\rho}$ and symbols $\tilde Q(\hat{X})$ for all operators $\hat{X}\in\mathcal{L}(\mcl{H})$ \cite{Fe11, ferrieetal2010}, both defined on a space $\Lambda$, thus forming a quasiprobability representation of quantum mechanics (see Section~\ref{s:QCEQP_new} for details). Here, $\mathcal{L}(\mcl{H})$ is the set of linear operators on $\mcl{H}$; we suppose throughout that $\dim\mathcal H=d<+\infty$. 
 
Among the countless quasiprobability representations of quantum mechanics so obtained feature the ones introduced by Dirac~\cite{Dirac45}.  Suppose one is given two  complete sets of commuting operators (CSCO) $\hat A$ and $\hat B$. Dirac then proposed for each quantum state $\hat{\rho}$ two {\it joint} quasiprobability distributions $Q^{\textrm{KD}, \ell}(\hat{\rho})$ and $Q^{\textrm{KD}, \mathrm r}(\hat{\rho})$ (both with the correct marginals, therefore) on the space $\sigma(\hat A)\times \sigma (\hat B)$, where $\sigma(\hat A)$ and $\sigma(\hat B)$ denote the spectra of $\hat A$ and $\hat B$:
\begin{equation}
    Q^{\textrm{KD}, \ell}_{a,b}(\hat{\rho})=\Tr(\hat{\Pi}_a^{\hat{A}}\hat \rho \hat{\Pi}_b^{\hat{B}}),\quad  Q^{\textrm{KD}, \mathrm r}_{a,b}(\hat{\rho})=\Tr(\hat{\Pi}^{\hat{B}}_ b\hat \rho \hat{\Pi}_a^{\hat{A}}),
\end{equation}
where $\hat{\Pi}_a^{\hat{A}}$ and $\hat{\Pi}_b^{\hat{B}}$ are the one-dimensional spectral projectors of $\hat A$ onto the eigenspace of $a\in\sigma(\hat A)$ and of $\hat B$ onto the eigenspace of $b\in\sigma(\hat B)$, respectively. Throughout, we will assume that $\hat A$ and $\hat B$ are complementary CSCO, meaning that $\hat\Pi_a^{\hat A}\hat\Pi_b^{\hat B}\neq0$ for all $a, b $. The reason for the choice of superscript $\ell$, for ``left'' and $\mathrm r$ for ``right'' will become clear below. 
 $Q^{\textrm{KD}, \ell}$ and $Q^{\textrm{KD}, \mathrm r}$ are today referred to as the Kirkwood-Dirac (KD) distributions~\cite{Dirac45}: they do indeed generalize the construction of Kirkwood~\cite{Kirkwood33} for the special case of position and momentum. The naturally associated  KD symbols $\tilde Q^{\textrm{KD}, \ell}(\hat{X})$ and  $\tilde Q^{\textrm{KD}, \mathrm r}(\hat{X})$ for any operator $\hat{X}$,
 \begin{equation}
 \tilde Q^{\textrm{KD}, \ell}(\hat{X})=\frac{\Tr(\hat{\Pi}_a^{\hat{A}}\hat X \hat{\Pi}_b^{\hat{B}})}{\Tr(\hat \Pi_a^{\hat A}\hat \Pi_b^{\hat B})},\quad \tilde Q^{\textrm{KD}, \mathrm r}(\hat{X})=\frac{\Tr(\hat{\Pi}_b^{\hat{B}}\hat X \hat{\Pi}_a^{\hat{A}})}{\Tr(\hat \Pi_b^{\hat B}\hat \Pi_a^{\hat A})},
 \end{equation}
 correspond  to a choice of ordering between the operators $\hat A$ and $\hat B$: $\hat A$ before $\hat B$ for $\ell$ and $\hat B$ before $\hat A$ for $\mathrm r$. Indeed, the KD symbol $\tilde Q^{\textrm{KD}, \ell}(\hat{X})$ of the operator $\hat{X}=f(\hat A)g(\hat B)$ is readily checked to be the function $(a,b)\mapsto f(a)g(b)$, whereas the KD symbol $\tilde Q^{\textrm{KD}, \textrm r}(\hat{X})$ of $\hat{X}=g(\hat B)f(\hat A)$ is also $(a,b)\mapsto f(a)g(b)$.  (See Section~\ref{s:QCEQP_new} for details.)  The flexibility of this construction, applying as it does to general non-commuting observables $\hat A$ and $\hat B$, has proven to be an asset in an increasing number of applications in the last few years; we refer to the recent reviews~\cite{lostaglio2023kirkwood,arvidssonshukur2024properties} for details.

The question then arises: what singles out the Kirkwood-Dirac quasiprobability representations of quantum mechanics among all possible such representations? The central result of  this paper is that they are, in a precise sense to be explained below (see Theorem~\ref{thm:KDunique1_bis}), the only ones that are well-behaved with respect to a natural notion of conditional expectation in quantum mechanics that we introduce.  Equivalently, the KD conditional expectation is  the only one that can be naturally interpreted as a best estimator and that, as such, coincides with the conditional expectation proposed in weak value physics. We present several applications of these results. We establish a simple state-dependent no-go result for the existence of joint probabilities in quantum mechanics (Section~\ref{s:Q_cond_exp_mean_var}). We show (Section~\ref{s:KD_phase_insensitive}) that in a standard model for phase estimation in quantum metrology, the Fisher information vanishes within the real sector of a given KD representation.  This result provides a novel interpretation of the real sector of KD representations and allows us to show that the optimal measurement yielding the quantum Fisher information cannot be KD real (Proposition~\ref{prop:LnotKDreal}). 
We further explain to what extent KD-positive states (for which $Q^{\mathrm{KD}}_{a,b}(\hat\rho)\geq 0$), can be considered ``classical'' by exhibiting some of their nonclassical features.

 In the rest of this introduction, we outline the paper. To precisely state, and then establish our results, we first need to revisit the notions of conditional expectation used in probability theory on the one hand (Section~\ref{s:clas_cond}) and in quantum mechanics on the other hand (Sections~\ref{s:Q_cond_exp_min} and~\ref{s:QCEQP_new}). We shall use two distinct approaches to the definition of quantum conditional expectation. 
One  approach consists in adapting to quantum mechanics the standard conditional expectation  $\mathbb{E }_{\mathbb{P } } ( X | Y ) $ of a random variable $X$ given a random variable $Y$, defined as  the complex-valued function $f(Y) $ of $Y$ that minimizes  the mean squared error
\begin{equation}\label{eq:L2classical}
     {\mathbb E_{\mathbb P}(|X-f(Y)|^2)}:=\int_\Omega |X(\omega)-f(Y(\omega))|^2 \mathrm d\mathbb P(\omega),
 \end{equation}
$\mathbb{P } $ being the underlying probability measure on the space $\Omega$ on which the random variables $X$ and $Y$ are defined. (See Section~\ref{s:clas_cond}.)  In statistics,  $\mathbb{E }_{\mathbb{P } } ( X | Y ) $  is referred to as the best predictor or best estimator of $X$ among all complex-valued functions of $Y$. The idea that underlies this definition is that the law of $\mathbb E_{\mathbb P}(X|Y)$ provides some information about the law of $X$.    
For example, it is easily checked that  their first moments agree:
\begin{equation}\label{eq:iterated_class}
    \mathbb E_{\mathbb P}(\mathbb E_{\mathbb P}(X|Y))=\mathbb E_{\mathbb P}(X).
\end{equation}
This relation, known as the iterated expectation formula, expresses the fact that the best estimator $\mathbb{E }_{\mathbb P }(X|  Y ) $ of $X$ is ``unbiased''.
In addition, one has the following well-known relation between their variances,  referred to as the law of total variance or the variance decomposition formula:
\begin{equation}\label{eq:variances_compare_class_intro}
\Delta^2 X=\Delta^2 \mathbb E_{\mathbb{P}}(X|Y) +  \mathbb{E}_{\mathbb{P}}(|X-\mathbb E_{\mathbb{P}}(X|Y)|^2)=\Delta^2 \mathbb E_{\mathbb P}(X|Y) +  \mathbb{E}_{\mathbb{P}}\left(\Delta^2_{\mathbb P}(X|Y)\right).
\end{equation}
This relation decomposes the total variability of $X$ in a term due to the variability of the conditional expectation $\mathbb E_{\mathbb{P}}(X|Y)$, plus an error term evaluating the fluctuations of $X$ around this conditional expectation. This error term equals the expected value of the conditional variance $\Delta^2_{\mathbb P}(X|Y)$ of $X$, given $Y$.  In particular, the smaller is this mean squared error, the closer are their variances. In that case, the variability of $X$ is considered to be well explained by that of $\mathbb E_{\mathbb P}(X|Y)$.  

In Section~\ref{s:Q_cond_exp_min}, we define a conditional expectation in the quantum mechanical context, using a similar minimization argument~\cite{dressel2015, brummelhuis2024,Brummelhuis2025a,Brummelhuis2025b}. We will then show that  identities (see Eq.~\eqref{eq:ConsMean} and~\eqref{eq:Variance}) analogous to, but in important ways strikingly different from   Eq.~\eqref{eq:iterated_class} and Eq.~\eqref{eq:variances_compare_class_intro}  hold for this quantum conditional expectation. 
Let $\hat{X}$ be  an operator on a Hilbert space $\mathcal{H}$ and let  $\hat{Y}$ be a CSCO. Then, for any  given mixed state $\hat{\rho} $ (positive trace-class operator of trace 1), we consider the following expression, 
\begin{equation}\label{eq:L2quantum}
 {\rm Tr } (\hat{\rho} (\hat{X} - f(\hat{Y}))^\dagger (\hat{X} - f(\hat{Y}) ))={\rm Tr}( (\hat{X}-f(\hat{Y}))\hat{\rho} (\hat{X}-f(\hat{Y}))^\dagger),
 \end{equation}
 where $f$ is complex valued, and which is a natural analog of Eq.\eqref{eq:L2classical}.
 We then  define, following~\cite{Brummelhuis2025a, Brummelhuis2025b}, the conditional expectation $\mathbb{E }_{\hat{\rho} }^\ell (\hat{X} | \hat{Y} ) $ as that operator function $f(\hat{Y})$ of $\hat{Y}$ which minimizes the expression in Eq.\eqref{eq:L2quantum}. (See Definition~\ref{def: left/right Econd}.)  As in classical probability theory, where the conditional expectation depends on the underlying probability $\mathbb P$, the quantum conditional probability $\mathbb{E }_{\hat{\rho} }^\ell (\hat{X} | \hat{Y} ) $ depends on the state $\hat \rho$. It can be  computed explicitly:
\begin{equation}\label{eq:left_cond_exp_intro}
    \mathbb{E}^{\ell}_{\hat{\rho}}(\hat{X}|\hat{Y})=\sum_{y\in \sigma(\hat{Y})}\frac{\langle \varphi_y^{\hat{Y}},\hat{X}\hat{\rho}\varphi_y^{\hat{Y}}\rangle}{\langle \varphi_y^{\hat{Y}},\hat{\rho}\varphi_y^{\hat{Y}}\rangle}\hat{\Pi}_y^{\hat{Y}},
\end{equation}
where, for technical reasons, and to ensure uniqueness of the solution,  $\hat\rho$ is supposed to belong to 
 \begin{equation}
     D_{\hat{Y}}=\{\hat{\rho}\in D(\mathcal{H}): \forall y\in \sigma(\hat{Y}), \ \Tr(\hat{\rho} \hat{\Pi}_y^{\hat{Y}})\neq 0\: \}.
 \end{equation}
  Note that $\mathbb E_{\hat\rho}^\ell(\hat X|\hat Y)$ is not necessarily self-adjoint, even if $\hat X$ is. This is a first marked difference from what happens in probability theory: if $X$ is a {\it  real} random variable, then the conditional expectation $\mathbb E_{\mathbb P}(X|Y)$ is also real-valued. The physical meaning of the real and imaginary parts of $\mathbb E_{\hat \rho}(\hat X|\hat Y)$ will be further discussed in Sections~\ref{s:Q_cond_exp_min} and~\ref{s:Q_cond_exp_int}. 

The coefficients in Eq.~\eqref{eq:left_cond_exp_intro} are so-called ``weak values''. They were introduced in~\cite{Aharonov88}, where they were given an operational meaning through an experimental procedure referred to as a weak measurement, that involves the weak coupling of the system to a meter whose position and/or momentum can be measured. For completeness and for the reader's convenience, we describe this setup and the relevant analysis in some detail in Appendix~\ref{app:weakvalues}.  The nature of this experimental procedure  has subsequently suggested the interpretation of the weak values as conditional expectations~\cite{hofmann2011,steinberg1995}. 
This course of events may leave one with the impression that the notion of conditional expectation in quantum mechanics depends on the experimental setup associated with weak values. Our result in Eq.~\eqref{eq:left_cond_exp_intro} shows that this is not the case.  On the contrary, weak values and their interpretation in terms of  quantum conditional expectations arise naturally, in an experiment-independent manner, from the definition of conditional expectations via a quadratic minimization problem analogous to the one used to define conditional expectations in probability theory. In particular, the conditional expectation of $\hat X$ given $\hat Y$ can be viewed as a ``best estimator'' of $\hat X$ by a function of $\hat Y$. Following this line of thought, the possibility of experimental determination of the coefficients of $\mathbb E^\ell_{\hat\rho}(\hat X|\hat Y)$ in Eq.~\eqref{eq:left_cond_exp_intro}  in a weak value experiments is then an a posteriori observation that provides an operational interpretation to the quantum conditional expectation.

Even if one does agree that defining conditional expectations via a minimization procedure is appealing and elegant, one may still wonder why one should choose to minimize a quadratic error.  In order to give additional justification for this choice beyond the fact that it provides a natural analogue of the classical definition and beyond the operational interpretation in terms of weak values just given, we provide, in Section~\ref{s:Q_econd_char}, a characterization of the quantum conditional expectation in Eq.~\eqref{eq:left_cond_exp_intro}:  we show in Theorem~\ref{thm:UniqueEcond} that, for $\hat \rho \in D_{\hat{Y}}$, the map $\hat X\mapsto \mathbb{E}_{\hat\rho}^\ell(\hat X|\hat Y)$ is uniquely characterized by the following two properties, naturally associated with a conditional expectation and which also characterize conditional  expectations in probability theory:
\begin{enumerate}   
\item [(i) ]  Pull-out formula: $\mathbb{E }_{\hat{\rho} }^\ell ( f(\hat Y) \hat{X} | \hat Y ) = f(\hat Y) \mathbb{E }_{\hat{\rho} }^\ell(\hat{X} | \hat Y ) $, for all functions $f(\hat Y) $ of $\hat Y$ and for all $\hat{X} \in \mathcal L(\mcl{H}) $;  
   
\item [(ii) ] Iterated expectations:  $\mathbb{E }_{\hat{\rho} } ( \mathbb{E }_{\hat{\rho} }^\ell ( \hat{X} | \hat Y ) ) = \mathbb{E } _{\hat{\rho} } (\hat{X} ) $.   
\end{enumerate}  
In (ii), we wrote $\mathbb{E }_{\hat{\rho} } (\hat{X} ) := {\rm Tr } ( \hat{X} \hat{\rho} ) $, to stress the similarity with the analogous properties that hold for the probabilistic notion of conditional expectation,  with $\hat{\rho} $ playing the r\^ole of the probability measure $\mathbb{P}$.  Property~(ii) expresses the fact that the best estimator $\mathbb{E }_{\hat{\rho} }^\ell(\hat{X} | \hat Y ) $ of $\hat X$ is ``unbiased'' in the sense that it has the same expected value (in the state $\hat \rho$) as $\hat X$ itself. It is the analog of Eq.~\eqref{eq:iterated_class} in probability theory. 

Finally, in Proposition~\ref{prop:variances_condexp1} we show a quantum version of the law of total variance
\begin{eqnarray}\label{eq:Variance_intro}
\Delta_{\hat \rho}^{\ell/\mathrm{r},2}\hat X 
 &=& \Delta_{\hat \rho}^{\ell/\mathrm{r},2}\mathbb E_{\hat \rho}^{\ell/\mathrm{r}}(\hat X|\hat Y)+\mathbb E_{\hat\rho}(\Delta^{\ell/r,2}_{\hat \rho} (\hat X|\hat Y)).
\end{eqnarray}
Here  $\Delta^{\ell/r,2}_{\hat \rho} (\hat X|\hat Y)$ is the conditional variance of $\hat X$, given $\hat Y$, defined in Eq.~\eqref{eq:cond_var}. This equation is clearly analogous to Eq.~\eqref{eq:variances_compare_class_intro}. 
It  shows that the variance of $\hat X$ equals the sum the variance of $\mathbb E_{\hat\rho}^{\ell}(\hat X|\hat Y)$ and of the expected value of the  conditional variance of $\hat X$, given $\hat Y$.  We use these results to derive a sharpened version of an additive uncertainty principle first proposed in~\cite{hofmann2011}, see Section~\ref{s:uncertainty}.

Having stressed the analogies between the quantum and classical (by which we mean here probabilistic) conditional expectations, we then identify a number of marked differences between them in Section~\ref{s:Q_cond_exp_mean_var} and explain their physical interpretation. We show in particular a simple but effective no-go theorem for the existence of joint probabilities. Consider $\hat \rho,\hat X$ and $\hat Y$  such that the quantum conditional expectation has a value $\mathbb E^{\ell}_{\hat\rho}(\hat X|\hat Y=y)$ that falls outside the interval $[x_{\min}, x_{\max}]$ (where $x_{\min,\max}$ are the extremal eigenvalues of $\hat X$): such values are said to be anomalous.   Then there does not exist a joint probability distribution for $\hat X$ and $\hat Y$ with the right Born-marginals in $\hat \rho$, and such that the corresponding conditional expectation for $\hat X$, given $\hat Y$, constructed with Bayes' rule, equals the quantum conditional expectation $\mathbb E^{\ell}_{\hat\rho}(\hat X|\hat Y=y)$. We refer to Lemma~\ref{lem:nogo} for a precise statement. The strength of this result lies in the fact that it puts a requirement only on a fixed triplet $(\hat\rho, \hat X,\hat Y)$ to preclude the existence of a joint probability. It also gives a precise meaning to the statement that the presence of anomalous values of the quantum conditional expectation points to ``nonclassicality''.

The superscript $\ell$ appearing in $\mathbb E_{\hat \rho}^{\ell}(\hat X|\hat Y)$ stands for ``left'' because in property (i) above $f(\hat Y)$ appears to the left of $\hat X$. 
In fact, although the above definition of a quantum conditional expectation as  a best estimator is very natural, the choice of  Eq.\eqref{eq:L2quantum} as the quantum equivalent to Eq.\eqref{eq:L2classical} of the quantity to be minimized  surreptitiously hides a choice of operator ordering. This should be expected, since the passage from the classical observable $\mathbb E_{\mathbb P}(X|Y)$ to a quantum equivalent involves  a choice  of quantization of observables, and as such can be expected to imply a choice of ordering. Indeed, an equally natural quantum analog of  Eq.\eqref{eq:L2classical} is
\begin{equation}\label{eq:L2quantumbis}
 {\rm Tr } (\hat{\rho} (\hat{X} - f(\hat{Y})) (\hat{X} - f(\hat{Y}) )^{{\dagger}})={\rm Tr}( (\hat{X}-f(\hat{Y}))^\dagger\hat{\rho} (\hat{X}-f(\hat{Y}))).
 \end{equation}
 Note that this quantity differs from Eq.\eqref{eq:L2quantum} only by the order in which $(\hat X-f(\hat Y))$ and $(\hat X-f(\hat Y))^\dagger$ appear. The same minimization procedure as above then leads to a different conditional expectation, that we denote by $\mathbb{E }_{\hat{\rho} }^{\mathrm r} (\hat X | \hat Y ) $. 
Again, the map $\hat X\mapsto \mathbb{E}_{\hat\rho}^{\mathrm r}(\hat X|\hat Y)$ is uniquely characterized by property (ii) above as well as by the following version of (i):
 \begin{equation}
\mathbb{E }_{\hat{\rho} }^{\mathrm r} (  \hat{X} f(\hat Y) | \hat Y ) = f(\hat Y) \mathbb{E }_{\hat{\rho} }^{\mathrm r}(\hat{X} | \hat Y ) , \quad \forall f(\hat Y),  \forall \hat{X} \in \mathcal L(\mathcal H),
\end{equation}
in which $f(\hat Y)$ appears to the right of $\hat X$. The ``left'' and ``right'' conditional expectation are simply related to each other:
\begin{equation}
\label{eq:link left right econd}
  \mathbb{E }_{\hat{\rho} }^{\mathrm r} (  \hat{X} | \hat Y )  =\mathbb{E }_{\hat{\rho} }^{\ell} (  \hat{X}^\dagger |\hat Y ) ^\dagger. 
\end{equation}
Contrary to the classical conditional expectation of a real random variable, which is real, the quantum conditional expectations $\mathbb E_{\hat\rho}^{\ell/\mathrm r}(\hat X|\hat Y)$ are not necessarily self-adjoint, even if $\hat X$ is. When, on the other hand, either $\mathbb E_{\hat\rho}^{\ell}(\hat X|\hat Y)$ or $\mathbb E_{\hat\rho}^{\mathrm r}(\hat X|\hat Y)$ is self-adjoint, for some self-adjoint $\hat X$, we have
\begin{equation}\label{eq:QCE_selfadjoint}
 \mathbb{E }_{\hat{\rho} }^{\mathrm r} (  \hat{X} | \hat Y )^{\dagger} =\mathbb{E }_{\hat{\rho} }^{\mathrm r} (  \hat{X} | \hat Y )  =\mathbb{E }_{\hat{\rho} }^{\ell} (  \hat{X} | \hat Y )=\mathbb{E }_{\hat{\rho} }^{\ell} (  \hat{X} | \hat Y )^\dagger. 
\end{equation}
In other words, in that case  the choice of ordering plays no role and the ``left'' and ``right'' conditional expectations are identical.   

There exists a second definition of conditional expectation that can be used in quantum mechanics, and that we revisit in Section~\ref{s:QCEQP_new}. It starts from the observation that, in probability theory, the notion of conditional expectation  $\mathbb{E }_{\mathbb{P } } ( X | Y ) $   of a random variable $X$, given a random variable $Y$, both defined on an underlying probability space $\Omega$ with probability $\mathbb P$ (both assumed to be discrete), can be defined in terms of their conditional probability distribution $\mathbb P(X=x|Y=y)$, which in turn is defined in terms of their joint probability distribution $\mathbb{P}(X=x, Y=y)$. Since, as pointed out above, such joint probabilities do not exist in quantum mechanics for non-commuting observables, this definition cannot straightforwardly be adapted to the quantum context. One approach to resolve this issue consists in replacing probabilities by quasiprobabilities and then proceeding in analogy with the classical treatment~\cite{steinberg1995,johansen2007}. As we explain in Section~\ref{s:QCEQP_new}, this approach leads to a notion of conditional expectation $\QPEcond{Q}{\hat{X}}{\hat{Y}}$ as an operator that  depends not only on $\hat \rho$, but also on the quasiprobability representation used in its definition.  In addition, as we will show, the conditional expectation $\QPEcond{Q}{\hat{X}}{\hat{Y}}$ so defined  does not generally have all the natural properties one would expect from a conditional expectation. In particular,  whereas it  does always satisfy the law of iterated expectations (property~(ii) above), it does  not in general satisfy the pull-out property (property~(i) above). This is not surprising since, as we saw, the pull-out property together with the law of iterated expectations does uniquely fix the definition of conditional expectation (Theorem~\ref{thm:UniqueEcond}). 

After these preparatory developments, we turn in Section~\ref{s:uniqueKD} and Section~\ref{s:QPR give unique Econd} to our principal result, which is the unique characterization of the KD quasiprobability representations, introduced in Section~\ref{s:KDdef}. We consider the set of all quasiprobability representations $(Q,\tilde Q)$, defined on some finite set $\Lambda$ with $d^2$ elements, that are Born-compatible with two given CSCO $\hat A$ and $\hat B$. This means that the marginals of the quasiprobability distribution $Q(\hat\rho)$ of any state $\hat \rho$ with respect to the symbols $\tilde Q(\hat A)$ and $\tilde Q(\hat B)$ yield the correct quantum mechanical Born probabilities for these two observables. (See Definition~\ref{def:Borncompat} for the precise definition.) Depending on the goal pursued, the set $\Lambda$ is sometimes referred to as the ``ontic space'' or simply as the ``phase space'' of the quasiprobability representation. Our central result concerning the KD representations of quantum mechanics is then summed up in the following theorem:

\begin{restatable}{thm}{maintheorem}
\label{thm:KDunique1_bis} Let $\hat A$ and $\hat B$ be complementary $\textrm{CSCO}$. Let $(Q,\tilde Q)$ be an $\hat A$ and $\hat B$-compatible quasiprobability representation of quantum mechanics defined on a set $\Lambda$ ($\left|\Lambda\right|=d^2)$. Then the following are equivalent:
\begin{itemize}
\item[(i)] $\forall \hat \rho\in D_{\hat B},\forall \hat X\in \mathcal L(\mathcal H), \quad \mathbb E^Q_{\hat \rho}(\hat X|\hat B)=\mathbb E^\ell_{\hat \rho}(\hat X|\hat B)$ 

\item[(ii)] $\forall \hat \rho\in D_{\hat A},\forall \hat X\in \mathcal L(\mathcal H), \quad  \mathbb E^Q_{\hat \rho}(\hat X|\hat A)=\mathbb E^{\mathrm{r}}_{\hat \rho}(\hat X|\hat A)$

\item[(iii)]   There exists a bijective map $\Phi :\Lambda\to \sigma(\hat A)\times \sigma(\hat B)$ such that for all $(a,b)\in\sigma(\hat A)\times \sigma(\hat B)$
\begin{equation*}
    \forall \hat \rho\in D(\mathcal{H}), Q_{\Phi^{-1}(a,b)}(\hat \rho)= Q^{\mathrm{KD},\ell}_ {a,b}(\hat \rho) \mathrm{ \ and \ } \forall \hat X\in \mathcal{L}(\mathcal{H}), \tilde{Q}_{\Phi^{-1}(a,b)}(\hat X)= \tilde{Q}^{\mathrm{KD},\ell}_ {a,b}(\hat X).
\end{equation*}
\end{itemize}
\end{restatable}
The theorem can be paraphrased by saying that of all $\hat A$ and $\hat B$ Born-compatible quasiprobability representations of quantum mechanics, only the left KD representation  has conditional expectations that coincide with the best estimators $\mathbb E_{\hat\rho}^\ell(\hat X|\hat B)$ and $\mathbb E_{\hat\rho}^{\mathrm r}(\hat X|\hat A)$. Theorem~\ref{thm:KDunique1_bis} is a generalization of a result announced by the authors in~\cite{spriet2025b}. In that work, the class of quasiprobability representations among which the KD representation is proven to be unique was considerably restricted by the a priori assumption that that $\Lambda=\sigma(\hat A)\times \sigma(\hat B)$. When using quasiprobability representations in the study of foundational issues, such as contextuality and hidden variables, this is not a natural assumption, that we lifted here. 

We will provide two independent proofs of this result. One uses the characterization of $\mathbb E^{\ell}(\hat X|\hat Y)$ via the pull-out property (Theorem~\ref{thm:KD_unique_pullout}). The other shows it follows from a slightly stronger statement (Theorem~\ref{thm:StrongerMain}) that is itself  based on Theorem~\ref{thm:QPR give unique Econd} which is of interest in its own right: it shows that, quite generally, $\hat A$ and $\hat B$ Born-compatible quasiprobability representations of quantum mechanics are always uniquely determined by their conditional expectations given $\hat A $ or given $\hat B$. 
The proof of Theorem~\ref{thm:QPR give unique Econd} uses the techniques developed here with an argument found in~\cite{jordan2026}, where a partial result similar to Theorem~\ref{thm:StrongerMain}, but limited to the case where $\Lambda=\hat\sigma(\hat A)\times \hat\sigma(\hat B)$, was recently announced and partially proven.

As we saw, one of the striking differences between classical and quantum conditional expectations is that the latter can be non-self-adjoint. In addition, even if they are self-adjoint, they can take classically forbidden values (Lemma~\ref{lem:nogo}) and, in this sense, still signal typical quantum behaviour such as, for example, weak value amplification or quantum advantage in quantum metrology~\cite{Arvidsson-Shukur2020, lupu2022negative,arvidsson2023closed,arvidssonshukur2024properties}. In Section~\ref{s:Q_cond_exp_int} we further analyse the meaning of the imaginary part of the quantum conditional expectation values and more specifically of its vanishing. For that purpose, we first recall the role played by the imaginary part of the quantum conditional expectation in a well-known problem of phase estimation in quantum mechanics. We consider a one-parameter family of states 
\begin{equation}
\hat\rho_{\hat X}(\theta)=\exp(-i\theta \hat X)\hat\rho \exp(i\theta\hat X),
\end{equation}
where $\theta\in\R$ is referred to as the phase and $\hat X$ is self-adjoint. The outcome probabilities of repeated measurements of a CSCO $\hat Y$ in $\hat\rho_{\hat X}(\theta)$  are given by 
$p(y;\theta)=\Tr(\hat\Pi_y^{\hat Y}\hat\rho_{\hat X}(\theta))$ ($y \in \sigma (\hat{Y } ) $), and we write $\mathrm{I_F}(\hat Y;\theta)$ for the  Fisher information of $p(y;\theta)$. We first then recall the well-known relation~\cite{hofmann2011, Dressel_Jordan_2012_ImaginaryPartWeakValue}  between the above Fisher information $\mathrm{I_F}(y;\theta)$ and the variance of the imaginary part of $\mathbb E_{\hat\rho_{\hat X}(\theta)}^{\ell/\mathrm{r}}(\hat X|\hat Y)$: 
 \begin{equation} \label{eq:FI_intro}     
\mathrm{I}_{\mathrm{F}} (\hat Y;\theta ) = 
4 {\rm Tr } \left( \left( {\rm Im } \, \mathbb{E } ^{\ell/\mathrm r }_{\hat{\rho }_{\hat{X } } (\theta ) } (\hat{X } | \hat{Y } )\right) ^2 \hat{\rho }_{\hat{X } } (\theta ) \right) .      
\end{equation}
Since it is well known that, in order to be able to obtain a good estimate on $\theta$, one needs a large value of the Fisher information, the imaginary part of the conditional expectation $\mathbb{E } ^{\ell/\mathrm r }_{\hat{\rho }_{\hat{X } } (\theta ) } (\hat{X } | \hat{Y } )$ contains information on the phase $\theta$. 
For that reason, we will say that  a state is phase insensitive (for a given choice of $\hat X$ and of $\hat Y$) if its Fisher information $\mathrm{I_F}(\hat Y;0)$ vanishes. To further justify this terminology, we give an operational interpretation of the vanishing of a Fisher information in Appendix~\ref{app:vanishingFI}. As pointed out above, in the context of the phase estimation problem considered, vanishing of the associated Fisher information is equivalent to saying that the conditional expectation $\mathbb E_{\hat \rho}^{\ell/\rm r}(\hat X|\hat Y)$ is self-adjoint.
When this is the case, the probability distribution $p(y;\theta)$ cannot be used to efficiently estimate $\theta$.  Using the quantum law of total variance (Eq.~\eqref{eq:Variance_intro}), we sharpen an additive uncertainty relation first established in~\cite{hofmann2011} between the Fisher information $I_{\mathrm{F}}(\hat Y,0)$ and the variance of the real part of the conditional expectation, denoted by $\mathbb E_{\hat\rho}^{\rm sa}(\hat X|\hat Y)$. We show
\begin{equation}
    \min_{\hat Y} \Delta_{\hat \rho}^{2}\mathbb E_{\hat \rho}^{\mathrm{sa}}(\hat X | \hat Y) \leq \Delta_{\hat \rho}^{2}\mathbb E_{\hat \rho}^{\mathrm{sa}}(\hat X | \hat L)\leq\Delta^{2}_{\hat \rho} \hat X -\frac14I_{\mathrm{QF}}(0)\leq \Delta^{2}_{\hat \rho} \hat X= \max_{\hat Y} \Delta_{\hat \rho}^{2}\mathbb E_{\hat \rho}^{\mathrm{sa}}(\hat X | \hat Y).
\end{equation}
Here $I_{\mathrm{QF}}(0)$ is the quantum Fisher information of $\hat\rho$ and $\hat X$ and $\hat L$ is the symmetric logarithmic derivative of $\hat \rho_{\hat X}(\theta)$. (See Appendix~\ref{app:Fisher} for an introduction to the quantum Fisher information.) When $\hat\rho$ is pure, one more precisely has
\begin{equation}
    \min_{\hat Y} \Delta_{\hat \rho}^{2}\mathbb E_{\hat \rho}^{\mathrm{sa}}(\hat X | \hat Y)= \Delta_{\hat \rho}^{2}\mathbb E_{\hat \rho}^{\mathrm{sa}}(\hat X | \hat L)=\Delta^{2}_{\hat \rho} \hat X -\frac14I_{\mathrm{QF}}(0)=0,\quad  \Delta^{2}_{\hat \rho} \hat X= \max_{\hat Y} \Delta_{\hat \rho}^{2}\mathbb E_{\hat \rho}^{\mathrm{sa}}(\hat X | \hat Y),
\end{equation}
where the maximum on the right is reached when $\hat \rho$ is phase insensitive for the given choice of $\hat X$ and $\hat Y$.  In other words, when $\hat Y$ is chosen close to $\hat L$, then the Fisher information $I_{\mathrm F}(\hat Y;0)$
is close to its maximal value $I_{\mathrm{QF}}$ and the variance of $\mathbb E^{\rm sa}_{\hat \rho}(\hat X|\hat Y)$ is close to its minimal value; and vice versa.

In light of this analysis, it is of interest to know for which triplets $\hat\rho, \hat X, \hat Y$, the Fisher information $\mathrm{I_F}(\hat Y;0)$ vanishes; we call such triplets phase insensitive. For that purpose, we place ourselves in the framework of KD representations and exploit Theorem~\ref{thm:KDunique1_bis}. We show that, if $\hat\rho$ and $\hat X$ have a real KD symbol with respect to some CSCO $\hat A$ and $\hat B=\hat Y$, then  $\mathrm{I_F}(\hat Y;0)$ vanishes (Section~\ref{s:KD_phase_insensitive}); hence $(\hat \rho, \hat X, \hat B)$ form a phase-insensitive triplet.   This observation provides an interpretation of the ``real sector'' of a given KD representation and shows it is inefficient for the above phase estimation problem.
Since recent work has shown that the set of all self-adjoint operators with a real KD symbol can in many cases be explicitly identified~\cite{langrenez2023characterizing2, debievreetal2025a, spriet2025}, these results also provide many examples of phase insensitive states. 
Finally, supposing that $\hat\rho$ and $\hat X$ are KD real, we investigate their associated quantum Fisher information. We show that the latter is non-vanishing if and only if $\hat \rho$ and $\hat X$ do not commute. It is well known that it is obtained by optimizing the choice of observable $\hat Y$, taking it to be equal to the symmetric logarithmic derivative $\hat L$ of $\hat\rho_{\hat X}(0)$.   We then show that $\hat L$ cannot be KD real when $\hat \rho$ is pure (Proposition~\ref{prop:LnotKDreal}). In other words, if $\hat \rho$ and $\hat X$ are KD-real, then the optimal observable $\hat L$ is not KD real.

It should be noted that a triplet $(\hat \rho, \hat X, \hat B)$ can be phase insensitive, with $\hat\rho$ KD positive, while the conditional expectation $\mathbb E^{\ell}_{\hat\rho}(\hat X|\hat B)$ still allows for anomalous values. This shows once again that the word ``classical'' needs to be given a precise meaning when used.

We have completed this paper with a number of appendices with the goal to make it reasonably self contained and to provide, for the reader's convenience, proofs of important properties not necessarily easily accessible in the literature. In Appendix~\ref{app:weakvalues}, we provide a short, but rather complete introduction to weak value theory and in particular to weak value measurement. We have included a discussion of the conditional expectations and variances of meter positions in the weak limit which connects with recent work~\cite{Ogawa_2021}. In Appendix~\ref{app:vonNeumann}, we briefly explain the link between our definition of conditional expectations and a definition used in the context of von Neumann algebras, restricted to the finite dimensional setup that is ours. 
Appendix~\ref{app:Wigner_KD} contains illustrative examples and counterexamples. Appendix~\ref{app:vanishingFI} gives operational interpretations of the vanishing of the Fisher information for parameter estimation. Appendix~\ref{app:Fisher} provides an introduction to the quantum Fisher information. Using the ideas developed in the paper, we give a simple proof of the result of Braunstein and Caves showing that the Helstr\"om-Holevo quantum Fisher information is  obtained as the maximum over all possible measurements of the classical Fisher information associated to those measurements. In Appendix~\ref{app:variational_char} we provide a variational characterization of the KD distribution of a state.

\medskip

 It is a pleasure to dedicate this work to Jean-Pierre Gazeau on the occasion of his 80th birthday. Jean-Pierre has been for decades an efficient and enthusiastic advocate of phase space representations of quantum mechanics, coherent states and frames, and their applications in quantum theory and signal analysis. We hope he will find our contribution of interest. SDB is particularly grateful to Jean-Pierre for having welcomed him in the French scientific world over 35 years ago and for his continued support and friendship throughout all these years. 

\section{Classical conditional expectation and conditional variance}\label{s:clas_cond}

In this section, we collect some basic definitions and properties of the classical conditional expectation as used in  probability theory.
This allows us to fix notation and prepare for the quantum formulation.

Let $\Omega$ be a finite set. Let $X:\Omega \to \mathbb{C}$ be a complex random variable and $Y:\Omega \to \mathbb{R}^m$ be a real vector-valued random variable. We write $\mathrm{Ran}(X)$ for the range of $X$, meaning
\begin{equation}
 \mathrm{Ran}(X)=\{x\in\mathbb C\mid \exists \omega \in \Omega, X(\omega)=x\}=X(\Omega),
\end{equation}
and similarly for $Y$. We denote by $D_{Y}$ the set of probability measures $\mathbb{P}$ on $\Omega$ such that for each $y\in \mathrm{Ran}(Y)$, $\mathbb{P}(Y=y)> 0$. The space of complex random variables on $\Omega$ is denoted by $\mathcal F(\Omega;\mathbb C)$. 
We first recall the elementary definition of ``conditional expectation'' of a random variable. 

\begin{definition}\label{def:condexpclass1}
Let $\mathbb{P}\in D_Y$. For $y\in \mathrm{Ran}(Y)$ and $x\in \mathbb{C}$, the conditional probability that $X$ takes the value $x$ knowing $Y=y$ is given by
\begin{equation}
    \mathbb{P}(X=x|Y=y)=\frac{\mathbb{P}(X=x,Y=y)}{\mathbb{P}(Y=y)}.
\end{equation}
The conditional expectation of $X$ knowing that $Y=y$ is  
\begin{equation}\label{eq:condXYprimitive}
    \mathbb{E}_{\mathbb{P}}(X|Y=y)=\sum_{x\in \mathrm{Ran(X)} }x\mathbb{P}(X=x|Y=y)=\sum_{x\in \mathrm{Ran(X)}}x\frac{\mathbb{P}(X=x,Y=y)}{\mathbb{P}(Y=y)}.
\end{equation}
\end{definition}
We will, as  usual, define  the conditional expectation of $X$, knowing $Y$, denoted by $\mathbb{E}_{\mathbb{P}}(X|Y)$, as the function
\begin{equation}
    \mathbb{E}_{\mathbb{P}}(X|Y)(y)=    \mathbb{E}_{\mathbb{P}}(X|Y=y).
\end{equation}
Note that $\mathbb{E}_{\mathbb{P}}(X|Y)$ is a random variable that is a function of $Y$. Also, 
\begin{equation}
\label{eq:f(Y)_ind_Y_clas}
\mathbb{E}_{\mathbb P}(f(Y)|Y)=f(Y)  \mathrm{\ and \ } \mathbb{E}_{\mathbb P}\left(\mathbb{E}_{\mathbb P}(X|Y)\right)=\mathbb{E}_{\mathbb P}(X),
\end{equation}
as is readily checked.  Finally, if $X$ is a {\it  real} random variable, then $\mathbb E_{\mathbb P}(X|Y)$ is also real-valued. As we will see, this is a marked difference from what happens in quantum mechanics, where the conditional expectation of an observable can be non self-adjoint.

We note that this definition, and all the properties that follow in this section, are strongly dependent on the initially chosen probability $\mathbb{P}$ on $\Omega$. We indicate this dependence in the notation, to anticipate on   the  quantum conditional probability defined in the next section, which also depends on  the state $\hat{\rho}$ considered.

We now provide two distinct characterizations of $\mathbb{E}_{\mathbb{P}}(X|Y)$ that will be essential in the following sections in order to define a quantum conditional expectation. 
For $\mathbb{P}\in D_{Y}$, we equip $\mathcal F(\Omega; \mathbb C)$ with the following sesquilinear form:
\begin{equation}\label{eq:L2Omegainnerproduct}
    \langle X, X'\rangle_{\mathbb{P}} ={\mathbb E}_{\mathbb P} (\overline X X')=\int_{\Omega} \overline X(\omega) X'(\omega)\mathrm d\mathbb{P}=\sum_{\omega\in\Omega} \overline X(\omega)X'(\omega) \mathbb P(\{\omega\}).
\end{equation}
\begin{Theorem}
\label{thm:classical_minimization}
Let $\mathbb{P}\in D_Y$. The conditional expectation of $X$ knowing $Y$ is the unique minimizer of
\begin{equation}\label{eq:minimize_dist_Y}
    f\mapsto \mathbb{E}_{\mathbb{P}}(|X-f(Y)|^2)=\langle X-f(Y), X-f(Y)\rangle_{\mathbb{P}},
\end{equation}
over the set of 
complex valued functions $f$ defined on $\mathrm{Ran}(Y)$.
\end{Theorem}
We denote by $\CompFuncObs{Y}$ the set:
\begin{equation}
    \CompFuncObs{Y} = \left\{ f(Y) \mid f : \mathrm{Ran}(Y)\to \mathbb{C} \right\}.
\end{equation}
 $\CompFuncObs{Y}$ is thus the set of all complex random variables on $\Omega$ that are functions of $Y$. In other words, $Z:\Omega\to \mathbb C$ belongs to $\CompFuncObs{Y}$ provided $Z$ is constant on all level sets of $Y$. This means there exists a function $f:\mathrm{Ran}(Y)\to \mathbb{C}$ so that $Z=f(Y)$. Note that $\CompFuncObs{Y}$ is a complex vector space ($\CompFuncObs{Y}\subset \mathcal F (\Omega;\mathbb C)$); it is  in addition closed under multiplication of functions. In the usual language of Hilbert space theory, the minimizer of Eq.\eqref{eq:minimize_dist_Y} -- which is the conditional expectation of $X$ given $Y$ --  is the orthogonal projection of $X$ onto $\CompFuncObs{Y}$. In more advanced treatments of probability, where the space $\Omega$ is not finite and where the probability measures are general, this property provides a simple way to  \emph{define} the conditional expectation, since Definition~\ref{def:condexpclass1} then does not necessarily make sense. We will see in the next section that this definition adapts nicely to quantum theory as well. 
\begin{proof}
For $f : \mathrm{Ran}(Y) \to \mathbb{C}$,
we have 
\begin{equation}
    f(y)=\sum_{y'\in \mathrm{Ran}(Y)}f(y')\mathds{1}_{y'}(y)
\end{equation}
where $\mathds{1}_{y'}(y)$ is $1$ if $y=y'$ and $0$ otherwise.
One then computes
\begin{align}
    \mathbb{E}_{\mathbb{P}}(|X-f(Y)|^2)&=\mathbb{E}_{\mathbb{P}}(|X|^2) - 2\sum_{y\in \mathrm{Ran}(Y)}\mathrm{Re}\left(\overline{f(y)}\mathbb{E}_{\mathbb{P}}(X\mathds{1}_{y}(Y))\right) + \sum_{y\in \mathrm{Ran}(Y)}|f(y)|^2\mathbb{P}(Y=y)\nonumber\\
    &=\mathbb{E}_{\mathbb{P}}(|X|^2) + \sum_{y\in \mathrm{Ran}(Y)}\mathbb{P}(Y=y)\left(|f(y)|^2-2\frac{\mathrm{Re}\left(\overline{f(y)}\mathbb{E}_{\mathbb{P}}(X\mathds{1}_{y}(Y))\right)}{\mathbb{P}(Y=y)}\right)\nonumber\\
    &=\mathbb{E}_{\mathbb{P}}(|X|^2) - \sum_{y\in \mathrm{Ran}(Y)}\frac{\left|\mathbb{E}_{\mathbb{P}}(X\mathds{1}_{y}(Y))\right|^2}{\mathbb{P}(Y=y)} + \sum_{y\in \mathrm{Ran}(Y)}\mathbb{P}(Y=y)\left|f(y)-\frac{\mathbb{E}_{\mathbb{P}}(X\mathds{1}_{y}(Y))}{\mathbb{P}(Y=y)}\right|^2.\label{eq:variances_compare}
\end{align}
This quantity is minimal if and only if for all $y\in \mathrm{Ran}(Y)$,
\begin{equation}
    f(y)=\frac{\mathbb{E}_{\mathbb{P}}(X\mathds{1}_{y}(Y))}{\mathbb{P}(Y=y)}.
\end{equation}
Moreover, we have that
\begin{align*}
    \mathbb{E}_{\mathbb{P}}(X\mathds{1}_{y}(Y))&=\sum_{(x,y')\in \mathrm{Ran}(X)\times \mathrm{Ran}(Y)}\mathbb{P}(X=x,Y=y')x\mathds{1}_{y}(y')\\
    &=\sum_{x\in \mathrm{Ran}(X)}x\mathbb{P}(X=x,Y=y)\\
    &=\mathbb{P}(Y=y)\mathbb{E}_{\mathbb{P}}(X|Y=y).
\end{align*}
This concludes the proof.
\end{proof}
From Eq.~\eqref{eq:f(Y)_ind_Y_clas}, we recall the law of iterated expectations,
\begin{equation}
\label{eq:iterated_ex_classical}
\mathbb E_{\mathbb{P}}\left(\mathbb E_{\mathbb{P}}(X|Y)\right) = \mathbb E_{\mathbb{P}}\left(X\right)
\end{equation}
and we note, for later reference, that Eq.~\eqref{eq:variances_compare} can be rewritten in the form
\begin{equation}\label{eq:variances_compare_class}
\Delta^2 X=\Delta^2 \mathbb E_{\mathbb P}(X|Y) +  \mathbb{E}_{\mathbb{P}}(|X-\mathbb E_{\mathbb P}(X|Y)|^2),
\end{equation}
 which is the law of total variance, where
\begin{equation}
    \Delta^2Z:=\mathbb{E}_{\mathbb{P}}(|Z-\mathbb{E}_{\mathbb{P}}(Z)|^2)
\end{equation}
is the variance of the random variable $Z:\Omega\to \mathbb{C}$. 
In other words, $X$ and its best estimator $\mathbb E(X|Y)$ have the same expected value and their variances differ by the minimal ``cost''. 
Introducing the conditional variance
\begin{equation}
    \Delta^2_{\mathbb P}(X|Y)=\mathbb E_{\mathbb P}(|X-\mathbb E_{\mathbb P}(X|Y)|^2|Y),
\end{equation}
the law of total variance reads alternatively:
\begin{equation}\label{eq:variances_compare_class_bis}
\Delta^2 X=\Delta^2 \mathbb E_{\mathbb P}(X|Y) +  \mathbb{E}_{\mathbb{P}}\left(\Delta^2_{\mathbb P}(X|Y)\right).
\end{equation}

 We note for later reference  that 
\begin{equation}\label{eq:Konig}
    \Delta_{\mathbb P}^2(X|Y)=\mathbb E_{\mathbb P}(|X|^2)-|\mathbb E_{\mathbb P}(X|Y)|^2.
\end{equation}
As we shall now prove, the conditional expectation is also characterized by two simple properties:
\begin{Theorem}\label{thm:econd_characterize2}\
Let $Y$ be a real vector-valued random variable on $\Omega$ and $\mathbb{P}\in D_Y$. Then there exists a unique linear map $\Econd: \mathcal F (\Omega; \mathbb C)\to \CompFuncObs{Y}\subset \mathcal F(\Omega; \mathbb C)$ 
that satisfies the following properties:
\begin{enumerate}
    \item  Pull-out formula. For any complex random variable $X$ on $\Omega$ and any $f:\mathrm{Ran}(Y) \to \mathbb{C}$, 
    \begin{equation}\label{eq:cond1}
\Econd(f(Y)X)=f(Y)\Econd(X);
    \end{equation}
   \item The law of iterated expectations. For any complex random variable $X$ on $\Omega$, 
   \begin{equation}\label{eq:cond2}
       \mathbb{E}_{\mathbb{P}}(\Econd(X))=\mathbb{E}_{\mathbb{P}}(X).
        \end{equation}
\end{enumerate}

 One has $\Econd(X)=\mathbb{E}_{\mathbb{P}}(X|Y).$
\end{Theorem}
\begin{proof}
Existence is clear, since it is easily checked that $X\mapsto \mathbb{E}_{\mathbb{P}}(X|Y)$ satisfies the two desired properties. It remains to show uniqueness. 
Let $\Econd: X\mapsto \Econd(X)\in\CompFuncObs{Y}$ satisfy properties $\textit{1}$ and $\textit{2}$.
Let $X$ be a random variable. For any $f(Y)\in \CompFuncObs{Y}$, we have
\begin{align*}
    \langle f(Y),X-\Econd(X)\rangle_{\mathbb{P}} &= \mathbb{E}_{\mathbb{P}}(\overline{f(Y)}(X-\Econd(X)))\\
    &=\mathbb{E}_{\mathbb{P}}(\overline{f(Y)}X) - \mathbb{E}_{\mathbb{P}}(\Econd(\overline{f(Y)}X)) \text{ by $\textit{1}$}\\
    &= \mathbb{E}_{\mathbb{P}}(\overline{f(Y)}X)-\mathbb{E}_{\mathbb{P}}(\overline{f(Y)}X) \text{ by $\textit{2}$}\\
    &=0.
\end{align*}
We therefore also have $\langle f(Y), X-\mathbb{E}_{\mathbb{P}}(X|Y)\rangle_{\mathbb{P}} =0$. It follows that
\begin{align*}
    &\langle \Econd(X)-\mathbb{E}_{\mathbb{P}}(X|Y),\Econd(X)-\mathbb{E}_{\mathbb{P}}(X|Y)\rangle_{\mathbb{P}}\\
    =&\langle \Econd(X)-X+X -\mathbb{E}_{\mathbb{P}}(X|Y), \Econd(X)-\mathbb{E}_{\mathbb{P}}(X|Y) \rangle_{\mathbb{P}}\\
    =&-\langle X-\Econd(X),\Econd(X)-\mathbb{E}_{\mathbb{P}}(X|Y)\rangle_{\mathbb{P}} + \langle X-\mathbb{E}_{\mathbb{P}}(X|Y),\Econd(X)-\mathbb{E}_{\mathbb{P}}(X|Y)\rangle_{\mathbb{P}}\\
    =&0
\end{align*}
since $\Econd(X)-\mathbb{E}_{\mathbb{P}}(X|Y)\in \CompFuncObs{Y}$. Therefore $\mathbb{E}_{\mathbb{P}}(X|Y)=\Econd(X)$ almost everywhere with respect to $\mathbb{P}$. Let $\omega\in \Omega$, then $Y(\omega)\in \mathrm{Ran}(Y)$ and thus, since $\mathbb{P}\in D_Y$, there exists $\omega'\in \Omega$ such that $\mathbb{P}(\{\omega'\})>0$ and $Y(\omega)=Y(\omega')$. The fact that $\mathbb{E}_{\mathbb{P}}(X|Y)=\Econd(X)$ almost everywhere with respect to $\mathbb{P}$ then implies $\mathbb{E}_{\mathbb{P}}(X|Y)(\omega')=\Econd(X)(\omega')$. Moreover since they are both functions of $Y$ we have $\mathbb{E}_{\mathbb{P}}(X|Y)(\omega')=\mathbb{E}_{\mathbb{P}}(X|Y)(\omega)$ and $\Econd(X)(\omega')=\Econd(X)(\omega)$. In the end we have indeed $\mathbb{E}_{\mathbb{P}}(X|Y)=\Econd(X)$ on $\Omega$. This concludes the proof.
\end{proof}
Let us point out that one can see quite directly that Eq.\eqref{eq:cond1} and Eq.\eqref{eq:cond2} imply that $I-\Econd(I)$ is orthogonal to the algebra $\CompFuncObs{Y}$, as well as belonging to $\CompFuncObs{Y}$. This implies
\begin{equation}
\Econd(I)=I,
\end{equation}
where $I$ is the constant function equal to $1$ on $\Omega$.
Using Eq.\eqref{eq:cond1} again, we obtain that 
\begin{equation}
\Econd(f(Y))=f(Y).
\end{equation}
These identities can also be retrieved from 
$\Econd(X)=\mathbb{E}_{\mathbb{P}}(X|Y).$ 

We finish this section by showing that one can recover the joint probability distribution of the random variables $X$ and $Y$ using the expectation and conditional expectation. We will discuss a similar relation between joint quasiprobabilities and the conditional expectation in the quantum realm in Proposition~\ref{prop: iterated expectation_Y}.

\begin{proposition}\label{Prop:joint_probab_from_cond_exp}
For any $(x,y)\in \mathrm{Ran}(X)\times \mathrm{Ran(Y)}$, we have
 \begin{equation}\label{eq:joint_proba_from_cond_exp}
     \mathbb{P}(X=x,Y=y)=\mathbb{E}_{\mathbb{P}}[\mathds{1}_{y}(Y)\mathbb{E}_{\mathbb{P}}(\mathds{1}_{x}(X)|Y)].
 \end{equation}
 \end{proposition}
 \begin{proof}
 We compute
 \begin{align*}
     \mathbb{E}_{\mathbb{P}}[\mathds{1}_{y}(Y)\mathbb{E}_{\mathbb{P}}(\mathds{1}_{x}(X)|Y)]&=\sum_{y'\in \mathrm{Ran}(Y)}\mathbb{P}(Y=y')\mathds{1}_{y}(y')\mathbb{E}_{\mathbb{P}}(\mathds{1}_{x}(X)|Y)(y')\\
     &=\mathbb{P}(Y=y)\mathbb{E}_{\mathbb{P}}(\mathds{1}_{x}(X)|Y)(y)\\
     &=\mathbb{P}(Y=y)\sum_{x'\in \mathrm{Ran}(X)}\mathds{1}_{x}(x')\frac{\mathbb{P}(X=x',y=y)}{\mathbb{P}(Y=y)}\\
     &=\mathbb{P}(X=x,Y=y).
 \end{align*}
 \end{proof}

We finally point out that $Y$ can also be taken complex vector-valued in what precedes. But, to stress the analogy with the quantum mechanical treatment, we have taken $Y$ real vector-valued here.

\section{Quantum conditional expectation and variance via  minimization}\label{s:Q_cond_exp_min}
 In Section~\ref{s:Q_cond_exp_def}, we show how to adapt the classical definition of a conditional expectation based on Theorem~\ref{thm:classical_minimization} to the quantum context  in order to define a quantum conditional expectation: see Theorem~\ref{thm:Uniqueminimizers}. This procedure defines the quantum conditional expectation as a ``best estimator'' of the operator $\hat X$ by functions of the observable $\hat Y$.  Such definition based on a minimization procedure in analogy with the classical definition clearly depends on the quantity (the cost function) that is chosen to be minimized.  We will, as in the classical context, use a quadratic minimizer and see that its choice is closely linked to the familiar ``ordering problem'' in quantization. 
 Indeed, whereas for random variables $X$ and $Y$, one has $f(Y)X=Xf(Y)$, this is no longer generally true when $\hat{X}$ and $\hat{Y}$ are operators on a Hilbert space.  This observation naturally leads to two conditional expectations, a ``left'' one and a ``right'' one, among many other possibilities. In Section~\ref{s:Q_econd_char},  we provide in Theorem~\ref{thm:UniqueEcond} a characterization of the ``left'' and ``right'' quantum conditional expectations introduced that is analogous to its classical counterpart, Theorem~\ref{thm:econd_characterize2} in that it relies on a ``pull-out property''. 
 In Section \ref{s:Q_cond_exp_mean_var}, we introduce the conditional variance of $\hat X$ given $\hat Y$ and we prove relations analogous to Eqs~\eqref{eq:iterated_ex_classical} and \eqref{eq:variances_compare_class}, relating the mean and variance of the operator $\hat X$ to those of its quantum conditional expectation and variance. We also discuss some of the marked differences between the classical and quantum notions, in particular with regard to pure states. Finally in Section \ref{s:real part}, we show that the self-adjoint part of the quantum conditional expectation can itself be interpreted as a best estimator of $\hat X$ by {\it self-adjoint} functions of $\hat Y$.
 
\subsection{The quantum conditional expectation: definition\label{s:Q_cond_exp_def}}
Let $\mathcal{H}$ be a Hilbert space of dimension $d<+\infty$. Let $\mathcal{L}(\mathcal{H})$ be the space of  operators on  $\mathcal{H}$,  and $D(\mathcal{H})$ the set of density matrices (positive semi-definite operators of trace $1$). We will first adapt the classical construction of the conditional expectation of Theorem~\ref{thm:classical_minimization} to the quantum context. For that purpose, we need to identify a suitable quantum analog of the sesquilinear form between classical observables in Eq.\eqref{eq:L2Omegainnerproduct}. One may expect such analogs not to be uniquely determined since the sesquilinear form involves the product of two classical observables, and in general, quantizations of observables don't respect products: the quantization of the product is usually not the product of the quantizations due to  ordering problems. We will see this phenomenon manifests itself here as well. 
One choice that presents itself naturally is to define, for any two  $\hat{X},\hat{X}'\in \mathcal L(\mathcal H)$, the sesquilinear form
\begin{equation}\label{eq:leftproduct}
    \langle \hat{X},\hat{X}'\rangle_\ell := \langle  \hat{X } , \hat{X }'  \rangle _{\hat{\rho } , \ell } :=  \Tr(\hat{\rho} \hat{X}^\dagger \hat{X}')=\Tr(\hat{X}'\hat{\rho} \hat{X}^\dagger),
\end{equation}
where $\hat \rho\in D(\mathcal{H})$ is a fixed quantum state, which we will mostly drop from the notations (except when we will consider varying states in Section~\ref{s:Q_cond_exp_int} below).
 Now, let $\hat Y=(\hat Y_1,\dots \hat Y_m)$ be a set of commuting observables on $\mathcal H$. We define the joint spectrum of $\hat Y$ as
 \begin{equation*}
     \sigma(\hat Y)=\sigma(\hat Y_1)\times\dots\times \sigma(\hat Y_m),
 \end{equation*}
 where $\sigma(\hat Y_j)$ is the spectrum of  $\hat Y_j$.
 We will say $\hat Y$ is \emph{complete} when, for each $y=(y_1,\dots,y_m)\in\sigma(\hat Y)$, 
there exists a unique (up to a an irrelevant global phase) normalized common eigenvector $\varphi_y^{\hat Y}$, for all $\hat Y_j:$
\begin{equation}
    \hat Y_j\varphi_y^{\hat{Y}}=y_{j}\varphi_y^{\hat{Y}}. 
\end{equation}
In other words, $\hat Y$ is complete if and only if its joint spectrum is non-degenerate. Such a family is called a {\it complete set of commuting observables} (CSCO). While demanding that $\hat Y$ is a family of commuting observables rather than an observable itself may seem superfluous, it is in fact of importance in many cases. A relevant example is found in spin systems, where $\mathcal{H}=\left(\mathbb{C}^2\right)^{\otimes m}$, and where the Pauli operators are given by $\hat Y_j^{\hat \sigma}=I\otimes \dots\otimes I \otimes \hat \sigma \otimes I\otimes \dots \otimes I$. Here $\hat \sigma$ is a $2\times 2$ Pauli matrix, acting on the $j$-th sector.

In what follows we shall write $\hat \Pi_y^{\hat{Y}}$ for the orthogonal projector onto $\varphi_y^{\hat{Y}}$.  We consider the algebra of all functions of $\hat Y$:
\begin{equation}\label{eq:mathcalY}
\CompFuncObs{\hat{Y}}=\{f(\hat Y)|f:\sigma(\hat Y)\to \mathbb C\}.
\end{equation}
Note that $\CompFuncObs{\hat{Y}}$ is a  commutative $*$-algebra. We denote by $\left(\CompFuncObs{\hat{Y}}\right)'$ the commutant of $\CompFuncObs{\hat{Y}}$. $\CompFuncObs{\hat{Y}}$ is  maximal, in the following sense:  
\begin{lemma} 
$\left(\CompFuncObs{\hat{Y}}\right)'=\CompFuncObs{\hat{Y}}.$
\end{lemma}

\begin{proof}
    It is clear that $\CompFuncObs{\hat{Y}} \subseteq \left(\CompFuncObs{\hat{Y}}\right)'$. Suppose therefore $\hat C\in \left(\CompFuncObs{\hat{Y}}\right)'$. Then  $[\hat C,\hat \Pi_y^{\hat{Y}}]=0$, for all $y\in\sigma(\hat Y)$. Hence, for all $y \in\sigma(\hat Y)$, 
\begin{equation}
    \hat C\varphi_y^{\hat{Y}}=\varphi_y^{\hat{Y}} \Tr(\hat C\hat\Pi_y^{\hat{Y}}).
\end{equation}
Defining the function $g(y)=\Tr(\hat C\hat \Pi_y^{\hat{Y}})$, one concludes that $\hat C=g(\hat Y)$ and hence that  $\hat C\in \CompFuncObs{\hat{Y}}$. 
\end{proof}

For a density matrix $\hat{\rho}$, in analogy with Theorem~\ref{thm:classical_minimization}, one can then define (following~\cite{dressel2015, brummelhuis2024,Brummelhuis2025a,Brummelhuis2025b} and generalized in~\cite{tsang2022,tsang2023}) a quantum conditional expectation of $\hat{X}$, knowing $\hat{Y}$, that we shall denote by $\QuantEcond{\ell}{\hat{X}}{\hat{Y}}$ as the minimizer of the following quantity:
\begin{equation}\label{eq:Q_min_prob_left}
    f\mapsto \Tr(\hat{\rho} |\hat{X}-f(\hat{Y})|^2)=\langle \hat{X}-f(\hat{Y}), \hat{X}-f(\hat{Y})\rangle_{\ell},
\end{equation}
over all $f:\sigma(\hat{Y})\to\mathbb{C}$. Here, for $Z\in \mathcal L(\mathcal H)$, we write $|Z|=\sqrt{Z^\dagger Z}$. As in the classical case, we think of the right hand side of this equation as a mean squared error cost. We shall discuss the existence and uniqueness of this minimizer below. We shall refer to $\QuantEcond{\ell}{\hat{X}}{\hat{Y}}$ as the left quantum conditional expectation of $\hat{X}$ knowing $\hat{Y}$ in the state $\hat{\rho}$, for reasons that will become clear shortly. There is a second, equally natural choice for the sesquilinear product, analogous to the inner product in Eq.\eqref{eq:L2Omegainnerproduct}, namely
\begin{equation}\label{eq:rightproduct}
    \langle \hat{X},\hat{X}'\rangle_{\mathrm r}=\Tr(\hat{\rho} \hat{X}' \hat{X}^\dagger )=\Tr(\hat{X}^\dagger\hat{\rho} \hat{X}').
\end{equation}
Note that, as anticipated, the difference between Eq.\eqref{eq:leftproduct} and Eq.\eqref{eq:rightproduct} lies indeed in the change of order of appearance of $\hat{X}^\dagger$ and $\hat{X}'$.  We then define the right quantum conditional expectation of $\hat{X}$, knowing $\hat{Y}$, that we shall denote by $\QuantEcond{\mathrm{r}}{\hat{X}}{\hat{Y}}$, as the minimizer over complex-valued functions of
\begin{equation}\label{eq:Q_min_prob_right}
    f\mapsto \langle \hat{X}-f(\hat{Y}), \hat{X}-f(\hat{Y})\rangle_{\mathrm r}.
\end{equation}

The existence and uniqueness of these conditional expectations is guaranteed by the following theorem, under a suitable hypothesis on $\hat{\rho}$.
We define 
 \begin{equation}\label{eq:DYdef}
        D_{\hat{Y}}=\{\hat{\rho}\in D(\mathcal{H}): \forall y\in \sigma(\hat{Y}), \ \Tr(\hat{\rho} \hat{\Pi}_y^{\hat{Y}})\neq 0\: \}.
    \end{equation}
    
\begin{Theorem}\label{thm:Uniqueminimizers}
    Assume that $\hat{\rho}\in D_{\hat{Y}}$. Then the functions given in Eq.\eqref{eq:Q_min_prob_left} and Eq.\eqref{eq:Q_min_prob_right} admit  unique minimizers given by
    \begin{equation}\label{eq:MinCondExp}
        f_0^{\ell}(y)=\frac{\langle \varphi_y^{\hat{Y}}, \hat{X}\hat{\rho}\varphi_y^{\hat{Y}}\rangle}{\langle \varphi_y^{\hat{Y}},\hat{\rho}\varphi_y^{\hat{Y}}\rangle},\quad f_0^{\mathrm{r}}(y)=\frac{\langle \varphi_y^{\hat{Y}}, \hat{\rho} \hat{X}\varphi_y^{\hat{Y}}\rangle}{\langle \varphi_y^{\hat{Y}},\hat{\rho}\varphi_y^{\hat{Y}}\rangle},
    \end{equation}
    respectively, for all $y\in \sigma(\hat{Y})$.
\end{Theorem}
   
\begin{proof}
We only prove the ``left'' case, the ``right'' case being similar. We compute:
\begin{align}
    \langle \hat{X}-f(\hat{Y}),\hat{X}-f(\hat{Y})\rangle_{\ell}&=\Tr(\hat{\rho} |\hat{X}|^2) - \Tr(\hat{\rho} \hat{X}^{\dagger}f(\hat{Y})) - \Tr(\hat{\rho} f(\hat{Y})^{\dagger}\hat{X}) + \Tr(\hat{\rho} |f(\hat{Y})|^2)\nonumber\\
    &=\Tr(\hat{\rho} |\hat{X}|^2) + \sum_{y\in \sigma(\hat{Y})}\langle \varphi_y^{\hat{Y}},\hat{\rho}\varphi_y^{\hat{Y}}\rangle\left(|f(y)|^2 - 2 \Re\left(\overline{f(y)}\frac{\langle \varphi_y^{\hat{Y}},\hat{X}\hat{\rho}\varphi_y^{\hat{Y}}\rangle}{\langle \varphi_y^{\hat{Y}},\hat{\rho}\varphi_y^{\hat{Y}}\rangle}\right)\right)\nonumber\\
    &=\Tr(\hat{\rho} |\hat{X}|^2) - \sum_{y\in \sigma(\hat{Y})}\frac{|\langle \varphi_y^{\hat{Y}},\hat{X}\hat{\rho}\varphi_y^{\hat{Y}}\rangle|^2}{\langle \varphi_y^{\hat{Y}},\hat{\rho}\varphi_y^{\hat{Y}}\rangle}\nonumber\\
    &\qquad+ \sum_{y\in \sigma(\hat{Y})}\langle \varphi_y^{\hat{Y}},\hat{\rho}\varphi_y^{\hat{Y}}\rangle\left|f(y)-\frac{\langle \varphi_y^{\hat{Y}},\hat{X}\hat{\rho}\varphi_y^{\hat{Y}}\rangle}{\langle \varphi_y^{\hat{Y}},\hat{\rho}\varphi_y^{\hat{Y}}\rangle}\right|^2 .\label{eq:quadform}
\end{align}
Hence this quantity is minimal if and only if
\begin{equation}
        f(y)=\frac{\langle \varphi_y^{\hat{Y}}, \hat{X}\hat{\rho}\varphi_y^{\hat{Y}}\rangle}{\langle \varphi_y^{\hat{Y}},\hat{\rho}\varphi_y^{\hat{Y}}\rangle}
    \end{equation}
    for all $y\in \sigma(\hat{Y})$.
\end{proof}   

\begin{definition}
\label{def: left/right Econd}
We define the left quantum conditional expectation of the operator $\hat{X}$ knowing $\hat{Y}$ in the state $\hat{\rho}\in D_{\hat{Y}}$ as
\begin{equation}\label{eq:left_cond_exp}
    \mathbb{E}^{\ell}_{\hat{\rho}}(\hat{X}|\hat{Y})=f_0^{\ell}(\hat{Y})=\sum_{y\in \sigma(\hat{Y})}\frac{\langle \varphi_y^{\hat{Y}},\hat{X}\hat{\rho}\varphi_y^{\hat{Y}}\rangle}{\langle \varphi_y^{\hat{Y}},\hat{\rho}\varphi_y^{\hat{Y}}\rangle}\hat{\Pi}_y^{\hat{Y}},
\end{equation}
where $\hat{\Pi}_y$ is the orthogonal projector to the eigenvector $\varphi_y^{\hat Y}$ for $y\in \sigma(\hat{Y})$. Similarly we define the right quantum conditional expectation of $\hat{X}$ knowing $\hat{Y}$ in the state $\hat{\rho}$ as
\begin{equation}\label{eq:right_cond_exp}
    \mathbb{E}^{\mathrm{r}}_{\hat{\rho}}(\hat{X}|\hat{Y})=f_0^{\mathrm{r}}(\hat{Y})=\sum_{y\in \sigma(\hat{Y})}\frac{\langle \varphi_y^{\hat Y},\hat{\rho}\hat{X}\varphi_y^{\hat Y}\rangle}{\langle \varphi_y^{\hat Y},\hat{\rho}\varphi_y^{\hat Y}\rangle}\hat{\Pi}_y^{\hat Y}.
\end{equation}
We will write $\mathbb{E}^{\ell/\mathrm{r}}_{\hat{\rho}}(\hat{X}|\hat{Y}=y)=\Tr\left(\mathbb{E}^{\ell/\mathrm{r}}_{\hat{\rho}}(\hat{X}|\hat{Y})\hat \Pi^{\hat Y}_y\right)$.
\end{definition}
It is readily checked that
for any $f:\sigma(\hat{Y})\to \mathbb{C}$,
\begin{equation}\label{eq:Q_f(Y)_ind_Y}
    \mathbb{E}^{\ell/\mathrm{r}}_{\hat{\rho}}(f(\hat{Y})|\hat{Y})=f(\hat{Y}).
\end{equation}

\begin{rmk}\label{rmk:QCE}

$(i)$ If the state $\hat{\rho}$ is not an element of $D_{\hat{Y}}$, i.e. if there exists $y\in \sigma(\hat{Y})$ such that $\langle \varphi_y^{\hat{Y}},\hat{\rho}\varphi_y^{\hat{Y}}\rangle=0$, then the functional \eqref{eq:Q_min_prob_left} still admits minimizers, but they are no longer unique. Since $\langle \varphi_y^{\hat{Y}},\hat{\rho}\varphi_y^{\hat{Y}}\rangle=0 $ implies that $\hat{\rho }^{1/2 } \varphi_y^{\hat{Y } } = 0 $ and therefore $\hat{\rho } \varphi_y^{\hat{Y} } = 0 $, the sums in (\ref{eq:quadform}) only extend over $y $'s with $\langle\varphi_y^{\hat{Y}},\hat{\rho}\varphi_y^{\hat{Y}}\rangle \neq 0$ and the minimizers are in fact given by
   
\begin{equation} \label{eq:QCE_degenerate}
f_0 (\hat{Y } ) \ = \sum_{\substack{y\in \sigma(\hat{Y})\\ \langle\varphi_y^{\hat{Y}},\hat{\rho}\varphi_y^{\hat{Y}}\rangle \neq 0}}\frac{\langle \varphi_y^{\hat{Y}},\hat{X}\hat{\rho}\varphi_y^{\hat{Y}}\rangle}{\langle \varphi_y^{\hat{Y}},\hat{\rho}\varphi_y^{\hat{Y}}\rangle}\hat{\Pi}_y^{\hat{Y}} + \sum_{\substack{y\in \sigma(\hat{Y})\\ \langle \varphi_y^{\hat{Y}},\hat{\rho}\varphi_y^{\hat{Y}}\rangle=0}}g_y\hat{\Pi}_y^{\hat{Y}},   
\end{equation}   
where $(g_y)_{\langle \varphi_y^{\hat{Y}},\hat{\rho}\varphi_y^{\hat{Y}}\rangle=0}\subset \mathbb{C}$ can be any complex numbers \cite{brummelhuis2024}. In the usual interpretation of quantum mechanics, the quantity $\langle \varphi_y^{\hat{Y}},\hat{\rho}\varphi_y^{\hat{Y}}\rangle$ represents the probability of the event ``$\hat{Y}$ takes the value $y$''.  It follows in accordance with classical probability theory that the minimizers of the functional \eqref{eq:Q_min_prob_left} are unique up to functions of $\hat{Y}$ which are ``$\hat{\rho}$-almost everywhere $0$'', such functions corresponding to the term $\sum_{\substack{y\in \sigma(\hat{Y})\\ \langle \varphi_y^{\hat{Y}},\hat{\rho}\varphi_y^{\hat{Y}}\rangle=0}}g_y\hat{\Pi}_y^{\hat{Y}}$. In this article we only consider non-degenerate states, except for in Appendix~\ref{app:Fisher} where we discuss the connection with quantum Fisher information. Note that here by non-degenerate we mean states $\hat{\rho}\in D_{\hat{Y}}$, which does not necessarily imply that the state $\hat{\rho}$ is positive definite. For example the set $D_{\hat{Y}}$ contains pure states. This non-degeneracy condition allows nice characterizations of the conditional expectations. Moreover it is not too restrictive, as shown by Lemma~\ref{lemma: D_Y dense}.
    
$(ii)$ The expressions $\langle\cdot,\cdot\rangle_{\ell/\textrm r}$ are not inner products because $\langle \hat X,\hat X\rangle_{\ell/\textrm{r}}$ may vanish even if $\hat X\not=0$. This will happen if $\hat \rho$ has a non-vanishing kernel, foremost when the state $\hat{\rho } $ is  pure. One may however think of them as degenerate inner products or as quadratic cost functions. Note that this cost depends on the state $\hat \rho$. From this point of view, one may think of the operator $\mathbb E^{\ell/\textrm{r}}_{\hat\rho}(\hat X|\hat Y)$ as the operator function of the observables $\hat Y$ that approximates $\hat X$ at minimal cost. It is also clear that such minimizers necessarily depend on the choice of cost function, that in turn necessarily depends on the physics of the problem considered. In the next subsection, we show the ``left'' and ``right'' choices made here are natural in the sense that they can be uniquely characterized by a pull-out property analogous to the classical conditional expectation.
    
$(iii)$ In view of what precedes, the quantum conditional expectations $\mathbb E^{\ell, \mathrm r}$ are orthogonal projections with respect to the two ``inner products'' induced by the state $\hat{\rho}$ of the operator $\hat{X}$ onto the algebra of functions of the operators $\hat{Y}$. 

    $(iv)$  The conditional expectations  $\mathbb{E}^{\ell/\mathrm{r}}_{\hat{\rho}}(\hat{X}|\hat{Y})$ are not self-adjoint in general, even if we take $\hat{X}$ to be self-adjoint. This is markedly different from the classical situation, where the conditional expectation of a real random variable $X$, given a real random variable $Y$, is always real.  We will extensively elaborate on this point below.
    Note that 
\begin{equation}
\label{eq: link Eleft Eright}
\mathbb{E}^{\mathrm{r}}_{\hat{\rho}}(\hat{X}|\hat{Y})= \mathbb{E}^{\mathrm{\ell}}_{\hat{\rho}}(\hat{X}^\dagger|\hat{Y})^\dagger.
\end{equation}

    $(v)$  When $\hat X$ is self-adjoint, the coefficients 
\begin{equation}\label{eq:weak_values_def}
\left(\frac{\langle \varphi_y,\hat{X}\hat{\rho}\varphi_y\rangle}{\langle \varphi_y,\hat{\rho}\varphi_y\rangle}\right)_{y\in\sigma(\hat{Y})}
\end{equation}
appearing in $\QuantEcond{\ell}{\hat{X}}{\hat{Y}}$ are weak values with as the initial state $\hat \rho$, a weak measurement of the observable $\hat{X}$ and a post-selection $\hat{\Pi}^{\hat Y}_{y}$ for each $y\in\sigma(\hat{Y})$.  This observation links the conditional expectation of $\hat X$ to the physics of weak values, which has been studied extensively, see~\cite{arvidssonshukur2024properties,dresselContextualvalueApproachGeneralized2012, Williams_Jordan_2008,Jozsa2007,Dressel_Jordan_2012,Dressel_Jordan_2012_ImaginaryPartWeakValue}.  For the readers' convenience, we provide, in Appendix~\ref{app:weakvalues},  a concise, complete, self-contained and mathematically clean introduction to weak value measurement, which is the procedure by which the above weak values can be measured. This procedure relies in particular on a postselection, which has historically provided the intuitive link between weak value measurement and conditioning. We refer to the appendix for details.  The definition of quantum conditional expectation through minimization that we propose here allows to sharpen this link as we shall further discuss is Sections~\ref{s:Q_cond_exp_mean_var} and~\ref{s:Q_cond_exp_int}.      
\end{rmk}

\subsection{Characterization of the left/right quantum conditional expectations}\label{s:Q_econd_char}

Let us point out that there are many other quantum analogs of the inner product in Eq.\eqref{eq:L2classical}, among which one could for example consider
\begin{equation}
    \langle \hat X, \hat X'\rangle_\alpha=\alpha\langle\hat X, \hat X'\rangle_\ell +(1-\alpha)\langle \hat{X}, \hat{X'}\rangle_{\mathrm r}.
\end{equation}
for $0\leq \alpha\leq 1$, interpolating between $\langle \hat{X}, \hat{X'}\rangle_\ell$ and $\langle \hat{X}, \hat{X'}\rangle_{\mathrm r}$.  Minimizing $\langle \hat X-f(\hat Y),\hat X-f(\hat Y)\rangle_{\alpha}$ leads to a candidate notion of conditional expectation, whose explicit expression can be easily computed to be
\begin{equation}
    \mathbb{E}^{\alpha}_{\hat \rho}(\hat X|\hat Y)=\sum_{y\in \sigma(\hat Y)}\frac{\langle \varphi_y^{\hat Y}, (\alpha \hat X \hat \rho + (1-\alpha)\hat \rho \hat X)\varphi_y^{\hat Y}\rangle}{\langle \varphi_y^{\hat Y},\hat \rho \varphi_y^{\hat Y}\rangle}\hat{\Pi}_y^{\hat Y}.
\end{equation}
Other, more elaborate choices have been considered in~\cite{tsang2023}.

The following theorem shows that to obtain a conditional expectation that satisfies the properties of the classical conditional expectation given in Theorem~\ref{thm:econd_characterize2}, $\alpha$ has to be either $0$ or $1$.

\begin{Theorem}\label{thm:UniqueEcond}
Let $\hat{\rho} \in D_{\hat{Y}}$. For any $\#\in \{\ell, \mathrm{r}\}$,
there exist a unique map $\Econd^{\#}: \mathcal{L}(\mathcal{H})\to \CompFuncObs{\hat{Y}}$ satisfying the following properties:
\begin{enumerate}
    \item Pull-out formula. For any $\hat{X}\in \mathcal{L}(\mathcal{H})$ and any $f:\sigma(\hat{Y})\to \mathbb{C}$,
        \begin{align}
        &\text{for }\#=\ell,\:\:\Econd^{\ell}(f(\hat{Y})\hat{X})=f(\hat{Y})\Econd^{\ell}(\hat{X})=\Econd^{\ell}(\hat{X})f(\hat{Y}),
        \label{eq:left invariance}\\
        &\text{for }\#=\mathrm{r},\:\:\Econd^{\mathrm{r}}(\hat{X}f(\hat{Y}))=f(\hat{Y})\Econd^{\mathrm{r}}(\hat{X})=\Econd^{\mathrm{r}}(\hat{X})f(\hat{Y}).
        \label{eq:right invariance}
    \end{align}
    \item Quantum law of iterated expectations. For any $\hat{X}\in \mathcal{L}(\mathcal{H})$,
    \begin{equation}\label{eq:iterated expectation}
        \mathbb{E}_{\hat{\rho}}(\Econd^{\#}(\hat{X}))=\mathbb{E}_{\hat{\rho}}(\hat{X}).
    \end{equation}
\end{enumerate}
This map is given by $\Econd^{\#}(\hat{X})=\mathbb{E}^{\#}_{\hat{\rho}}(\hat{X}|\hat{Y})$.
\end{Theorem}
This theorem singles out two natural definitions of quantum conditional expectation through the conditions Eq.\eqref{eq:left invariance} and Eq.\eqref{eq:iterated expectation} or Eq.\eqref{eq:right invariance} and Eq.\eqref{eq:iterated expectation}. 

As in the classical case, the uniqueness in Theorem~\ref{thm:UniqueEcond} proves that Eq.\eqref{eq:Q_f(Y)_ind_Y} can be deduced from Eq.\eqref{eq:left invariance}-\eqref{eq:right invariance} and Eq.\eqref{eq:iterated expectation}. Again, one could have shown this fact directly.

\begin{proof}
Again, we write the proof for $\#=\ell$, the other case being similar. We begin by showing that the left quantum conditional expectation satisfies these two properties. For point 1, we compute
\begin{align*}
    \mathbb{E}^{\ell}_{\hat{\rho}}(f(\hat{Y})\hat{X}|\hat{Y})&=\sum_{y\in \sigma(\hat{Y})}\frac{\langle \varphi_y^{\hat{Y}}, f(\hat{Y})\hat{X}\hat{\rho} \varphi_y^{\hat{Y}}\rangle}{\langle \varphi_y^{\hat{Y}}, \hat{\rho}\varphi_y^{\hat{Y}}\rangle}\hat{\Pi}_y\\
    &=\sum_{y\in \sigma(\hat{Y})}f(y)\frac{\langle \varphi_y^{\hat{Y}}, \hat{X}\hat{\rho} \varphi_y^{\hat{Y}}\rangle}{\langle \varphi_y^{\hat{Y}},\hat{\rho} \varphi_y^{\hat{Y}}\rangle}\hat{\Pi}_y\\
    &=f(\hat{Y})\mathbb{E}^{\ell}_{\hat{\rho}}(\hat{X}|\hat{Y})=\mathbb{E}^{\ell}_{\hat{\rho}}(\hat{X}|\hat{Y})f(\hat{Y}).
\end{align*}
For point 2, we compute
\begin{align*}
    \mathbb{E}_{\hat{\rho}}(\mathbb{E}^{\ell}_{\hat{\rho}}(\hat{X}|\hat{Y}))&=\Tr\left(\hat{\rho}\sum_{y\in \sigma(\hat{Y})}\frac{\langle \varphi_y^{\hat{Y}}, \hat{X}\hat{\rho} \varphi_y^{\hat{Y}}\rangle}{\langle \varphi_y^{\hat{Y}},\hat{\rho}\varphi_y^{\hat{Y}}\rangle}\hat{\Pi}_y^{\hat{Y}}\right)\\
    &=\sum_{y\in \sigma(\hat{Y})}\langle \varphi_y^{\hat{Y}},\hat{X}\hat{\rho}\varphi_y^{\hat{Y}}\rangle\\
    &=\Tr(\hat{X}\hat{\rho})\\
    &=\mathbb{E}_{\hat{\rho}}(\hat{X}).
\end{align*}
Now, let $\Econd^{\ell}: \mathcal{L}(\mathcal{H})\to \CompFuncObs{\hat{Y}}$ be a map that satisfies Eq.~\eqref{eq:left invariance} and Eq.~\eqref{eq:iterated expectation}. We will show that $\Econd^{\ell}(\hat{X})=\mathbb{E}^{\ell}_{\hat{\rho}}(\hat{X}|\hat{Y})$ for any $\hat{X}\in \mathcal{L}(\mathcal{H})$. 
Let $\hat{X}\in \mathcal{L}(\mathcal{H})$. For any $f(\hat{Y})\in \CompFuncObs{\hat{Y}}$, we have
\begin{align*}
    \langle f(\hat{Y}), \hat{X}-\Econd^{\ell}(\hat{X})\rangle_{\ell} &=\Tr(\hat{\rho} f(\hat{Y})^{\dagger}(\hat{X}-\Econd^{\ell}(\hat{X})))\\
    &=\Tr(\hat{\rho} f(\hat{Y})^{\dagger} \hat{X})- \Tr(\hat{\rho} \Econd^{\ell}(f(\hat{Y})^{\dagger}\hat{X})) \text{ by Eq.\eqref{eq:left invariance}}\\
    &=\Tr(\hat{\rho} f(\hat{Y})^{\dagger}\hat{X})-\Tr(\hat{\rho} f(\hat{Y})^{\dagger}\hat{X}) \text{ by Eq.\eqref{eq:iterated expectation}}\\
    &=0,
\end{align*}
and of course $\langle f(\hat{Y}), \hat{X}-\mathbb{E}_{\hat{\rho}}(\hat{X}|\hat{Y})\rangle_{\ell}=0$ as well. It follows
\begin{align*}
    &\langle \Econd^{\ell}(\hat{X}) - \mathbb{E}^{\ell}_{\hat{\rho}}(\hat{X}|\hat{Y}),\Econd^{\ell}(\hat{X}) - \mathbb{E}^{\ell}_{\hat{\rho}}(\hat{X}|\hat{Y})\rangle_{\ell}\\ 
    =&\langle \Econd^{\ell}(\hat{X}) -\hat{X} + \hat{X} - \mathbb{E}^{\ell}_{\hat{\rho}}(\hat{X}|\hat{Y}),\Econd^{\ell}(\hat{X}) - \mathbb{E}^{\ell}_{\hat{\rho}}(\hat{X}|\hat{Y})\rangle_{\ell}\\
    =&- \langle \hat{X}-\Econd^{\ell}(\hat{X}),  \Econd^{\ell}(\hat{X}) - \mathbb{E}^{\ell}_{\hat{\rho}}(\hat{X}|\hat{Y})\rangle_{\ell}  +  \langle \hat{X}-\mathbb{E}^{\ell}_{\hat{\rho}}(\hat{X}|\hat{Y}),  \Econd^{\ell}(\hat{X}) - \mathbb{E}^{\ell}_{\hat{\rho}}(\hat{X}|\hat{Y})\rangle_{\ell}\\
    =&0
\end{align*}
since $ \Econd^{\ell}(\hat{X}) - \mathbb{E}^{\ell}_{\hat{\rho}}(\hat{X}|\hat{Y})\in \CompFuncObs{\hat{Y}}$. Therefore, $\Tr(\hat{\rho}|\Econd^{\ell}(\hat{X}) - \mathbb{E}^{\ell}_{\hat{\rho}}(\hat{X}|\hat{Y})|^2)=0$. Since $\Econd^{\ell}(\hat{X}) - \mathbb{E}^{\ell}_{\hat{\rho}}(\hat{X}|\hat{Y})\in \CompFuncObs{\hat{Y}}$, there exists $C_{\hat{\rho}}:\sigma(\hat{Y})\to \mathbb{C}$ such that $\Econd^{\ell}(\hat{X}) - \mathbb{E}^{\ell}_{\hat{\rho}}(\hat{X}|\hat{Y})=C_{\hat{\rho}}(\hat{Y})$. One gets
\begin{align*}
    0=\Tr(\hat{\rho} |C_{\hat{\rho}}(\hat{Y})|^2)
    =\sum_{y\in \sigma(\hat{Y})}|C_{\hat{\rho}}(y)|^2\langle \varphi_y^{\hat{Y}}, \hat{\rho} \varphi_y^{\hat{Y}}\rangle,
\end{align*}
which implies, since $\hat{\rho} \in D_{\hat{Y}}$, that $C_{\hat{\rho}}=0$, \textit{i.e.} $\Econd^{\ell}(\hat{X})=\mathbb{E}^{\ell}_{\hat{\rho}}(\hat{X}|\hat{Y})$.
\end{proof}

It is legitimate to ask if one can define a conditional expectation which satisfies both the left- and right-pullout properties \eqref{eq:left invariance} and \eqref{eq:right invariance}. This has been done in the context of von Neumann algebras in~\cite{umegaki1954}. In Appendix~\ref{app:vonNeumann}, we briefly explain the link between that approach and ours.

The conditional expectation of $\hat{X}$ knowing $\hat{Y}$ is defined only for states $\hat{\rho} \in D_{\hat{Y}}$, to avoid dividing by $0$. As we now prove in the following Lemma, this condition on the state $\hat{\rho}$ is not too restrictive, as the set $D_{\hat{Y}}$ is, in a topological sense, a large set (it is dense in the set of all density matrices).
\begin{lemma}
\label{lemma: D_Y dense}
    Let $\hat{Y}$ be a self-adjoint operator on $\mathcal{H}$. Then $D_{\hat Y}$, defined in Eq.~\eqref{eq:DYdef} 
    is dense in $D(\mathcal{H})$ (for any norm topology).
\end{lemma}

\begin{proof}
    Since we are working in finite dimension, all norms are equivalent. We chose to work with the trace-class norm. Let $\hat{\rho}_0\in D(\mathcal{H})$. If $\hat{\rho}_0\in D_{\hat{Y}}$ there is nothing to do. Otherwise, let 
    \begin{equation*}
        I_0=\{y\in \sigma(\hat{Y}): \Tr(\hat{\rho}_{0} \hat{\Pi}_y^{\hat{Y}})=0\}.
    \end{equation*}
    We denote by $\left|I\right|$ the cardinal of the finite set $I$. For $\varepsilon >0$, let 
    \begin{equation*}
        \hat{\rho}_{\varepsilon}=\frac{1}{1+\varepsilon |I_0|}\left(\hat{\rho}_0 + \varepsilon \sum_{y\in I_0}\hat{\Pi}_y^{\hat{Y}}\right)
    \end{equation*}
    Then it is easily checked that $\hat{\rho}_{\varepsilon}\in D_{\hat{Y}}$ for any $\varepsilon >0$ and moreover
    \begin{align*}
        \Tr(|\hat{\rho}_{\varepsilon}-\hat{\rho}_0|)&=\Tr\left( \left|\hat{\rho}_0\left(\frac{1}{1+\varepsilon|I_0|}-1\right) + \frac{\varepsilon}{1+\varepsilon|I_0|}\sum_{y\in I_0}\hat{\Pi}_y^{\hat{Y}}\right|\right)\\
        &\leq \left|1-\frac{1}{1+\varepsilon |I_0|}\right| + \frac{\varepsilon |I_0|}{1+\varepsilon |I_0|}\\
        &\xrightarrow[\varepsilon \to 0]{}0.
    \end{align*}
    This proves the lemma.
\end{proof}

\begin{rmk}
    $(i)$ If $\hat{Y}$ is an incomplete set of commuting observables, the associated projectors $(\hat{\Pi}_y^{\hat{Y}})_{y\in \sigma(\hat{Y})}$ are not all of rank $1$. However, the results in this section still adapt to this case. For $\hat{\rho}\in D_{\hat{Y}}$ and $\hat X\in \mathcal{L}(\mathcal{H})$, it is easily shown that the minimizers of Eq.\eqref{eq:MinCondExp} become
    \begin{equation}
    \forall y\in\sigma(\hat{Y}), f_{0}^{\ell}(y) = \frac{\mathrm{Tr}(\hat{X}\hat{\rho}\hat{\Pi}_{y}^{\hat{Y}})}{\mathrm{Tr}(\hat{\rho}\hat{\Pi}_{y}^{\hat{Y}})}, \: \: f_{0}^{\mathrm{r}}(y)=\frac{\Tr(\hat{\rho}\hat{X}\hat{\Pi}_y^{\hat{Y}})}{\Tr(\rho\hat{\Pi}_y^{\hat{Y}})},
\end{equation}
and that one can define the left and right conditional expectations by
\begin{equation}
    \mathbb{E}^{\ell}_{\hat{\rho}}(\hat{X}|\hat{Y}) = \sum_{y\in\sigma(\hat{Y})} \frac{\mathrm{Tr}(\hat{X}\hat{\rho}\hat{\Pi}_{y})}{\mathrm{Tr}(\hat{\rho}\hat{\Pi}_{y})} \hat{\Pi}_{y}, \:\: \mathbb{E}^{\mathrm{r}}_{\hat{\rho}}(\hat{X}|\hat{Y}) = \sum_{y\in\sigma(\hat{Y})} \frac{\mathrm{Tr}(\hat{\rho}\hat{X}\hat{\Pi}_{y})}{\mathrm{Tr}(\hat{\rho}\hat{\Pi}_{y})} \hat{\Pi}_{y}.
\end{equation}
The characterization given in Theorem \ref{thm:UniqueEcond} remains true in this setting as do the equivalences of Proposition  \ref{prop:leftequalsright} but not Eq.~\eqref{eq:left_right_commute}. In the following sections, we require the completeness of $\hat{Y}$.

 $(ii)$ We finally note that in the case of commuting operators $\hat{X } $ and $\hat{Y } $, the quantum conditional expectation reduces to the classical conditional expectation (\ref{def:condexpclass1}), when considering the two operators as classical random variables with joint probability $\mathbb{P } (\hat{X } = x , \hat{Y } = y ) = {\rm Tr } (\hat{\Pi }_x \hat{\Pi }_y \hat{\rho } ) $, the Born measure of the commuting system $(\hat{X } , \hat{Y } )$.
\end{rmk}

\subsection{The quantum laws of iterated expectations and of total variance}\label{s:Q_cond_exp_mean_var}

The first two moments of the quantum conditional expectations $\mathbb E_{\hat \rho}^{\ell/\mathrm r}(\hat X|\hat Y)$ have properties that are reminiscent of those of the classical conditional expectation: the law of iterated expectations and the law of total variance. There are however some important differences between the classical and quantum versions that we shall point out. The results are summed up in Proposition~\ref{prop:variances_condexp1} and in the discussion following it. 

We introduce, for any operator $\hat C\in\mathcal{L}(\mathcal{H})$, the left and right variance of $\hat{C } $ in the state $\hat \rho$ by:
\begin{equation}
\Delta_{\hat\rho}^{\ell,2}(\hat C)=\Tr\left(\hat \rho(\hat C-\mathbb{E}_{\hat \rho}( \hat C))^\dagger(\hat C-\mathbb{E}_{\hat \rho}( \hat C))\right),
\Delta_{\hat\rho}^{\mathrm{r},2}(\hat C)=\Tr\left(\hat \rho(\hat C-\mathbb{E}_{\hat \rho}( \hat C))(\hat C-\mathbb{E}_{\hat \rho}( \hat C))^\dagger\right).
\end{equation} 
When $\hat C^\dagger=\hat C$ or $\hat C^\dagger=-\hat C$, we will write
\begin{equation}
    \Delta^2_{\hat\rho} \hat C=\Delta_{\hat\rho}^{\ell,2}(\hat C)=\Delta_{\hat\rho}^{\mathrm r,2}(\hat C).
\end{equation}
 We further introduce the left and right conditional variances of $\hat C$, in analogy with the classical definition:  
\begin{equation}\label{eq:cond_var}
    \Delta_{\hat\rho}^{\ell,2 }(\hat C|\hat Y)=\mathbb E_{\hat\rho}^{\ell}(|\hat C-\mathbb E_{\hat \rho}^\ell(\hat C|\hat Y)|^2|\hat Y),\quad \Delta_{\hat\rho}^{\mathrm r,2 }(\hat C|\hat Y)=\mathbb E_{\hat\rho}^{\mathrm r}\left((\hat C-\mathbb E_{\hat \rho}^{\mathrm r}(\hat C|\hat Y))(\hat C-\mathbb E_{\hat \rho}^{\mathrm r}(\hat C|\hat Y))^\dagger|\hat Y\right).
\end{equation} 
Recall that, for any operator $Z$, one has $|\hat Z|^2=\hat Z^\dagger \hat Z$.
 In addition, a simple computation allows to establish  that the expected value of the conditional variance equals the error:
\begin{equation}\label{eq:MeanCondVar}
    \mathbb E_{\hat\rho}(\Delta^{\ell/r,2}_{\hat \rho} (\hat C|\hat Y))=\langle \hat C-\mathbb E_{\hat \rho}^{\ell/\mathrm{r}}(\hat C|\hat Y),\hat C-\mathbb E_{\hat \rho}^{\ell/\mathrm{r}}(\hat C|\hat Y)\rangle_{\ell/\mathrm{r}}.
\end{equation}

The notations are cumbersome, but the treatment is similar for the ``left'' and ``right'' cases. Recall that, even if $\hat X$ is self-adjoint, its conditional expectation may not be. The left/right conditional variances will then differ, generally. 
The following proposition sums up the basic properties of the conditional expectation and the conditional variance.

\begin{proposition}
\label{prop:variances_condexp1}
    If $\hat\rho\in D_{\hat Y}$ then, for any 
    $\hat X\in \mathcal{L}(\mathcal{H})$ the quantum law of iterated expectations holds:
\begin{equation}\label{eq:ConsMean}
    \mathbb E_{\hat \rho}(\mathbb{E}_{\hat \rho}^{\ell/\mathrm{r}}(\hat{X}|\hat{Y})) = \mathbb E_{\hat \rho}(\hat{X}).
\end{equation}
The quantum law of total variance reads:
\begin{eqnarray}\label{eq:Variance}
\Delta_{\hat \rho}^{\ell/\mathrm{r},2}\hat X &=& \Delta_{\hat \rho}^{\ell/\mathrm{r},2}\mathbb E_{\hat \rho}^{\ell/\mathrm{r}}(\hat X|\hat Y)+ \langle \hat X-\mathbb E_{\hat \rho}^{\ell/\mathrm{r}}(\hat X|\hat Y),\hat X-\mathbb E_{\hat \rho}^{\ell/\mathrm{r}}(\hat X|\hat Y)\rangle_{\ell/\mathrm{r}}\\
 &=& \Delta_{\hat \rho}^{\ell/\mathrm{r},2}\mathbb E_{\hat \rho}^{\ell/\mathrm{r}}(\hat X|\hat Y)+\mathbb E_{\hat\rho}(\Delta^{\ell/r,2}_{\hat \rho} (\hat X|\hat Y)).
\end{eqnarray}
If $\hat\rho\in D_{\hat Y}$ is pure,
\begin{equation}\label{eq:CondVarpure}   
\Delta_{\hat\rho}^{\ell,2 }(\hat X |\hat Y) = 0  
\end{equation}
and thus,
\begin{equation}\label{eq:Variance_pure}
    \Delta_{\hat \rho}^{\ell/\mathrm{r},2}\hat X = \Delta_{\hat \rho}^{\ell/\mathrm{r},2}\mathbb E_{\hat \rho}^{\ell/\mathrm{r}}(\hat X|\hat Y).
\end{equation}
\end{proposition}   

\medskip   
   
The quantum law of iterated expectations, Eq.~\eqref{eq:ConsMean}, says that the best estimator of $\hat X$ by a function of $\hat Y$ has the same expected value as $\hat X$ itself: it is, in this sense, unbiased, as the classical conditional expectation.  Eq.~\eqref{eq:Variance} is in turn a quantum analog of the classical law of total variance Eq.~\eqref{eq:variances_compare_class}.  It provides a relation  between the variance of $\hat X$, the variance of the conditional expectation $\mathbb{E}_{\hat \rho}^{\ell/\mathrm{r}}(\hat{X}|\hat{Y})$, and the expected value of the conditional variance $\Delta^{\ell/\mathrm r,2}(X|Y)$, in complete analogy with the the classical situation. The error term in this relation (the second term on the right hand side of Eq.~\eqref{eq:Variance}) can, as in the classical case, be expressed as the expected value of the conditional variance: this is the content of Eq.~\eqref{eq:MeanCondVar}. 

\begin{proof}
 We give the proof for the left conditional expectation $\mathbb{E }^{\ell} _{\hat\rho} (\hat{X } | \hat{Y } )$.
Note that Eq.~\eqref{eq:ConsMean} is exactly Eq.~\eqref{eq:iterated expectation} and it has already be shown in the proof of Theorem~\ref{thm:UniqueEcond}.

To prove Eq.~\eqref{eq:Variance}, we note that, from Eq.~\eqref{eq:quadform}, we obtain:  
\begin{eqnarray}    
\min_{f\in \CompFuncObs{\hat Y}}
\langle (\hat X-f(\hat Y) ),(\hat X-f(\hat Y) )\rangle_{\ell}&=& {\rm Tr } (\hat{X }^{\dagger }  \hat{X } \hat\rho ) - \sum _y \frac{ | \langle \varphi_{y}^{\hat Y}, \hat{X } \hat\rho \varphi_{y}^{\hat Y}\rangle |^2 }{ \langle \varphi_y^{\hat Y}, \hat\rho \varphi_y^{\hat Y} \rangle}\nonumber\\
&=&   \langle \hat{X } , \hat{X } \rangle _{\ell} - \langle \mathbb{E }^{\ell} _{\hat\rho} (\hat{X } | \hat{Y } ) , \mathbb{E }^{\ell } _{\hat\rho} (\hat{X } | \hat{Y } ) \rangle _{\ell}.  \label{eq:val_min_1}
\end{eqnarray}
Thus:
\begin{equation}\label{eq:ValueMinLeft}
    \langle \hat X - \mathbb{E }^{\ell} _{\hat\rho} (\hat{X } | \hat{Y } ) , \hat X- \mathbb{E }^{\ell } _{\hat\rho} (\hat{X } | \hat{Y } ) \rangle _{\ell }=\langle \hat{X } , \hat{X } \rangle _{\ell} - \langle \mathbb{E }^{\ell} _{\hat\rho} (\hat{X } | \hat{Y } ) , \mathbb{E }^{\ell } _{\hat\rho} (\hat{X } | \hat{Y } ) \rangle _{\ell }.
\end{equation}
We have that:
\begin{eqnarray} \label{eq:decomp_quantum_variance}   
\Delta_{\hat\rho}^{\ell,2}(\hat X) & = & \mathrm{Tr}(\hat \rho  \hat X^{\dagger}\hat X) - \overline{\mathbb E_{\hat \rho}(\hat X)}\mathrm{Tr}(\hat \rho  \hat X) - \mathbb E_{\hat \rho}(\hat X)\mathrm{Tr}(\hat \rho  \hat X^{\dagger}) + \mathbb E_{\hat \rho}(\hat X)\overline{\mathbb E_{\hat \rho}(\hat X)} \nonumber\\
&=&\langle \hat X,\hat X\rangle_{\ell} - \left|\mathbb E_{\hat \rho}(\hat X)\right|^2 \nonumber\\
&=& \langle \hat X - \mathbb{E }^{\ell} _{\hat\rho} (\hat{X } | \hat{Y } ) , \hat X- \mathbb{E }^{\ell } _{\hat\rho} (\hat{X } | \hat{Y } ) \rangle _{\ell }+ \langle \mathbb{E }^{\ell} _{\hat\rho} (\hat{X } | \hat{Y } ) , \mathbb{E }^{\ell } _{\hat\rho} (\hat{X } | \hat{Y } ) \rangle _{\ell } - |\mathbb E_{\hat \rho}(\mathbb{E }^{\ell} _{\hat\rho} (\hat{X } | \hat{Y } ))|^2\nonumber\\
&=& \langle \hat X - \mathbb{E }^{\ell} _{\hat\rho} (\hat{X } | \hat{Y } ) , \hat X- \mathbb{E }^{\ell } _{\hat\rho} (\hat{X } | \hat{Y } ) \rangle _{\ell }+ \Delta_{\hat\rho}^{\ell,2}(\mathbb{E }^{\ell } _{\hat\rho} (\hat{X } | \hat{Y } ))
\end{eqnarray}
where we used Eq.~\eqref{eq:ConsMean} and Eq.~\eqref{eq:ValueMinLeft} on the third line.

Eq.~\eqref{eq:MeanCondVar} now follows:
\begin{eqnarray*}
    \mathbb E_{\hat\rho}(\Delta^{\ell,2} (\hat X|\hat Y))&=& \mathbb E_{\hat\rho}\left(\mathbb E_{\hat\rho}^{\ell}(|\hat X-\mathbb E_{\hat \rho}^\ell(\hat X|\hat Y)|^2|\hat Y)\right)\\
    &=& \mathbb E_{\hat\rho}\left(|\hat X-\mathbb E_{\hat \rho}^\ell(\hat X|\hat Y)|^2\right) \\
    &=& \mathrm{Tr}\left(\rho|\hat X-\mathbb E_{\hat \rho}^\ell(\hat X|\hat Y)|^2\right) = \langle \hat{X}-\mathbb E_{\hat \rho}^\ell(\hat X|\hat Y), \hat{X}-\mathbb E_{\hat \rho}^\ell(\hat X|\hat Y)\rangle_{\ell}
\end{eqnarray*}
where we used the quantum law of iterated expectations.

We finally prove Eq.~\eqref{eq:CondVarpure}.
We have that :
\begin{eqnarray} \nonumber   
 \Delta^{\ell,2}_{\hat \rho}(\hat X |\hat Y) &=&\mathbb{E }^{\ell } _{\hat \rho } ( |\hat{X } - \mathbb{E }^{\ell } _{\hat \rho } (\hat X | \hat Y ) |^2 | \hat Y ) \\   
 &=&\mathbb{E }^{\ell } _{\hat \rho } (\hat X ^\dagger \hat X | \hat Y ) - \mathbb{E }^{\ell } _{\hat \rho } (\hat X ^\dagger \mathbb{E }^{\ell } _{\hat \rho } (\hat X | \hat Y ) | \hat Y ) - \mathbb{E }^{\ell } _{\hat \rho } ( \mathbb{E }^{\ell } _{\hat \rho } ( \hat X | \hat Y )^{\dagger } \hat X | \hat Y ) + \mathbb{E }^{\ell } _{\hat \rho } (| \mathbb{E }^{\ell } _{\hat \rho } (\hat X | \hat Y ) |^2 | \hat{Y } ) \nonumber \\ 
 &=& \mathbb{E }^{\ell } _{\hat \rho } (\hat X ^\dagger \hat X | \hat Y ) - \mathbb{E }^{\ell } _{\hat \rho } (\hat X ^\dagger \mathbb{E }^{\ell } _{\hat \rho } (\hat X | \hat Y ) | \hat Y ) , \label{eq:exp_quantum_cond_var_bis}    
\end{eqnarray}
We now compute:
\begin{eqnarray*}   
\mathbb{E }^{\ell } _{\hat \rho } (\hat X ^\dagger \mathbb{E }^{\ell } _{\hat \rho } (\hat X | \hat Y ) | \hat Y ) &=& 
\sum _{y\in\sigma(\hat Y)} \frac{{\rm Tr } (\hat \Pi _y^{\hat Y} \hat X ^\dagger \mathbb{E }^{\ell } _{\hat \rho } (\hat X | \hat Y ) \hat \rho ) }{{\rm Tr } (\hat \Pi _y^{\hat Y} \hat \rho ) } \hat \Pi _y^{\hat Y} \\   
&=& \sum _{y\in\sigma(\hat Y)} \frac{1 }{{\rm Tr } (\hat \Pi _y^{\hat Y} \hat \rho ) } \left( \sum _{y'\in\sigma(\hat Y)} \frac{{\rm Tr } (\hat \Pi _{y' }^{\hat Y} \hat{X } \hat \rho ) {\rm Tr } (\hat \Pi _y^{\hat Y} \hat X ^\dagger \hat \Pi _{y' }^{\hat Y} \hat \rho ) }{{\rm Tr } (\hat \Pi _{y' }^{\hat Y} \hat \rho ) } \right) \hat \Pi _y^{\hat Y}. 
\end{eqnarray*}
As $\hat \rho$ is a pure state, we note $\hat\rho = \cket{\psi}\bra{\psi}$ and we obtain that:
\begin{eqnarray*}
    \sum _{y'\in\sigma(\hat Y)} \frac{{\rm Tr } (\hat \Pi _{y' }^{\hat Y} \hat{X } \hat \rho ) {\rm Tr } (\hat \Pi _y^{\hat Y} \hat X ^\dagger \hat \Pi _{y' }^{\hat Y} \hat \rho ) }{{\rm Tr } (\hat \Pi _{y' }^{\hat Y} \hat \rho ) } &=&  \sum _{y'\in\sigma(\hat Y)}\frac{\langle \psi , \varphi _{y' }^{\hat Y} \rangle \langle \varphi _{y' }^{\hat Y} , \hat X \psi \rangle \langle \psi , \varphi _y^{\hat Y} \rangle \langle \varphi _y^{\hat Y} , \hat X ^\dagger \varphi _{y' }^{\hat Y} \rangle \langle \varphi _{y' }^{\hat Y} , \psi \rangle}{| \langle \varphi _{y' }^{\hat Y} , \psi \rangle |^2 }\\
    &=& \sum _{y'\in\sigma(\hat Y)} \langle \psi , \varphi _y^{\hat Y} \rangle  \langle \varphi _y^{\hat Y} , \hat X ^\dagger \varphi _{y' }^{\hat Y} \rangle\langle \varphi _{y' }^{\hat Y} , \hat X \psi \rangle\\
    &=& \langle \psi , \varphi _y^{\hat Y} \rangle \langle \varphi _y^{\hat Y} , \hat X ^\dagger\hat X \psi \rangle = \mathrm{Tr}(\hat\Pi_{y}^{\hat Y} \hat X ^\dagger\hat X \hat \rho).
\end{eqnarray*}
Consequently, we conclude that:
\begin{eqnarray*}
    \mathbb{E }^{\ell } _{\hat \rho } (\hat X ^\dagger \mathbb{E }^{\ell } _{\hat \rho } (\hat X | \hat Y ) | \hat Y ) = \sum _{y\in\sigma(\hat Y)} \frac{\mathrm{Tr}(\hat\Pi_{y}^{\hat Y} \hat X ^\dagger\hat X \hat \rho)}{{\rm Tr } (\hat \Pi _{y' }^{\hat Y} \hat \rho ) }\hat \Pi _y^{\hat Y} = \mathbb{E }^{\ell } _{\hat \rho } (\hat X ^\dagger \hat X | \hat Y ). 
\end{eqnarray*}
This proves Eq.~\eqref{eq:CondVarpure}. Moreover, Eq.\eqref{eq:Variance_pure} is a direct consequence of Eq.~\eqref{eq:Variance}, Eq.~\eqref{eq:MeanCondVar} and Eq.~\eqref{eq:CondVarpure}.

\end{proof}

Having stressed the analogies between the classical and quantum conditional expectations and their basic properties, we now discuss  four important differences between them. 

First, classically, one has, if X is bounded ($X_{\min} \leqslant X(\omega) \leqslant X_{\max}$), that
\begin{equation}\label{eq:condexpboundclass}
    X_{\min}\leq \mathbb E_{\mathbb P}(X|Y=y)\leq X_{\max}.
\end{equation}
We have already pointed out that quantum mechanically, even if $\hat X$ is self-adjoint, $\mathbb E_{\mathbb P}(\hat X|\hat Y=y)$ may be complex non-real. In addition, when it is indeed real, it may not be confined to lie between the minimal and maximal eigenvalue of $\hat X$. One says a value $\mathbb E_{\hat\rho}(\hat X|\hat Y=y)$ of the conditional expectation $\mathbb E_{\hat\rho}(\hat X|\hat Y)$ is anomalous if $\mathbb E_{\hat\rho}(\hat X|\hat Y=y)\not\in [x_{\min}, x_{\max}]$, where $x_{\min, \max}$ are the extremal eigenvalues of $\hat X$. The historic example of such anomalous values corresponds to a spin 1/2~\cite{Aharonov88}. Let $\mathcal H=\C^2$, $\hat X=\hat\sigma_z$, $\hat Y=\hat \sigma_x$ and $\hat\rho=|\alpha\rangle\langle \alpha|$, $|\alpha\rangle=\cos\alpha|-1\rangle+\sin\alpha|1\rangle$. Then,
\begin{equation}\label{example:anomalous_weak}
    \mathbb E^{\ell}_{\hat \rho}(\hat X|\hat Y=1)=\frac{\sin\alpha-\cos\alpha}{\sin\alpha+\cos\alpha},\quad 
\mathbb E^{\ell}_{\hat \rho}(\hat X|\hat Y=-1)=\frac{\sin\alpha+\cos\alpha}{\sin\alpha-\cos\alpha}.
\end{equation}
 Choosing $\alpha$ close to $\pi/4$, for example, the second term can be made arbitrary large, while the first one tends to $0$. This phenomenon lies at the basis of weak value amplification; see~\cite{arvidssonshukur2024properties} for a review on this topic. Note that, when $\alpha=\pi/4$, then $|\alpha\rangle$ does not belong to $D_{\hat Y}$, since it is then an eigenvector of $\hat Y=\hat \sigma_x$.

The existence of anomalous values is a strong indicator of typically quantum mechanical behaviour as the following immediate consequence of Eq.~\eqref{eq:condexpboundclass} shows. 
\begin{lemma}\label{lem:nogo}
    Let $\hat X$ and $\hat Y$ be observables and $\hat \rho\in D_{\hat Y}$. Suppose there exists a probability $\mu: (x,y)\in \sigma(\hat X)\times \sigma(\hat Y)\mapsto \mu(x,y)\in [0,1]$ satisfying
    \begin{itemize}
        \item[(i)] \emph{Born compatibility}: for all $x\in\sigma(\hat X)$, for all $y\in\sigma(\hat Y)$, 
        \begin{equation}
     \sum_{x\in\sigma(\hat X)}\mu(x,y)=\Tr(\hat\Pi_y^{\hat Y}\hat \rho),\quad \sum_{y\in\sigma(\hat Y)}\mu(x,y)=\Tr(\hat\Pi_x^{\hat X}\hat \rho);
     \end{equation}
        \item[(ii)] For all $y\in\sigma(\hat Y)$, $\mathbb E^{\ell}_{\hat\rho}(\hat X|\hat Y=y)=\sum_{x\in\sigma(\hat X)} x \frac{\mu(x,y)}{\Tr(\hat\rho\hat\Pi_y^{\hat Y})}$.
    \end{itemize}
    Then 
    \begin{equation}\label{eq:CONDEXPclassbound}
    x_{\min}\leq \mathbb E_{\hat \rho}^{\ell}(\hat X|\hat Y=y) \leq x_{\max}.
    \end{equation}
\end{lemma}
The same statement holds with $\ell$ replaced by $\mathrm r$. Condition (i) asserts that, for the  given triple $(\hat X, \hat Y, \hat \rho)$, $\mu$ is a joint probability distribution  on the joint spectrum of $\hat X$ and $\hat Y$ that has the correct quantum mechanical Born marginals in the state $\hat \rho$ and $\hat X$ and $\hat Y$. Condition~(ii) adds a condition on the natural associated conditional probability obtained from Bayes' rule, $\mu(x|y)=\frac{\mu(x,y)}{\Tr(\hat\rho\hat Y)}$: it must lead to a conditional expectation for $\hat X$, given $\hat Y$, that coincides with the quantum mechanical conditional expectation $\mathbb E_{\hat \rho}^{\ell}(\hat X|\hat Y=y)$. This then implies that the quantum mechanical conditional expectation cannot take on anomalous values. This lemma is of interest because its contrapositive provides a no-go theorem for the existence of joint probabilities: if, for a given triple $(\hat X,\hat Y,\hat \rho)$, the conditional expectation $\mathbb E_{\hat\rho}^{\ell}(\hat X|\hat Y=y)$ admits at least one anomalous value, then there does not exist a joint probability distribution satisfying (i) and (ii) of the lemma. The usual ``no-go'' theorems of this type, ruling out the existence of Born-compatible joint probability distributions in quantum mechanics, rely on assumptions that need to hold \emph{for all} $\hat \rho$, given a pair of observables $\hat X$ and $\hat Y$ or, alternatively, for all $\hat \rho$ and all $\hat X$ (see for example~\cite{ludwig1983, Fe11, lostaglio2023kirkwood}). Here we consider a single state $\hat \rho$ and add to Born-compatibility of this state a second compatibility condition of the same type, still only for this unique state; indeed (ii)   requires equality between the ``classical'' conditional expectation of $\hat X$ and the quantum one defined in the previous section. If the quantum conditional expectation admits at least one anomalous value (\emph{i.e.} if Eq.~\eqref{eq:CONDEXPclassbound} is violated) then such a joint probability distribution cannot exist.  It should be noted that the Born-compatibility condition in and by itself is very weak and can always be satisfied by taking, for example:
$$
\mu(x,y)=\Tr(\hat\rho \hat \Pi_x^{\hat X})\Tr(\hat\rho\hat \Pi_y^{\hat Y}).
$$
But this treats $x$ and $y$ as independent random variables, which is not coherent with the second condition, generally. The by now standard way to avoid this no-go theorem is to allow $\mu(x,y)$ to be complex-valued, \emph{i.e.} a quasiprobability distribution rather than a probability distribution, as we will see in the next section.

To explain the second difference between the classical and quantum conditional expectations,  we express Eq.~\eqref{eq:Variance} in terms of the self-adjoint and the anti-self-adjoint part of $\mathbb{E}_{\hat \rho}^{\ell/\mathrm{r}}(\hat{X}|\hat{Y})$.
For the self-adjoint part of $\mathbb E_{\hat\rho}^{\ell/\mathrm r}(\hat X|\hat Y)$, we write
\begin{eqnarray}
    \label{eq: def of real econd leftright1}
    \SACondExp{\hat \rho}{\ell}{\hat X}{\hat Y}&:=&\frac{1}{2}\left(\mathbb{E}^{\ell}_{\hat{\rho}}(\hat{X}|\hat{Y}) + \mathbb{E}^{\ell}_{\hat{\rho}}(\hat{X}|\hat{Y})^{\dagger}\right)=\sum_{y\in \sigma(\hat{Y})}\mathrm{Re}\left(\frac{\langle \varphi_y,\hat{X}\hat{\rho}\varphi_y\rangle}{\langle \varphi_y,\hat{\rho}\varphi_y\rangle}\right)\hat{\Pi}_y\\
    \label{eq: def of real econd leftright2}
    \SACondExp{\hat \rho}{\mathrm{r}}{\hat X}{\hat Y}&:=&\frac{1}{2}\left(\mathbb{E}^{r}_{\hat{\rho}}(\hat{X}|\hat{Y}) + \mathbb{E}^{r}_{\hat{\rho}}(\hat{X}|\hat{Y})^{\dagger}\right)=\sum_{y\in \sigma(\hat{Y})}\mathrm{Re}\left(\frac{\langle \varphi_y,\hat{\rho}\hat{X}\varphi_y\rangle}{\langle \varphi_y,\hat{\rho}\varphi_y\rangle}\right)\hat{\Pi}_y,
\end{eqnarray}
for any $\hat{\rho}\in D_{\hat{Y}}$ and any $\hat{X}\in \mathcal{L}(\mathcal{H})$. 
Similarly, we introduce
\begin{equation}
\mathrm{Im}\mathbb{E}^{\ell/\mathrm r}_{\hat{\rho}}(\hat{X}|\hat{Y})=\frac{1}{2i} \left(\mathbb{E}^{\ell/\mathrm r}_{\hat{\rho}}(\hat{X}|\hat{Y})-\mathbb{E}^{\ell/\mathrm{r}}_{\hat{\rho}}(\hat{X}|\hat{Y})^{\dagger}\right)
\end{equation}
As the self-adjoint and anti-self-adjoint parts of $\mathbb{E}_{\hat \rho}^{\ell/\rm r}(\hat X|\hat Y)$ are both functions of $\hat{Y}$, they commute. Then
\begin{equation}
    \Delta_{\hat \rho}^{\ell/\mathrm{r},2}\mathbb E_{\hat \rho}^{\ell/\mathrm{r}}(\hat X|\hat Y)=\Delta_{\hat \rho}^{2}\mathbb E_{\hat \rho}^{\ell/\mathrm{r},\mathrm{sa}}(\hat X|\hat Y)+\Delta_{\hat \rho}^{2}\mathrm{Im}\mathbb E_{\hat \rho}^{\ell/\mathrm{r}}(\hat X|\hat Y).
\end{equation}
In conclusion,  thet law of total variance, Eq.~\eqref{eq:Variance}, shows that the variance of $\hat X$ differs from the sum of the variances of the self-adjoint and anti-self-adjoint part of its conditional expectation given $\hat Y$ by the mean quadratic error between the two. 
Note that, for self-adjoint $\hat X$, the underlying chosen ordering is irrelevant:
\begin{eqnarray}
    \label{eq:equality of real econd for self adjoint}
    \SACondExp{\hat \rho}{\ell}{\hat X}{\hat Y}=\SACondExp{\hat \rho}{\mathrm{r}}{\hat X}{\hat Y}&=&\frac12\sum_{y\in\sigma(\hat Y)} \frac{\langle \varphi_y,\{\hat{\rho},\hat{X}\}\varphi_y\rangle}{\langle \varphi_y,\hat{\rho}\varphi_y\rangle}, \\ \mathrm{Im}\mathbb{E}^{\ell}_{\hat{\rho}}(\hat{X}|\hat{Y})=-\mathrm{Im}\mathbb{E}^{\mathrm{r}}_{\hat{\rho}}(\hat{X}|\hat{Y})&=&\frac1{2i}\sum_{y\in\sigma(\hat Y)} \frac{\langle \varphi_y,[\hat{\rho},\hat{X}]\varphi_y\rangle}{\langle \varphi_y,\hat{\rho}\varphi_y\rangle},
\end{eqnarray}
where $\{\cdot, \cdot\}$ denotes the anti-commutator and $[\cdot, \cdot]$ the commutator. 
Thus, for self-adjoint $\hat X$, we will write 
\begin{equation} \label{eq:E^sa}   
    \mathbb E_{\rho}^{\mathrm{sa}}(\hat X | \hat Y) := \SACondExp{\hat \rho}{\ell}{\hat X}{\hat Y}=\SACondExp{\hat \rho}{\mathrm{r}}{\hat X}{\hat Y}. 
\end{equation}
It should be noted that this equality need not be true if $\hat{X}$ is not a self-adjoint operator. In particular, we have that, generally,
\begin{equation*}
    \SACondExp{\hat \rho}{\ell}{f(\hat{Y})\hat{X}}{\hat Y}\neq\SACondExp{\hat \rho}{\mathrm{r}}{f(\hat{Y})\hat{X}}{\hat Y}
\end{equation*}
for $f:\sigma(\hat{Y})\to \mathbb{R}$, even if $\hat{X}$ is self-adjoint. The law of total variance, Eq.~\eqref{eq:Variance}, then becomes, for self-adjoint $\hat X$,
\begin{equation}
\label{eq: variance real and imaginary parts}
\Delta_{\hat \rho}^{2}\hat X =\Delta_{\hat \rho}^{2}\mathbb E_{\hat \rho}^{\mathrm{sa}}(\hat X|\hat Y)+\Delta_{\hat \rho}^{2}\left(\mathrm{Im}\mathbb E_{\hat \rho}^{\ell/\mathrm{r}}(\hat X|\hat Y)\right) + \langle (\hat X-\mathbb E_{\hat \rho}^{\ell/\mathrm{r}}(\hat X|\hat Y)),(\hat X-\mathbb E_{\hat \rho}^{\ell/\mathrm{r}}(\hat X|\hat Y))\rangle_{\ell/\mathrm{r}}.
\end{equation}
Comparing this to the classical equality Eq.~\eqref{eq:variances_compare_class}, one notices the extra term given by the variance of $\mathrm{Im}\mathbb E_{\hat \rho}^{\ell/\mathrm{r}}(\hat X|\hat Y)$ which comes from the non-self-adjointness of the conditional expectation.  This term depends on the commutator $[\hat \rho, \hat X]$ and as such, its absence classically is expected. We will see in Section~\ref{s:Q_cond_exp_int} that it can vanish even if $\hat\rho$ and $\hat X$ do not commute and we will  use  results on the KD representations of quantum mechanics to show many examples where $\mathrm{Im}\mathbb E_{\hat \rho}^{\ell/\mathrm{r}}(\hat X|\hat Y)$ vanishes. We also further explore there  the link between $\mathrm{Im}\mathbb E_{\hat \rho}^{\ell/\mathrm{r}}(\hat X|\hat Y)$ and an appropriately defined Fisher information associated to a parameter estimation protocol naturally associated with a triple $\hat \rho$, $\hat X$, and $\hat Y$.

The third difference with the classical situation concerns what happens for pure states. In classical probability theory, the pure states are Dirac delta measures $\delta_{\omega_0}$, on $\Omega$, with $\omega_0\in\Omega$. If $X, Y$ are random variables, $\mathbb E_{\delta_{\omega_o}}(X|Y=y)=X(\omega_0)\delta_{Y(\omega_0)}(y)$, and is therefore also a Dirac delta measure, on $\R$. As a result, the classical analogues to Eq.~\eqref{eq:CondVarpure} and Eq.~\eqref{eq:Variance_pure} also hold, but with one major difference: since the Dirac delta measures are dispersion-free, both sides of this classical equivalent of Eq.~\eqref{eq:Variance_pure} then vanish. In the quantum case, pure states $\hat \rho=|\psi\rangle\langle \psi|$ are never dispersion-free and neither side vanishes as soon as $|\psi\rangle$ is not an eigenstate of  $\hat X$. So in quantum theory, the error term in the law of total variation vanishes for pure states, and the variance of $\hat X$ equals that of its conditional expectation $\mathbb E_{|\psi\rangle\langle \psi|}(\hat X|\hat Y)$.

Fourth, and finally, we note that the well-known relation that holds  in classical probability
\begin{equation*}
    \Delta^2(X|Y) = \mathbb{E}_{\mathbb{P}}(\left|X\right|^2|Y) - \left|\mathbb{E}_{\mathbb{P}}(X|Y)\right|^2,
\end{equation*}
does not hold in the quantum case.  Indeed, it follows from  Eq.~\eqref{eq:exp_quantum_cond_var_bis} that, for the left quantum conditional expectation, for example:
\begin{equation*}
    \Delta^{\ell,2}(\hat X|\hat Y) = \mathbb{E }^{\ell } _{\hat \rho } (\hat X ^\dagger \hat X | \hat Y ) - \mathbb{E }^{\ell } _{\hat \rho } (\hat X ^\dagger \mathbb{E }^{\ell } _{\hat \rho } (\hat X | \hat Y ) | Y )  \neq \mathbb{E }^{\ell } _{\hat \rho }(|\hat X|^2|Y) - \left|\mathbb{E }^{\ell } _{\hat \rho }(\hat X|\hat Y)\right|^2.
\end{equation*}
   
\noindent One can also define an alternative conditional variance by the final expression above, specifically,     
\begin{equation}   
\label{eq:alternative cond variance}
\tilde \Delta ^{2, \ell  } _{\hat \rho } (\hat X | \hat Y ) := \mathbb{E }^{\ell } _{\hat \rho }(|\hat X|^2|Y) - \left|\mathbb{E }^{\ell } _{\hat \rho }(\hat X|\hat Y)\right|^2 ,      
\end{equation}   
with an obvious right variant obtained by replacing $|\hat X |^2 = \hat X ^{\dagger } \hat X $ by $\hat X \hat X ^{\dagger } $ (the two coincide if $\hat X $ is self-adjoint).   
These will then still satisfy the law of total variance, in the form   
$$   
\Delta ^{2, \ell } (\hat X ) = \Delta ^{2, \ell } \bigl( \mathbb{E } ^{\ell } _{\hat \rho } (\hat X | \hat Y ) \bigr) + \mathbb{E }_{\hat \rho } \bigl( \tilde \Delta ^{2,  \ell  } (\hat X | \hat Y ) \bigr).  
$$   
Another notion of conditional variance that has been proposed in the literature is the {\it weak variance} introduced in  \cite{Ogawa_2021}, which with our notations is defined as     
\begin{equation}   
\sigma _{w, \hat \rho } (\hat X | \hat Y ) = \mathbb{E }^{\ell}_{\hat \rho } (\hat X ^2 | \hat Y ) - \mathbb{E }^{\ell}_{\hat \rho } ( \hat X | \hat Y )^2.   
\end{equation}  
   
This weak variance can be given an operational meaning in the context of weak measurements (cf. Appendix A, in particular formula~\eqref{eq:OgawaResult}) but does not satisfy the law of total variance, unlike (\ref{eq:alternative cond variance}) which can be given a similar operational significance: see Corollary \ref{cor:alt_expr_c2}. There exists in fact a whole range of plausible definitions of conditional variances, including variants which can be introduced in the context of quasiprobabilistic representations of quantum mechanics, whose relative merits still have to be investigated.

\subsection{The real part of \texorpdfstring{$\mathbb E_{\hat\rho}^{\ell/\mathrm{r}}(\hat X|\hat Y)$}{E[...]}  as a best estimator}\label{s:real part}

 As one might expect, the left and right self-adjoint quantum conditional expectations $\SACondExp{\hat \rho}{\ell/\mathrm{r}}{\hat X}{\hat Y}$ satisfy properties similar to the ones given in Theorem \ref{thm:UniqueEcond}. We specify them for the left self-adjoint quantum conditional expectation. They also hold for the right self-adjoint one, with adapted ordering. We denote by $\ReFuncObs{\hat Y}$ the set of self-adjoint functions of $\hat{Y}$:
\begin{equation}
    \ReFuncObs{\hat Y} = \left\{f(\hat{Y}) \mid f : \sigma(\hat{Y}) \to \R \right\}.
\end{equation}
Adapting slightly the proof of Theorem \ref{thm:UniqueEcond}, one easily shows that for $\hat{\rho}\in D_{\hat{Y}}$, $\SACondExp{\hat \rho}{\ell}{\cdot}{\hat Y}$ is the unique map with values in the set $\ReFuncObs{\hat Y}$ with the two following properties:
\begin{enumerate}
    \item For any $f:\sigma(\hat{Y})\to \mathbb{R}$ and any $\hat{X}\in \mathcal{L}(\mathcal{H})$,
    \begin{equation}
        \SACondExp{\hat \rho}{\ell}{f(\hat{Y})\hat{X}}{\hat Y}=f(\hat{Y})\SACondExp{\hat \rho}{\ell}{\hat{X}}{\hat Y}.
    \end{equation}
    \item For any $\hat{X}\in \mathcal{L}(\mathcal{H})$,
    \begin{equation}
        \mathbb{E}_{\hat{\rho}}(\SACondExp{\hat \rho}{\ell}{\hat{X}}{\hat Y})=\mathrm{Re}(\mathbb{E}_{\hat{\rho}}(\hat{X})).
    \end{equation}
\end{enumerate}
Finally, it is also easily seen that this left/right  conditional expectations can be obtained as the minimizer of the mean square error functional
\begin{equation} \label{eq:min_problem_E^R}
    f\mapsto \langle \hat{X}-f(\hat{Y}), \hat{X}-f(\hat{Y})\rangle_{\ell/\rm r},
\end{equation}
where the minimization is considered over the set of \textit{real valued} functions $f:\sigma(\hat{Y})\to\mathbb{R}$. The self-adjoint left conditional expectation of $\hat{X}$ is therefore the observable closest to $\hat{X}$ (with respect to the above degenerate left inner product) that is a self-adjoint function of $\hat{Y}$. 

For self-adjoint $\hat X$, this real minimization problem was previously addressed in ~\cite{dressel2015,brummelhuis2024,Brummelhuis2025a,Brummelhuis2025b,tsang2023}, where the link between the minimizer, viewed as a conditional expectation, and the real part of weak values was also established.
As we pointed out above, our analysis in Section~\ref{s:Q_cond_exp_def} shows that, in fact, the  conditional expectation $\mathbb E_{\hat\rho}^{\ell/\mathrm r}(\hat X|\hat Y)$ itself is a best estimator of $\hat X$ (self-adjoint or not), when the minimization is performed over all complex valued operator functions of $\hat Y$, and not only over the self-adjoint ones.

As in the proof of Proposition \ref{prop:variances_condexp1}, one finds that for any $\hat X\in \mathcal{L}(\mathcal{H})$
\begin{equation}
    \langle \hat X-\mathbb{E}_{\rho}^{\ell/\rm r,\rm sa}(\hat X|\hat Y),\hat X-\mathbb{E}_{\rho}^{\ell/\rm r,\rm sa}(\hat X|\hat Y)\rangle _{\ell/\rm r}=\mathbb{E}_{\hat \rho}(|\hat X|^2)-\mathbb{E}_{\hat{\rho}}\left(\mathbb{E}_{\hat \rho}^{\ell/\rm r,\rm sa}(\hat X|\hat Y)^2\right).
\end{equation}

If furthermore $\hat X$ is self-adjoint, we have $\mathbb{E}_{\hat \rho}(\mathbb{E}_{\hat \rho}^{\rm sa}(\hat X|\hat Y))=\mathrm{Re}\left(\mathbb{E}_{\hat{\rho}}(\hat X)\right)=\mathbb{E}_{\hat{\rho}}(\hat X)$. We deduce
\begin{equation}
    \label{eq: error variance self-adjoint}
    \Delta_{\hat \rho}^2\hat X = \Delta_{\hat \rho}^2\mathbb{E}_{\hat{\rho}}^{\rm sa}(\hat X|\hat Y) + \langle \hat X-\mathbb{E}_{\hat{\rho}}^{\rm sa}(\hat X|\hat Y),\hat X-\mathbb{E}_{\hat{\rho}}^{\rm sa}(\hat X|\hat Y)\rangle_{\ell/\rm r}.
\end{equation}

Comparing this with Eq~\eqref{eq:Variance}, we have that the error made when approximating the variance of $\hat X$ by the variance of $\mathbb{E}^{\rm sa}_{\hat \rho}(\hat X|\hat Y)$ is greater than the error  made when approximating the variance of $\hat X$ by the variance of $\mathbb{E}_{\hat{\rho}}^{\ell/\rm r}(\hat X|\hat Y)$. This is a consequence of the fact that  the mean squared error functional is minimized over a larger class of functions in the latter case than in the former. More precisely, from Eq~\eqref{eq: variance real and imaginary parts}, we have
\begin{equation*}
    \langle \hat X-\mathbb{E}_{\hat{\rho}}^{\rm sa}(\hat X|\hat Y),\hat X-\mathbb{E}_{\hat{\rho}}^{\rm sa}(\hat X|\hat Y)\rangle_{\ell/\rm r} = \Delta_{\hat \rho}^{2}\left(\mathrm{Im}\mathbb E_{\hat \rho}^{\ell/\mathrm{r}}(\hat X|\hat Y)\right) + \langle (\hat X-\mathbb E_{\hat \rho}^{\ell/\mathrm{r}}(\hat X|\hat Y)),(\hat X-\mathbb E_{\hat \rho}^{\ell/\mathrm{r}}(\hat X|\hat Y))\rangle_{\ell/\mathrm{r}}.
\end{equation*}

\section{Quantum conditional expectation via quasiprobability representation of quantum mechanics}\label{s:QCEQP_new}
As shown in the previous section, defining the quantum conditional expectation via a minimization problem as in Eq.\eqref{eq:Q_min_prob_left} or Eq.\eqref{eq:Q_min_prob_right} is a way of mimicking what happens classically, see Theorem~\ref{thm:classical_minimization}. It also allowed us to  compare the classical and quantum notions and to emphasize their differences.   One can alternatively try to imitate the classical conditional expectation by using Definition~\ref{def:condexpclass1}. This definition uses the joint probability distribution of the two random variables under consideration, which is naturally  defined in classical probability theory. In quantum mechanics,  as recalled in the introduction, a notion of joint probability distribution does not exist for non-commuting observables, but one can try to use a quasiprobability distribution instead. 

To that end, we first describe the class of quasiprobability representations of quantum mechanics that we consider. We will use the formalism of frames, as detailed in~\cite{ferrieetal2010,Fe11},
to do so.  Let $\mathcal H$ be a Hilbert space of dimension $d$, let $\Lambda$ be a finite set, with $\left|\Lambda\right|=d^2$, where $\left|A\right|$ is the cardinal of the finite set $A$. Let $(S_\lambda)_{\lambda\in\Lambda}$ be a basis of $\mathcal L(\mathcal H)$, satisfying
\begin{equation}
    \sum_\lambda S_\lambda =\id{d}.
\end{equation} 
Let $(T_\lambda)_{\lambda\in \Lambda}$ denote its unique dual basis, which satisfies:
\begin{equation}
\label{eq: dual frame}
    \Tr(T_\lambda^\dagger S_{\lambda'})=\delta_{\lambda, \lambda'}.
\end{equation} 
Given an operator $\hat C$ on $\mathcal H$, we define 
\begin{equation}\label{eq:defQPRFrame}
    Q_\lambda(\hat C)=\Tr(\hat C \hat S^\dagger_\lambda),\quad \tilde Q_\lambda(\hat C)=\Tr(\hat C T^\dagger_\lambda).
\end{equation}
It follows that the maps $Q, \tilde Q:\mathcal L(\mathcal H)\to \mathbb C^\Lambda$ are bijective and that 
\begin{equation}\label{eq:hatCdecompS}
    \hat C =\sum_\lambda \tilde Q_\lambda(\hat C) S_\lambda,\quad \hat C=\sum_\lambda Q_\lambda(\hat C) T_\lambda. 
\end{equation}
Given $\hat C, \hat D\in\mathcal L(\mathcal H)$, one then has 
\begin{equation}
    \Tr\hat C=\sum_\lambda Q_\lambda(\hat C),\quad \Tr(\hat C^\dagger \hat D)=\sum_\lambda \overline{\tilde Q_\lambda(\hat C)} Q_\lambda(\hat D).
\end{equation}
The pair $(Q,\tilde Q)$, which is completely determined by the  choice of $\Lambda$ and of the basis $(S_\lambda)_\lambda$,  is referred to as a quasiprobability representation of quantum mechanics. It associates to each density matrix $\hat \rho$ a quasiprobability $Q_\lambda(\hat \rho)$ on $\Lambda$, which is a complex-valued function that satisfies 
\begin{equation}
    1=\sum_\lambda Q_\lambda(\hat \rho),
\end{equation}
and to each observable $\hat C=\hat C^\dagger$ its symbol $\tilde Q_\lambda(\hat C)$, with
\begin{equation}\label{eq:overlap}
    \Tr(\hat \rho \hat C)=\sum_\lambda Q_\lambda(\hat\rho)\overline{\tilde Q_\lambda(\hat C).}
\end{equation}

Depending on context, one thinks of $\Lambda$ as an ontic space or as a classical phase space, on which the quantum states and observables are represented by quasiprobabilities and functions respectively. Many quasiprobability representations of quantum mechanics fitting in the above framework have been introduced and studied. A number of those~\cite{wootters1987, cohendetetal1988,leonhardt1995,leonhardt1996, bouzdb96, wootters2004,gross2006} aim at defining quasiprobability representations for quantum systems on finite dimensional Hilbert space reproducing many of the known properties of the Wigner-Weyl-Moyal representation \cite{wigner1932, moyal1949, cagl69a, cagl69b} associated with conjugate variables and the Fourier transform on $L^2(\mathbb R^2)$. Extensions of such constructions to locally compact Abelian groups have been  considered  more recently in~\cite{benyetal2025,NicolaRiccardi2025}. 
A more versatile family of quasiprobability representation is provided by the Kirkwood-Dirac representations: they are defined using two arbitrary  observables $\hat A$ and $\hat B$ and we shall describe them in detail in the next section.    Further examples of quasiprobability representations can be found in~\cite{Fe11}.

\begin{definition}\label{def:Borncompat} Let $\hat Y$ be a $\mathrm{CSCO}$ and let $(Q,\tilde Q)$ be a quasiprobability representation of quantum mechanics. Writing $\tilde{Q}_{\lambda } (\hat{Y } ) := (\tilde{Q }_{\lambda } (\hat{Y }_1 ) , \ldots, \tilde{Q }_{\lambda } (\hat{Y }_m ) ) $, we define, for all $y\in\C^m:$ 
$$
\Lambda_y^{\hat{Y}}=\{\lambda\in\Lambda\mid \tilde Q_\lambda(\hat Y)=y\}.
$$
We say the quasiprobability representation $(Q,\tilde Q)$ of quantum mechanics is $\hat Y$-compatible provided  
for all density matrices $\hat\rho$, one has:
\begin{eqnarray}\label{eq:marginal1}
    \sum_{\lambda\in\Lambda_y^{\hat{Y}}} Q_\lambda(\hat \rho)&=&\Tr(\hat \rho \hat{\Pi}_y^{
\hat Y}
    ),\quad \mathrm{when}\quad y\in \sigma(\hat Y),\\
    &=&0,\quad \mathrm{when}\quad y\not\in\sigma(\hat Y).\label{eq:marginal2}
\end{eqnarray}
\end{definition}   
 
If $Q_\lambda(\hat\rho)$ is a probability on $\Lambda$ (meaning it is non-negative for all $\lambda$), then what this condition means is that the joint probability law of the vector-valued random variable   $\tilde Q(\hat Y)=(\tilde Q(\hat Y_1),\dots, \tilde Q(\hat Y_m))\in (\mathbb{C }^m )^{\Lambda } $ induced by the probability $Q_\lambda(\hat\rho)$ on $\Lambda$ is identical to the quantum mechanical Born-probability law of $\hat Y$ when the system is in the state $\hat \rho$. 
More generally, to be $\hat Y$-compatible, a quasiprobability representation must reproduce the correct Born-probabilities for the observable $\hat Y$ and for all states $\hat\rho$, even those for which the quasiprobability $Q_\lambda(\hat\rho)$ takes on non-positive values and is therefore only a quasiprobability distribution.   
  
The following lemma expresses the fact that, if $(Q,\tilde Q)$ is $\hat Y$-compatible, then it agrees naturally with the functional calculus of $\hat Y$ in the sense that the symbol of $f(\hat Y)$ equals $\lambda\mapsto f(\tilde Q_\lambda(\hat Y))$.    
   
\begin{lemma}\label{lem:compatible} If $(Q,\tilde Q)$ is a $\hat Y$-compatible quasiprobability representation, then the family $(\Lambda_y^{\hat Y})_{y\in\sigma(\hat Y)}$ is a partition of $\Lambda$ and   
    \begin{equation}
    \forall f(\hat Y)\in \CompFuncObs{\hat Y},\ \forall \lambda\in \Lambda,\, \tilde Q_\lambda(f(\hat Y))=f(\tilde Q_\lambda(\hat Y)). 
\end{equation}
\end{lemma}
The lemma states that the symbol of a function $f$ of $\hat Y$ equals the same function $f$ of the symbol of $\hat Y$. Note that, generally, if $\hat X$ and $\hat Z$ are operators on $\mathcal H$, then it is not true that 
$\tilde Q(\hat X\hat Z)=\tilde Q(\hat X)\tilde Q(\hat Z)$. When both $\hat X$ and $\hat Z$ are functions of $\hat Y$, this is true, however, by the above lemma, for $\hat{Y} $-compatible quasiprobability representations.

\begin{proof}
From Eq.~\eqref{eq:marginal1}, one finds that for all $\hat\rho$ and for all $y\in\sigma(\hat Y)$
\begin{equation} \label{eq:proof_lemma_compatible}
\Tr\left(\hat\rho\sum_{\lambda\in\Lambda_y^{\hat{Y}}} S_\lambda^\dagger\right)=\Tr\left(\hat \rho \hat{\Pi}_y^{
\hat Y}\right).
\end{equation}
Therefore,
\begin{equation}
     \hat{\Pi}_y^{\hat Y}=\sum_{\lambda\in\Lambda_y^{\hat{Y}}} S_\lambda^\dagger=\sum_{\lambda\in\Lambda_y^{\hat{Y}}} S_\lambda.
\end{equation}
Eq.~\eqref{eq:hatCdecompS} then implies that 
\begin{equation}
    \tilde Q(\hat \Pi_y^{\hat Y})=\mathds{1}_{\Lambda_y^{\hat Y}},
\end{equation}
where $\mathds{1} _E $ denotes the indicator function of a set $E $, and therefore
\begin{equation}
1=\tilde Q(\id{d})=\sum_{y\in\sigma(\hat Y)} \tilde Q(\hat \Pi_y^{\hat Y})=\sum_{y\in\sigma(\hat Y)} \mathds{1}_{\Lambda_y^{\hat Y}}.
\end{equation}
Moreover, it is clear that $\Lambda_y^{\hat{Y}}\cap \Lambda_{y'}^{\hat{Y}}=\emptyset$ if $y\neq y'$. Since none of the $\Lambda ^{\hat{Y } } _y $ are empty, by Eq.~\eqref{eq:proof_lemma_compatible},
this implies that $(\Lambda_y^{\hat Y})_{y\in\sigma(\hat Y)}$ is a partition of $\Lambda$.
Consequently, one has that 
\begin{equation} \label{eq:proof_lemma_compatible2}   
    \tilde Q(\hat Y)=\sum_{y\in\sigma(\hat Y)} y \mathds{1}_{\Lambda_y^{\hat Y}} .   
\end{equation}   
Finally,   
\begin{equation}
\tilde Q(f(\hat Y))=\sum_{y\in\sigma(\hat Y)} f(y)\tilde Q(\hat \Pi_y^{\hat Y})=\sum_{y\in\sigma(\hat Y)} f(y)\mathds{1}_{\Lambda_y^{\hat Y}}=f(\tilde Q(\hat Y)).
\end{equation}

\end{proof}

Suppose now  that we have  a $(Q,\tilde Q)$ that is $\hat Y$-compatible. 
Given $\hat\rho\in D_{\hat Y}$, we then define the quasiprobability of $\lambda$, given $y \in \sigma(\hat Y)$, as  follows:
\begin{alignat}{3}
    \label{eq:DefCondQPR} Q_{\lambda|y}(\hat \rho)&=&\frac{Q_\lambda(\hat\rho)}{Q(\Lambda_y^{\hat Y})}=\frac{Q_\lambda(\hat \rho)}{\langle\varphi_y|\hat\rho|\varphi_y\rangle} &\quad \lambda\in\Lambda_y^{\hat{Y}},\\
    &=& 0 \qquad\qquad\qquad\quad& \quad \lambda\not\in \Lambda_y^{\hat{Y}}.
\end{alignat}
This allows us to define a notion of conditional expectation naturally associated to the quasiprobability representation $(Q,\tilde Q)$, as follows. 
\begin{definition}
\label{def:Qconditionalexpectation_Sigma}
     For $\hat{Y } $-compatible $(Q, \tilde{Q } ) $ we define the $Q$-conditional expectation of $\hat{X}\in \mathcal{L}(\mathcal{H})$ knowing $\hat Y$ in the state $\hat{\rho}\in D_{\hat Y}$ by
    \begin{equation}
        \QPEcond{Q}{\hat{X}}{\hat Y} = \sum_{y\in \sigma(\hat Y)}\sum_{\lambda\in \Lambda_y^{\hat{Y}}}\overline{\tilde{Q}_{\lambda}(\hat{X}^{\dagger})}Q_{\lambda|y}(\hat{\rho})\hat{\Pi}_y^{\hat Y}.
    \end{equation}
\end{definition}
Note that, whenever $\hat \rho$ is $Q$-positive, by which we mean that $Q_\lambda(\hat\rho)\geq 0$ for all $\lambda$, $Q_{\lambda|y}(\hat\rho)$ is equal to the conditional probability of $\lambda$, given that the random variable $\tilde Q(\hat Y)$ takes the value $y$: $\tilde Q(\hat Y)=y$. Nevertheless, even in that case, and even if $\hat X$ is self-adjoint, $\mathbb E_{\hat \rho}^Q(\hat X|\hat Y)$ is not necessarily self-adjoint. This will be the case, provided $\tilde Q_\lambda(\hat X)$ is real.

The following result is now immediate:
\begin{proposition}
\label{prop: iterated expectation_Y}
    The $Q$-conditional expectation has the following property: for any $\hat{\rho} \in D_{\hat Y}$ and any $\hat{X}\in \mathcal{L}(\mathcal{H})$,
    \begin{equation}
    \label{eq:Qiteratedexpectation_Sigma}
        \mathbb{E}_{\hat{\rho}}(\QPEcond{Q}{\hat{X}}{\hat Y})=\mathbb{E}_{\hat{\rho}}(\hat{X}).
    \end{equation}
    In particular, for any $\rho\in D_{\hat Y}$ and any $\lambda\in \Lambda$
\begin{equation}\label{eq:Qjoint_probab_from_cond_exp_Y}
        Q_{\lambda}(\hat{\rho})=\Tr(\hat{\rho} S^{\dagger}_{\lambda})=\mathbb{E}_{\hat{\rho}}(\QPEcond{Q}{S_{\lambda}^{\dagger}}{\hat{Y}}).
    \end{equation}
\end{proposition}

We can now formulate our first main result. It provides a characterization of all quasiprobability distributions that are compatible with the projective measurement of a given CSCO $\hat Y$ and the associated conditional expectation of which coincides with the one defined in terms of minimization, as in the previous section. Recall that $\lambda \mapsto \tilde Q_\lambda(\hat Y) $ sends $\Lambda $ onto the spectrum of $\hat{Y } $, by Eq.~\eqref{eq:proof_lemma_compatible2}. 
\begin{Theorem}\label{thm:main_hatY}
    Let $(Q,\tilde Q)$ be a $\hat Y$-compatible quasiprobability representation of quantum mechanics on $\mathcal H$. Then the following statements are equivalent:
    \begin{itemize}
        \item[(i)] $\forall \hat X\in\mathcal L(\mathcal H), \forall\hat\rho\in D_{\hat Y}, \mathbb E^Q_{\hat \rho}(\hat X|\hat Y)=\mathbb E_{\hat\rho}^\ell(\hat X|\hat Y)$;
        \item[(ii)]  $\forall \hat X \in\mathcal{L}(\mathcal{H}), \forall\hat\rho\in D_{\hat Y}, \ \forall f:\sigma(\hat Y)\to\mathbb C$, $\mathbb E_{\hat{\rho}}^Q(f(\hat Y) \hat X|\hat Y)=f(\hat Y)\mathbb E_{\hat{\rho}}^Q(\hat X|\hat Y)$.   
        \item[(iii)] $\forall\lambda \in\Lambda, \forall\hat\rho\in D_{\hat Y}, S_\lambda=S_\lambda \hat\Pi_{y}^{\hat Y}, \ T_\lambda=T_\lambda \hat\Pi_{y}^{\hat Y}$, with $y=\tilde Q_\lambda(\hat Y)\in\sigma(\hat Y)$.
    \end{itemize}   
\end{Theorem}
 Condition (iii) implies that both the frame operators $S_\lambda$ and their duals $T_\lambda$ are rank one operators, which is quite a restrictive condition. It is this property that will allow us to single out the KD distributions in Section~\ref{s:uniqueKD} below.
 
\begin{proof}
That (i) and (ii) are equivalent follows from Theorem~\ref{thm:UniqueEcond} and Eq.~\eqref{eq:Qiteratedexpectation_Sigma}. We now prove that (i) implies (iii). From (i), one finds that, for all $\hat X\in\mathcal L(\mathcal H)$, for all $y\in\sigma(\hat Y)$ and $\hat\rho\in D_{\hat Y}$,
\begin{equation}
    \sum_{\lambda\in\Lambda_y^{\hat{Y}}}\overline{\tilde Q_\lambda(\hat X^\dagger)}Q_{\lambda|y}(\hat{\rho})=\frac{\Tr(\hat \rho\hat\Pi_y^{\hat Y}\hat X)}{\Tr(\hat\rho \hat\Pi_y^{\hat Y})}.
\end{equation}
Inserting $\hat X=S_{\lambda'}^\dagger$ and using that $\tilde Q_{\lambda}(S_{\lambda'})=\delta_{\lambda,\lambda'}$, one finds that
\begin{equation}
\sum_{\lambda\in\Lambda_y^{\hat{Y}}}\delta_{\lambda\lambda'} Q_\lambda(\hat\rho)=\Tr(\hat\rho \hat\Pi_y^{\hat Y} S_{\lambda'}^\dagger).
\end{equation}
Since both sides are continuous in $\hat\rho$, this equality holds for all $\hat\rho$ by Lemma \ref{lemma: D_Y dense} and hence, for all $y\in\sigma(\hat Y), \lambda'\in\Lambda$,
\begin{equation}\label{eq:tool}
\sum_{\lambda\in\Lambda_y^{\hat{Y}}}\delta_{\lambda\lambda'}S_\lambda^\dagger=\hat\Pi_y^{\hat Y} S_{\lambda'}^\dagger.
\end{equation}
We know from Lemma~\ref{lem:compatible} that $\tilde Q_{\lambda'}(\hat Y)\in\sigma(\hat Y)$. Consequently,  Eq.~\eqref{eq:tool} implies that, if $\lambda'\not\in \Lambda_y^{\hat{Y}}$, then $S_{\lambda'}\hat\Pi_y^{\hat Y}=0$ and if $\lambda'\in\Lambda_y^{\hat Y}$, then $S_{\lambda'}=S_{\lambda'}\hat\Pi_y^{\hat Y}$. Since $\sum_{y\in\sigma(\hat Y)}\hat\Pi_y^{\hat Y}=\id{d}$, it follows that, for all $\lambda'\in\Lambda$
\begin{equation}
S_{\lambda'}=\sum_{y\in\sigma(\hat Y)}S_{\lambda'}\hat\Pi_y^{\hat Y}=S_{\lambda'}\hat\Pi_{\tilde Q_{\lambda'}(\hat Y)}^{\hat Y}+\sum_{y\in\sigma(\hat Y)\setminus\{\tilde Q_{\lambda'}(\hat Y)\}}S_{\lambda'}\hat\Pi_y^{\hat Y}=S_{\lambda'} \hat\Pi_{\tilde Q_{\lambda'}(\hat Y)}^{\hat Y}.
\end{equation}
Introducing, for $y\in\sigma(\hat Y)$,
\begin{equation}
    \mathcal L_y=\mathrm{span} \{S_\lambda\mid \lambda\in\Lambda_y^{\hat{Y}}\},
\end{equation}
it follows that these spaces are orthogonal with respect to the Hilbert Schmidt inner product, 
\begin{equation}
    \mathcal L_y\perp \mathcal L_{y'} \ \textrm{if}\ y\not=y'.
\end{equation}
Since the $\Lambda_y^{\hat Y}$ form a partition of $\Lambda$, one therefore concludes that
\begin{equation}
    \mathcal L(\mathcal H)=\bigoplus_{y\in\sigma(\hat Y)}^{\perp} \mathcal L_y, 
\end{equation}
where the sum is an orthogonal direct sum for the Hilbert-Schmidt inner product.  It follows from this that 
\begin{equation}
T_\lambda=T_\lambda\hat \Pi_{\tilde Q_\lambda(\hat Y)}^{\hat Y}.
\end{equation}
Indeed, $T_\lambda$ is by definition orthogonal to all $S_{\lambda'}$ with $\lambda'\not=\lambda$. It is therefore orthogonal to all $S_{\lambda'}$ with $\lambda'$ so that $\tilde Q_{\lambda'}(\hat Y)\not=\tilde Q_\lambda(\hat Y)$.
Hence $T_\lambda$ is orthogonal to each of the subspaces $\mathcal L_{y'}$ with $y'\not=\tilde Q_\lambda(\hat Y)$ and therefore
\begin{equation}
T_\lambda=\sum_{\lambda'\in\Lambda_{\tilde Q_\lambda(\hat Y)}}c_{\lambda'} S_{\lambda'}=\sum_{\lambda'\in\Lambda_{\tilde Q_\lambda(\hat Y)}}c_{\lambda'} S_{\lambda'}\hat \Pi_{\tilde Q_\lambda(\hat Y)}^{\hat Y}=T_\lambda\hat\Pi_{\tilde Q_\lambda(\hat Y)}^{\hat Y}.
\end{equation}
This implies (iii).

We finally show that (iii) implies (ii). For that purpose, we compute
\begin{eqnarray}
    \tilde Q_\lambda(\hat X^\dagger f(\hat Y)^\dagger)&=&\Tr(\hat X^\dagger f(\hat Y)^\dagger T_\lambda^\dagger)\nonumber\\
    &=&\Tr(\hat X^\dagger f(\hat Y)^\dagger \hat\Pi_{\tilde Q_\lambda(\hat Y)}^{\hat Y}T_\lambda^\dagger)\nonumber\\
&=&\overline{f(\tilde Q_\lambda(\hat Y))}\Tr(\hat X^\dagger T_\lambda^\dagger)\nonumber\\
&=&\overline{f(\tilde Q_\lambda(\hat Y))}\tilde Q_\lambda(\hat X^\dagger),
\end{eqnarray}
where we used for the third line that $f(\hat{Y } )^{\dagger } \hat\Pi_{\tilde Q_\lambda(\hat Y)}^{\hat Y} = \overline{f(Q_{\lambda }(\hat{Y } )) } \hat\Pi_{\tilde Q_\lambda(\hat Y)}^{\hat Y} $. One then computes:     
\begin{eqnarray*}   
\mathbb{E }_{\hat{\rho } } ^Q (f(\hat{Y } ) \hat{X } | \hat{Y } ) &=& \sum_{y\in \sigma(\hat Y)}\sum_{\lambda\in \Lambda_y^{\hat{Y}}} \left( f(\tilde{Q}_{\lambda } (\hat{Y } ) ) \overline{\tilde{Q}_{\lambda}(\hat{X}^{\dagger})}Q_{\lambda|y}(\hat{\rho})\right) \hat{\Pi}_y^{\hat {Y} } \\   
&=& \sum_{y\in \sigma(\hat Y)}f(y) \sum_{\lambda\in \Lambda_y^{\hat{Y}}} \left(  \overline{\tilde{Q}_{\lambda}(\hat{X}^{\dagger})}Q_{\lambda|y}(\hat{\rho})\right) \hat{\Pi}_y^{\hat{Y } } \\   
&=& f(\hat{Y } ) \mathbb{E }_{\hat{\rho } } ^Q ( \hat{X } | \hat{Y } ),   
\end{eqnarray*}
which establishes (ii).
\end{proof}

\section{Characterizing quasiprobability representations}

In Section \ref{s:QCEQP_new}, we have associated a notion of conditional expectation to every quasiprobability representation $(Q,\tilde{Q})$ that is Born-compatible with a CSCO $\hat Y$. In this section, we consider all quasiprobability representations that are compatible with two given complementary CSCO $\hat A$ and $\hat B$ (see Definition \ref{def:complementary} for the notion of complementarity of CSCO). 

In Section \ref{s:KDdef}, we recall the definition of the Kirkwood-Dirac representations Born-compatible with $\hat A$ and $\hat B$. In Section \ref{s:uniqueKD}, we show that those KD representations are the only quasiprobability representations for which the associated conditional expectation (given $\hat A$ or given $\hat B$) satisfies the ``pull-out formula'' (see Theorem \ref{thm:KD_unique_pullout}). This structural result provides a unique characterization of the KD quasiprobability representations that allows us to show 
Theorem~\ref{thm:KDunique1_bis}, which states that the (left/right) Kirkwood-Dirac representation is the only quasiprobability representation for which the associated conditional expectations (given $\hat A$ or given $\hat B$) agree with the (left/right) conditional expectations defined as best estimators in Section \ref{s:Q_cond_exp_min}.

In Section \ref{s:QPR give unique Econd}, we show that all $\hat A$ and $\hat B$ Born-compatible quasiprobability representations with ontic space $\Lambda=\sigma(\hat A)\times \sigma(\hat B)$ are completely determined by their associated conditional expectations (Theorem~\ref{thm:StrongerMain}). This result provides a strengthening of  Theorem~\ref{thm:KDunique1_bis} and allows for some further applications that we discuss. 

Sections \ref{s:uniqueKD} and \ref{s:QPR give unique Econd} provide  independent proofs of Theorem \ref{thm:KDunique1_bis} and can be read in any order.

\subsection{The Kirkwood-Dirac distribution}\label{s:KDdef}

The definition of a KD quasiprobability representation of quantum mechanics depends on the choice of two bases in the Hilbert space~\cite{Kirkwood33,Dirac45,arvidssonshukur2024properties,lostaglio2023kirkwood, debievreetal2025a, spriet2025}. Let $\hat A = (\hat A_1, \dots, \hat A_n)$ and $\hat B = (\hat B_1, \dots, \hat B_m)$ be two CSCO.   
We write  $(\varphi_a^{\hat A})_{a\in\sigma(\hat A)}, (\varphi_b^{\hat B})_{b\in \sigma(\hat B)}$ for the corresponding eigenbases:
\begin{equation}
    \hat A_i\varphi^{\hat A}_a=a_i\varphi^{\hat A}_a,\quad \hat B_j\varphi^{\hat B}_b= b_j\varphi^{\hat B}_b. 
\end{equation}
 We denote by $d^{\hat A}_{i}$ and $d^{\hat B}_{j}$, the number of eigenvalues of $\hat A_i$, and of $\hat B_j$, for each $i\in\IntEnt{1}{n}$, respectively $j\in\IntEnt{1}{m}$. For any $p=(p_1, \dots, p_n)\in\mathbb{N}^{n}$, we write
\begin{equation}
    \hat A^{p} := \hat A_1^{p_1}\dots \hat A_n^{p_n},
\end{equation}
and for any $a= (a_1,\dots, a_n)\in\sigma(\hat A)$, we write
\begin{equation}
     a^{p} := a_1^{p_1}\dots a_n^{p_n}.
\end{equation}
We define $\Gamma^{\hat A}$ as
\begin{equation}
    \Gamma^{\hat A} = \IntEnt{0}{d_1^{\hat A}-1} \times\dots\times\IntEnt{0}{d_{n}^{\hat A}-1},
\end{equation}
and $\Gamma^{\hat B}$ is obtained analogously.

We consider, as in Eq.~\eqref{eq:mathcalY}, the maximal commutative algebras 
\begin{equation}
\CompFuncObs{\hat{A}}=\{f(\hat A)|f:\sigma(\hat A)\to \mathbb C\}, \quad \CompFuncObs{\hat{B}}=\{f(\hat B)|f:\sigma(\hat B)\to \mathbb C\}.
\end{equation}
We will use several times the fact that, as we are working on a finite dimensional Hilbert space, any $f(\hat A)\in \CompFuncObs{\hat{A}}$ can be written as a polynomial $f(\hat A)=\sum_{p\in \Gamma^{\hat A}}c_p\hat A^p$, where $(c_p)_{p\in \Gamma^{\hat A}}\subset \mathbb{C}$. This is a direct consequence of Lagrange interpolation. Of course, a similar statement holds for the elements of $\CompFuncObs{\hat{B}}$.

 We have the following lemma: 
\begin{lemma}\label{lem:complementary}
    Let $\hat A,\hat B$ be two $\mathrm{CSCO}$, as above. Consider the statements:
    \begin{itemize}
    \item[(i)] $\forall (a,b)\in \sigma(\hat A)\times \sigma (\hat B), \langle \varphi_a^{\hat A},\varphi_b^{\hat B}\rangle\neq0;$
        \item[(ii)] $\mathrm{span}_{\mathbb C}\CompFuncObs{\hat{A}}\CompFuncObs{\hat{B}}=\mathcal L(\mathcal H)=\mathrm{span}_{\mathbb C}\CompFuncObs{\hat{B}}\CompFuncObs{\hat{A}}$;
        \item[(iii)] $\CompFuncObs{\hat{A}}\cap \CompFuncObs{\hat{B}} =\mathbb C \id{d}.$
    \end{itemize}
Then $(i)\Leftrightarrow (ii) \implies (iii)$.  
\end{lemma}

\begin{proof}
That $(i)$ implies $(ii)$ follows from the fact that $\hat\Pi_{a}^{\hat A}\hat\Pi_{b}^{\hat B}$, with $(a,b)\in\sigma(\hat A)\times \sigma(\hat B)$ form a basis of $\mathcal L(\mathcal H)$ and from the observation that $\hat\Pi_{a}^{\hat A}\in\CompFuncObs{\hat{A}}$, $\hat\Pi_{b}^{\hat B}\in \CompFuncObs{\hat{B}}$. 

To see that $(ii)$ implies $(i)$, let us note that 
\begin{equation}
f(\hat A)g(\hat B)=\sum_{(a,b)\in\sigma(\hat A)\times\sigma(\hat B)}f(a)g(b)\hat\Pi_{a}^{\hat A}\hat\Pi_{b}^{\hat B}.
\end{equation}
Hence
\begin{equation}
    \mathrm{span}_{\mathbb C}\CompFuncObs{\hat{A}}\CompFuncObs{\hat{B}}=\mathrm{span}_{\mathbb C}\{\hat\Pi_{a}^{\hat A}\hat\Pi_{b}^{\hat B} \mid \langle \varphi_a^{\hat A},\varphi_b^{\hat B}\rangle\not=0\},
\end{equation}
which implies the result. 

We now show that $(i)$ implies $(iii)$. Let $f(\hat A)\in\CompFuncObs{\hat{A}}$. We need to show that, if $f(\hat A)\in\CompFuncObs{\hat{B}}$, then it is a a multiple of the identity. But since $\CompFuncObs{\hat{B}} = \left(\CompFuncObs{\hat{B}}\right)'$, this means we need to show that, if $f(\hat A)\in \left(\CompFuncObs{\hat{B}}\right)'$, then it is a multiple of the identity. But $f(\hat A)\in \left(\CompFuncObs{\hat{B}}\right)'$ is equivalent to $[f(\hat A), \hat \Pi_{b}^{\hat B}]=0$, for all $b\in\sigma(\hat B)$, which is equivalent to 
\begin{equation}
    \forall a,a'\in\sigma(\hat A), \forall b\in \sigma(\hat B), \ 0=\langle \varphi_a^{\hat A},[f(\hat A), \hat\Pi_{b}^{\hat B}]\varphi_{a'}^{\hat A}\rangle=(f(a)-f(a'))\langle\varphi_a^{\hat A},\varphi_b^{\hat B}\rangle \langle \varphi_b^{\hat B} ,\varphi_{a'}^{\hat A}\rangle  .
\end{equation}
This shows $f(\hat A)=c\id{d}$, for some $c\in\mathbb C$. 
\end{proof}
We point out that it is easy to construct two bases in dimension $3$ for which condition $(iii)$ of Lemma~\ref{lem:complementary} is fulfilled but condition $(i)$ fails. This shows that those conditions are not equivalent.

\begin{definition}
\label{def:complementary}
    We say two complete sets of commuting observables $\hat A,\hat B$ are \emph{complementary} if they satisfy any one of conditions (i) and (ii) in Lemma~\ref{lem:complementary}.
\end{definition}
The term ``complementary'' is used for a variety of different notions in the literature~\cite{kraus1987, johansen2007, debievre2021, DeB23}.
 In the context of finite dimensional Hilbert space, it is sometimes used as synonymous to ``mutually unbiased'', in other words to $|\langle a|b\rangle|=d^{-1/2}$, for all $a\in\sigma(\hat A), b\in\sigma(\hat B)$. Our definition here is therefore a  relaxation of the latter definition. We refer to~\cite{DeB23} for the link of this definition to various notions of incompatibility.

We now introduce the left and right Kirkwood-Dirac distributions associated with two CSCO $\hat{A}$ and $\hat{B}$, assumed to be complementary.
    Two natural choices of Born-compatible frames for $\hat{A}$ and $\hat{B}$ are
    \begin{equation*}
       S^{\ell}_{a,b}=\hat{\Pi}_a^{\hat{A}}\hat{\Pi}_b^{\hat{B}}  \text{  and  } S^{\mathrm{r}}_{a,b}=\hat{\Pi}_b^{\hat{B}}\hat{\Pi}_a^{\hat{A}}.
    \end{equation*}
    It is then easily shown that the dual frames are
    \begin{equation*}
        T^{\ell}_{a,b}=\frac{1}{|\langle \varphi_a^{\hat{A}},\varphi_b^{\hat{B}}\rangle|^2}\hat{\Pi}_a^{\hat{A}}\hat{\Pi}_b^{\hat{B}} \text{  and  } T^{\mathrm{r}}_{a,b}=\frac{1}{|\langle \varphi_a^{\hat{A}},\varphi_b^{\hat{B}}\rangle|^2}\hat{\Pi}_b^{\hat{B}}\hat{\Pi}_a^{\hat{A}}.
    \end{equation*}
    The associated quasiprobability representations are the left and right KD, representations:
    \begin{align*}
        &\KDDist{\ell}_{a,b}(\hat{\rho})=\langle \varphi_b^{\hat{B}},\varphi_a^{\hat{A}}\rangle\langle\varphi_a^{\hat{A}},\hat{\rho}\varphi_b^{\hat{B}}\rangle, \:\: \KDSymb{\ell}_{a,b}(\hat{X})=\frac{1}{\langle \varphi_a^{\hat{A}},\varphi_b^{\hat{B}}\rangle}\langle \varphi_a^{\hat{A}},\hat{X}\varphi_b^{\hat{B}}\rangle\\
        &\KDDist{\mathrm{r}}_{a,b}(\hat{\rho})=\langle \varphi_a^{\hat{A}},\varphi_b^{\hat{B}}\rangle\langle \varphi_b^{\hat{B}},\hat{\rho}\varphi_a^{\hat{A}}\rangle, \:\: \KDSymb{\mathrm{r}}_{a,b}(\hat{X})=\frac{1}{\langle \varphi_b^{\hat{B}},\varphi_a^{\hat{A}}\rangle}\langle \varphi_b^{\hat{B}},\hat{X}\varphi_a^{\hat{A}}\rangle.
    \end{align*}
     For later reference, we note that, for all $\hat X\in\mathcal L(\mathcal H)$, one has that 
    \begin{equation}\label{eq:KDleftright}
        \overline{\tilde Q_{a,b} ^{\mathrm{KD},\ell}(\hat X)}=\tilde Q_{a,b}^{\mathrm {KD,r}}(\hat X^\dagger)
    \end{equation}
    They are called the left and right Kirkwood-Dirac representations because of the following fact.
    For $f:\sigma(\hat{A})\to \mathbb{C}$ and $g:\sigma(\hat{B})\to \mathbb{C}$, one gets
    \begin{equation}
        \KDSymb{\ell}_{a,b}(f(\hat{A})g(\hat{B}))=f(a)g(b)
    \end{equation}
    so the symbol is well-behaved when $f(\hat{A})$ is ``to the left'', and
    \begin{equation}
        \KDSymb{\mathrm{r}}_{a,b}(g(\hat{B})f(\hat{A}))=g(b)f(a),
    \end{equation}
    the symbol is well-behaved when $f(\hat{A})$ is ``to the right''. 

    We now analyze the right and left KD-conditional expectations of arbitrary observables, given $\hat A$ or $\hat B$.
    \begin{lemma} \label{lem:ABsymmetric} Given two complementary $\mathrm{CSCO}$ $\hat A$ and $\hat B$, we have that, for every $\hat X\in \mathcal L(\mathcal H)$,
     \begin{equation}
    \label{eq: left KD gives left Econd}
        \QPEcond{\KDDist{\ell}}{\hat{X}}{\hat{B}}=\mathbb{E}_{\hat{\rho}}^{\ell}(\hat{X}|\hat{B}),\quad
        \QPEcond{\KDDist{\mathrm{r}}}{\hat{X}}{\hat{B}}=\mathbb{E}_{\hat{\rho}}^{\mathrm{r}}(\hat{X}|\hat{B}).
    \end{equation}
 and
    \begin{equation}
        \QPEcond{\KDDist{\mathrm{\ell}}}{\hat{X}}{\hat{A}}=\mathbb{E}_{\hat{\rho}}^{\mathrm{r}}(\hat{X}|\hat{A}),\quad 
        \QPEcond{\KDDist{\mathrm{r}}}{\hat{X}}{\hat{A}}=\mathbb{E}_{\hat{\rho}}^{\ell}(\hat{X}|\hat{A}).
    \end{equation}
    \end{lemma}
    \begin{proof}
    Following Definition~\ref{def:Qconditionalexpectation_Sigma}, one computes
, for all $\hat{\rho}\in D_{\hat{B}}$ and any $\hat{X}\in \mathcal{L}(\mathcal{H})$,
    \begin{align*}
        \QPEcond{\KDDist{\ell}}{\hat{X}}{\hat{B}}&=\sum_{(a,b)\in \sigma(\hat{A})\times \sigma(\hat{B})}\overline{\frac{1}{\langle\varphi_a^{\hat{A}}, \varphi_b^{\hat{B}} \rangle}\langle \varphi_a^{\hat{A}}, \hat{X}^{\dagger}\varphi_b^{\hat{B}}\rangle}\frac{\langle \varphi_b^{\hat{B}}, \varphi_a^{\hat{A}} \rangle \langle \varphi_a^{\hat{A}}, \hat{\rho}\varphi_b^{\hat{B}} \rangle}{\langle \varphi_b^{\hat{B}}, \hat{\rho}\varphi_b^{\hat{B}} \rangle}\hat{\Pi}_b^{\hat{B}}\\
        &=\sum_{(a,b)\in \sigma(\hat{A})\times \sigma(\hat{B})}\frac{\langle \varphi_b^{\hat{B}}, \hat{X}\varphi_a^{\hat{A}}\rangle \langle \varphi_a^{\hat{A}},\hat{\rho}\varphi_b^{\hat{B}}\rangle}{\langle \varphi_b^{\hat{B}}, \hat{\rho} \varphi_b^{\hat{B}} \rangle}\hat{\Pi}_b^{\hat{B}}\\
        &=\sum_{b\in \sigma(\hat{B})}\frac{\langle \varphi_b^{\hat{B}}, \hat{
        X
        }\hat{\rho}\varphi_b^{\hat{B}}\rangle}{\langle\varphi_b^{\hat{B}}, \hat{\rho}\varphi_b^{\hat{B}} \rangle}\hat{\Pi}_b^{\hat{B}}.
    \end{align*}
    One recognizes the left conditional expectation given in Definition~\ref{def: left/right Econd}, proving the first equality in Eq.~\eqref{eq: left KD gives left Econd}. The other three equalities follow from a similar computation. 
    \end{proof}
    
    In conclusion, both the left and right KD representations are not only $\hat A$ and $\hat B$ compatible, but in addition, the associated conditional expectations given either $\hat A$ or $\hat B$ both can be interpreted as best estimators with respect to either the right or left sesquilinear forms introduced in Section~\ref{s:Q_cond_exp_min}. In particular, their associated conditional expectation satisfy the pull-out formula, see the first point in Theorem \ref{thm:UniqueEcond}. 

    \subsection{Characterization of the Kirkwood-Dirac quasiprobability representations via the pull-out property}\label{s:uniqueKD}

     We shall prove  in this section that the KD representations are the only $\hat A$ and $\hat B$ Born-compatible quasiprobability representations that give rise to $Q$-conditional expectations satisfying the pull-out property: this is the content of Theorem~\ref{thm:KD_unique_pullout}.  
    
The main  technical result we need is contained in the following proposition, which is formulated for the left KD representation. An analogous statement holds for the right KD representation. 

\begin{proposition}\label{prop:KDisbest_new}
     Let $\mathcal H$ be a $d$ dimensional Hilbert space and let $(Q,\tilde Q)$ be a  quasiprobability representation of quantum mechanics on $\mathcal H$ over a space $\Lambda$ with $\left|\Lambda\right|=d^2$. Suppose there exists a $\mathrm{CSCO}$ $\hat B$ so that
     \begin{enumerate}
     \item[(i)] $(Q,\tilde Q)$ is $\hat B$-compatible;
     \item[(ii)] $\forall \hat X\in\mathcal L(\mathcal H),\ \forall f(\hat B)\in \mathcal F_{\mathbb C, \hat B},\ \forall\hat\rho\in D_{\hat B}, \mathbb E^Q_{\hat \rho}(f(\hat B)\hat X|\hat B)=f(\hat B)\mathbb E_{\hat\rho}^Q(\hat X|\hat B)$;
     \end{enumerate}
     Then the following holds. 
   If there exists a second  $\mathrm{CSCO}$ $\hat A$, complementary to $\hat B$, and for which $(Q,\tilde Q)$ is $\hat A$-compatible, then there exists a bijective map $\Phi :\Lambda\to \sigma(\hat A)\times \sigma(\hat B)$ such that for all $(a,b)\in\sigma(\hat A)\times \sigma(\hat B)$
\begin{equation*}
    \forall \hat \rho\in D(\mathcal{H}), Q_{\Phi^{-1}(a,b)}(\hat \rho)= Q^{\mathrm{KD},\ell}_ {a,b}(\hat \rho) \mathrm{ \ and \ } \forall \hat X\in \mathcal{L}(\mathcal{H}), \tilde{Q}_{\Phi^{-1}(a,b)}(\hat X)= \tilde{Q}^{\mathrm{KD},\ell}_ {a,b}(\hat X).
\end{equation*}
\end{proposition}

We illustrate this proposition in Appendix~\ref{app:Wigner_KD} where we compare the conditional expectations arising from the Wigner and the KD  quasiprobability representations for a qutrit system and show that they are different. We show in particular that the conditional expectation associated to the Gross-Wigner representation does not satisfy the pull-out formula.

The pull-out formula ($(ii)$ above) is an intrinsic condition on the $Q$-conditional expectation that can be weakened in the presence of two complementary CSCO, as shown in the following lemma:
\begin{lemma}
\label{lem: A^n pull-out}Suppose that $\hat A$ and $\hat B$ are complementary CSCO on a Hilbert space $\mathcal H$ and that $(Q,\tilde Q)$ is a $\hat B$-compatible quasiprobability representation of quantum mechanics on $\mathcal H$. Then the following statements are equivalent, for all $\hat\rho\in D_{\hat B}$ :
      \begin{itemize}
     \item[(i)] $\forall p\in \Gamma^{\hat A},\ \forall q\in \Gamma^{\hat B}, \mathbb E^Q_{\hat \rho}(\hat B^q\hat A^p|\hat B)=\hat B^q\mathbb E_{\hat\rho}^Q(\hat A^p|\hat B)$;
     \item[(ii)] $\forall \hat X\in\mathcal L(\mathcal H),\ \forall f(\hat B)\in \mathcal F_{\mathbb C, \hat B},\  \mathbb E^Q_{\hat \rho}(f(\hat B)\hat X|\hat B)=f(\hat B)\mathbb E_{\hat\rho}^Q(\hat X|\hat B)$;
      \end{itemize}
  \end{lemma}

\begin{proof}
    We only need to prove that (i) implies (ii). By Lemma \ref{lem:complementary}, one can assume that $\hat X=g(\hat B)h(\hat A)$. By Lagrange interpolation and by linearity one can therefore assume $\hat X=g(\hat B)\hat A^p$, for $p\in \Gamma^{\hat A}$. By Lagrange interpolation again, one has that $f(\hat B)g(\hat B)=\sum_{q\in \Gamma^{\hat B}}c_q\hat B^q$. One has
    \begin{align*}
        \mathbb{E}_{\hat{\rho}}^Q(f(\hat B)\hat X|\hat B)&=\mathbb{E}_{\hat{\rho}}^Q\left(\sum_{q\in \Gamma^{\hat B}}c_q\hat B^qA^p|\hat B\right)=\sum_{q\in \Gamma^{\hat B}}c_q\hat B^q\mathbb{E}_{\hat{\rho}}^Q(\hat A^p|\hat B)\\
        &=f(\hat B)g(\hat B)\mathbb{E}_{\rho}(\hat A^p|\hat B)=f(\hat B)\mathbb{E}_{\rho}^Q(g(\hat B)\hat A^p|\hat{B})\\
        &=f(\hat B)\mathbb{E}_{\hat \rho}^Q(\hat X|\hat B),
    \end{align*}
    where we used (i) twice and the fact that $g(\hat B)\in \CompFuncObs{\hat B}$ can also be expressed as a polynomial. This proves the result. 
\end{proof}

As a direct consequence of Proposition \ref{prop:KDisbest_new} and Lemma \ref{lem: A^n pull-out}, we get the following characterization of the KD distribution, in term of the pull-out formula with respect to powers of the CSCO $\hat A$ and $\hat B$:
\begin{Theorem}\label{thm:KD_unique_pullout}
    Let $\mathcal H$ be a $d$ dimensional Hilbert space and let $(Q,\tilde Q)$ be a  quasi-probability representation of quantum mechanics on $\mathcal H$ over a space $\Lambda$ with $\left|\Lambda\right|=d^2$. Suppose there exist $\hat{A}$ and $\hat B$ two complementary CSCO so that
     \begin{enumerate}
     \item[(i)] $(Q,\tilde Q)$ is $\hat A$ and $\hat B$-compatible;
     \item[(ii)] $\forall p\in \Gamma^{\hat A},\ \forall q\in \Gamma^{\hat{B}},\ \forall\hat\rho\in D_{\hat B}, \mathbb E^Q_{\hat \rho}(\hat B^q\hat A^p|\hat B)=\hat B^q\mathbb E_{\hat\rho}^Q(\hat A^p|\hat B)$;
     \end{enumerate}
     Then there exists a bijective map $\Phi :\Lambda\to \sigma(\hat A)\times \sigma(\hat B)$ such that for all $(a,b)\in\sigma(\hat A)\times \sigma(\hat B)$
     \begin{equation*}
    \forall \hat \rho\in D(\mathcal{H}), Q_{\Phi^{-1}(a,b)}(\hat \rho)= Q^{\mathrm{KD},\ell}_ {a,b}(\hat \rho) \mathrm{ \ and \ } \forall \hat X\in \mathcal{L}(\mathcal{H}), \tilde{Q}_{\Phi^{-1}(a,b)}(\hat X)= \tilde{Q}^{\mathrm{KD},\ell}_ {a,b}(\hat X).
\end{equation*}
\end{Theorem}

We shall now prove Proposition~\ref{prop:KDisbest_new}.

\noindent{{\textit{Proof of Proposition~\ref{prop:KDisbest_new}}}}
As in the proof of Theorem \ref{thm:main_hatY}, conditions (i) and (ii) of the proposition imply that, for all $\lambda\in\Lambda$ and for all $b\in\sigma(\hat B)$, $S_\lambda=S_\lambda\hat\Pi^{\hat B}_{B(\lambda)}$, where we wrote $B(\lambda)=\tilde Q_\lambda(\hat B)$.  Consequently, $S_\lambda \hat \Pi^{\hat B}_{b'}=S_\lambda\hat\Pi^{\hat B}_{B(\lambda)} \hat \Pi^{\hat B}_{b'}=0$, when $b'\not=B(\lambda)$. 
Since $(Q,\tilde Q)$ is both $\hat B$ and $\hat A$ compatible, and since $\hat A$ and $\hat B$ are complementary and complete, one concludes, via Theorem~\ref{thm:main_hatY}, for all $(a,b)\in\sigma(\hat A)\times \sigma(\hat B)$, that

\begin{equation}\label{eq:pibpia}
    \hat \Pi_b^{\hat B}\hat \Pi_a^{\hat A}=\sum_{\lambda\in\Lambda_b^{\hat{B}}}\sum_{\lambda'\in\Lambda_a^{\hat{A}}} S_\lambda S_{\lambda'}^\dagger=\sum_{\lambda\in\Lambda_b^{\hat{B}}}\sum_{\lambda'\in\Lambda_a^{\hat{A}}}  S_\lambda \hat\Pi^{\hat B}_b S_{\lambda'}^\dagger=\sum_{\lambda\in\Lambda_b^{\hat{B}}}\sum_{\lambda'\in\Lambda_a^{\hat{A}}\cap \Lambda_b^{\hat{B}}}  S_\lambda \hat\Pi^{\hat B}_b S_{\lambda'}^\dagger=\sum_{\lambda\in\Lambda_b^{\hat{B}}}\sum_{\lambda'\in\Lambda_a^{\hat{A}}\cap \Lambda_b^{\hat{B}}}  S_\lambda S_{\lambda'}^\dagger.
\end{equation}
Since we assume that $\hat A$ and $\hat B$ are complementary, the left hand side of this equality does not vanish for any $(a,b)\in \sigma(\hat A)\times \sigma(\hat B)$. Hence, there exists $\lambda'\in \Lambda_a^{\hat{A}}\cap\Lambda_b^{\hat{B}}$ so that $\Lambda ^{\hat{A } } _a \cap \Lambda ^{\hat{B } } _b \neq \emptyset$. Consequently, writing as above $A(\lambda)=\tilde Q_\lambda(\hat A)$, the map    
\begin{equation}
    \Phi :=(A,B): \lambda\in\Lambda \mapsto (A(\lambda), B(\lambda))\in\sigma(\hat A)\times \sigma(\hat B)
\end{equation}
is surjective and since $\left|\Lambda\right|=d^2=\left|\sigma(\hat A)\times\sigma(\hat B)\right|$, it is a bijection. It follows that 
\begin{equation}
    \Lambda_a^{\hat{A}}=\bigcup
    _{b\in\llbracket 1, d\rrbracket} \Lambda_a^{\hat{A}}\cap\Lambda_b^{\hat{B}},\quad \left|\Lambda_a^{\hat{A}}\cap\Lambda_b^{\hat{B}}\right|=1, \quad \left| \Lambda_a^{\hat A}\right|=d.
\end{equation}
We can therefore identify $\Lambda$ with $\sigma(\hat A)\times \sigma(\hat B)$, by considering $(A(\lambda), B(\lambda))$ as coordinates on $\Lambda$.
We can now conclude as follows. From Eq.~\eqref{eq:pibpia}, we find
$$
\hat \Pi ^{\hat B}_{b} \hat \Pi^{\hat A} _a = \hat \Pi^{\hat B}_{b}\left( \sum_{\Lambda ^{\hat A}_a \cap \Lambda ^{\hat B}_b } S_{\lambda '} ^{\dagger }  \right) = \hat \Pi ^{\hat B}_b S_{\lambda (a, b ) } ^{\dagger } ,   
$$   
for the unique $\lambda (a, b ) \in \Lambda ^{\hat A} _a \cap \Lambda ^{\hat B} _b $. Since $\lambda (a, b ) \in \Lambda ^{\hat B} _b $, Theorem~\ref{thm:main_hatY}~(iii) implies that $S_{\lambda (a, b ) } = S_{\lambda (a, b ) } \hat \Pi ^{\hat B} _b $. This in turn implies that $S_{\lambda } = \hat \Pi ^{\hat A} _{A(\lambda ) } \hat \Pi ^{\hat B} _{B(\lambda ) }$. 

This concludes the proof with $\Phi(\lambda) = (A(\lambda),B(\lambda))$.
\qed

This first characterization of the left KD distribution relies on an intrisic property of its quantum conditional expectation: among the $\hat A$ and $\hat B$ compatible quasiprobability representations of quantum mechanics, the only one for which the associated quantum conditional expectation satisfies a left pull-out property is the left Kirkwood-Dirac distribution.

We now recall Theorem~\ref{thm:KDunique1_bis}.
\maintheorem*

A first proof of this Theorem can be done as follows: Lemma~\ref{lem:ABsymmetric} proves that $(iii)$ implies $(i)$ and $(iii)$ implies $(ii)$, after direct computations. 

We prove that $(i)$ implies $(iii)$: from Theorem~\ref{thm:main_hatY}, we know that $(i)$ of Theorem~\ref{thm:KDunique1_bis} implies that:
\begin{equation*}
     \forall \hat X\in\mathcal L(\mathcal H),\ \forall f(\hat B)\in \mathcal F_{\mathbb C, \hat B},\  \mathbb E^Q_{\hat \rho}(f(\hat B)\hat X|\hat B)=f(\hat B)\mathbb E_{\hat\rho}^Q(\hat X|\hat B).
\end{equation*}
Thus, as $(Q,\tilde{Q})$ satisfies a left pull out property (given just above) and is $\hat A$ and $\hat B$ compatible, $(Q,\tilde{Q})$ satisfies the hypothesis of Proposition~\ref{prop:KDisbest_new}. This proves $(iii)$. The fact that $(ii)$ implies $(iii)$ also follows from the same reasoning.

  \subsection{Characterization of quasiprobability representations via their quantum conditional expectations}
  \label{s:QPR give unique Econd}

In this subsection, we show that $\hat A$ and $\hat B$ Born-compatible quasiprobability representations are uniquely determined by their associated conditional expectations. To do so, we combine the framework we developed in Section~\ref{s:QCEQP_new} and Sections~\ref{s:KDdef}-\ref{s:uniqueKD} with an idea first developed in~\cite{jordan2026}.

\begin{Theorem}
\label{thm:QPR give unique Econd}
    Let $\hat A$ and $\hat B$ be two complementary CSCO. Let $(R,\tilde{R})$ be a $\hat A$ and $\hat B$-compatible quasiprobability representation of quantum mechanics on $\sigma(\hat A)\times \sigma(\hat B)$. Let $(Q,\tilde{Q})$ be a $\hat A$ and $\hat B$-compatible quasiprobability representation of quantum mechanics on $\Lambda$, with $\left|\Lambda\right|=d^2$. Then, the following statements are equivalent :
    \begin{itemize}
        \item[(i)] for all $\hat \rho\in D_{\hat B}$, for all $p\in\Gamma^{\hat A}$, 
        \begin{equation*}
        \mathbb E_{\hat\rho}^Q(\hat A^p|\hat B)=\mathbb E_{\hat\rho}^{R}(\hat A^p|\hat B);
        \end{equation*}
        \item[(ii)] for all $\hat \rho\in D_{\hat B}$, for all $f(\hat A)\in\CompFuncObs{\hat A}$, $\mathbb E_{\hat\rho}^Q(f(\hat A)|\hat B)=\mathbb E_{\hat\rho}^{R}(f(\hat A)|\hat B)$;
        \item[(iii)] for all $\hat \rho\in D_{\hat B}$, for all $(p,q)\in\Gamma^{\hat A}\times \Gamma^{\hat B}$, 
        \begin{equation*}
            \mathbb E_{\hat\rho}^Q(\hat A^p\hat B^q|\hat B)=\mathbb E_{\hat\rho}^{R}(\hat A^p\hat B^q|\hat B).
        \end{equation*}   
        \item[(iv)] for all $\hat \rho\in D_{\hat B}$, for all $\hat X \in\mathcal{L}(\mathcal H)$, 
        \begin{equation*}
            \mathbb E_{\hat\rho}^Q(\hat X|\hat B)=\mathbb E_{\hat\rho}^{R}(\hat X|\hat B).
        \end{equation*}
        \item[(v)] there exists a  bijective map $\Phi : \Lambda \to \sigma(\hat A)\times\sigma(\hat B)$ such that:
        \begin{equation*}
            \forall (a,b)\in \sigma(\hat A)\times\sigma(\hat B), R_{a,b} = Q_{\Phi^{-1}(a,b)} ;
        \end{equation*}
    \end{itemize}
\end{Theorem}

\begin{proof}
Let us prove first that $(v)\Rightarrow (iv)$. Let $S^R$ and $S^Q$ be the respective frames associated to $(R,\tilde{R})$ and $(Q,\tilde{Q})$. The fact that $R_{a,b}=Q_{\Phi^{-1}(a,b)}$ for all $(a,b)\in \sigma(\hat A)\times \sigma(\hat B)$ implies that $S^R_{a,b}=S^Q_{\Phi^{-1}(a,b)}$ for all $(a,b)\in \sigma(\hat A)\times \sigma(\hat B)$. Moreover from Eq.~\eqref{eq: dual frame}, the symbols also satisfy $\tilde{R}_{a,b}=\tilde{Q}_{\Phi^{-1}(a,b)}$ for all $(a,b)\in \sigma(\hat A)\times \sigma(\hat B)$. Recall from Lemma \ref{lem:compatible} that the fact that $(Q,\tilde{Q})$ is $\hat B$-Born compatible implies that $(\Lambda_{b}^{\hat B})_{b\in \sigma(\hat B)}$ is a partition of $\Lambda$. In particular, we can define the projection $\pi_0:\Lambda\to \sigma(\hat B)$ which verifies that for all $\lambda\in \Lambda$, $\lambda\in \Lambda_{\pi_0(\lambda)}^{\hat B}$. For $b\in \sigma(\hat B)$, we have, from $\hat B$-Born compatibility of $(R,\tilde{R})$,
\begin{align*}
    \hat \Pi_b^{\hat B}&=\sum_{a\in \sigma(\hat A)}S^R_{a,b}=\sum_{a\in \sigma(\hat A)}S^Q_{\Phi^{-1}(a,b)}=\sum_{\lambda\in \Phi^{-1}(\sigma(\hat A)\times \{b\})}S^Q_{\lambda}.
\end{align*}
Moreover, from $\hat B$-compatibility of $(Q,\tilde{Q})$,
\begin{equation*}
    \hat \Pi_b^{\hat B}=\sum_{\lambda\in \Lambda_b^{\hat B}}S^Q_{\lambda}.
\end{equation*}
Taking the Hilbert Schmidt inner product on these two expressions with the dual frame $T^Q$ yields that $\Phi$ is, for any $b\in \sigma(\hat B)$, a bijection between $\Lambda_b^{\hat B}$ and $\sigma(\hat A)\times \{b\}$. In particular, we have that $\pi_0(\Phi^{-1}(a,b))=b$ for all $(a,b)\in \sigma(\hat A)\times \sigma(\hat B)$.
Let $\hat X\in \mathcal{L}(\mathcal{H})$ and $\hat \rho \in D_{\hat B}$. One has
\begin{align*}
    \mathbb{E}^R_{\hat \rho}(\hat X|\hat B)&=\sum_{(a,b)\in \sigma(\hat A)\times \sigma(\hat B)}\overline{\tilde{R}_{a,b}(\hat X^{\dagger})}\frac{R_{a,b}(\hat \rho)}{\langle \varphi_b^{\hat B},\hat \rho\varphi_b^{\hat B}\rangle}\hat \Pi_{b}^{\hat B}\\
    &=\sum_{(a,b)\in \sigma(\hat A)\times \sigma(\hat B)}\overline{\tilde{Q}_{\Phi^{-1}(a,b)}(\hat X^{\dagger})}\frac{Q_{\Phi^{-1}(a,b)}(\hat \rho)}{\langle \varphi_{\pi_0(\Phi^{-1}(a,b))}^{\hat B},\hat \rho\varphi_{\pi_0(\Phi^{-1}(a,b))}^{\hat B}\rangle}\hat \Pi_{\pi_0(\Phi^{-1}(a,b))}^{\hat B}\\
    &=\sum_{\lambda\in \Lambda}\overline{\tilde{Q}_{\lambda}(\hat X^{\dagger})}\frac{Q_{\lambda}(\hat \rho)}{\langle \varphi_{\pi_0(\lambda)}^{\hat B},\hat \rho \varphi_{\pi_0(\lambda)}^{\hat B}\rangle}\hat \Pi_{\pi_0(\lambda)}^{\hat B}\\
    &=\sum_{b\in \sigma(\hat B)}\sum_{\lambda\in \Lambda_b^{\hat B}}\overline{\tilde{Q}_{\lambda}(\hat X^{\dagger})}\frac{Q_{\lambda}(\hat \rho)}{\langle \varphi_b^{\hat B},\hat \rho \varphi_b^{\hat B}\rangle}\hat \Pi_b^{\hat B}\\
    &=\mathbb{E}^Q_{\hat \rho}(\hat X|\hat B),
\end{align*}
proving $(iv)$.

It is clear that $(iv)\Rightarrow (iii)\Rightarrow (i)$. Moreover, as already mentioned, any $f(\hat A)\in\CompFuncObs{\hat A}$ can be written as a polynomial $P_{f}(\hat A)$ and thus, $(i)\Leftrightarrow (ii)$ by linearity. We now prove that $(i)\Rightarrow (v)$ to conclude the proof. We compute, for all  $\hat \rho\in D_{\hat B}$, for all $p\in\Gamma^{\hat A}$:
\begin{eqnarray*}
    \mathbb E_{\hat\rho}^Q(\hat A^p|\hat B) &=& \sum_{b\in\sigma(\hat B)} \sum_{\lambda\in\Lambda_{b}^{\hat B}} \overline{\tilde{Q}_{\lambda}((\hat A^p)^{\dagger})}Q_{\lambda | b}(\hat \rho)\hat \Pi_{b}^{\hat B} = \sum_{b\in\sigma(\hat B)} \sum_{\lambda\in\Lambda_{b}^{\hat B}} \overline{\tilde{Q}_{\lambda}(\hat A^p)}Q_{\lambda | b}(\hat \rho)\hat \Pi_{b}^{\hat B} \\
    &=& \sum_{b\in\sigma(\hat B)} \sum_{\lambda\in\Lambda_{b}^{\hat B}} \overline{\left(\tilde{Q}_{\lambda}(\hat A)\right)^p}Q_{\lambda | b}(\hat \rho)\hat \Pi_{b}^{\hat B} = \sum_{b\in\sigma(\hat B)} \sum_{a\in\sigma(\hat A)} \sum_{\lambda\in\Lambda_{a}^{\hat A}\cap\Lambda_{b}^{\hat B}} \overline{\left(\tilde{Q}_{\lambda}(\hat A)\right)^p}Q_{\lambda | b}(\hat \rho)\hat \Pi_{b}^{\hat B} \\
    &=& \sum_{b\in\sigma(\hat B)} \sum_{a\in\sigma(\hat A)} \sum_{\lambda\in\Lambda_{a}^{\hat A}\cap\Lambda_{b}^{\hat B}} \overline{a^p}Q_{\lambda | b}(\hat \rho)\hat \Pi_{b}^{\hat B} = \sum_{b\in\sigma(\hat B)} \sum_{a\in\sigma(\hat A)} a^{p} \sum_{\lambda\in\Lambda_{a}^{\hat A}\cap\Lambda_{b}^{\hat B}} Q_{\lambda | b}(\hat \rho)\hat \Pi_{b}^{\hat B}.
\end{eqnarray*}
where $\tilde{Q}_{\lambda}(\hat A^p) = \tilde{Q}_{\lambda}(\hat A)^p$ follows directly from Lemma~\ref{lem:compatible}. We also have that :
\begin{eqnarray*}
    \mathbb E_{\hat\rho}^R(\hat A^p|\hat B) &=& \sum_{b\in\sigma(\hat B)} \sum_{a\in\sigma(\hat A)} a^{p} R_{(a,b)| b}(\hat \rho)\hat \Pi_{b}^{\hat B}.
\end{eqnarray*}

Thus, for all polynomial $P$ with for all $i\in\IntEnt{1}{n} \deg_{i}(P) \leqslant d_{i}^{\hat A}-1$ and all $\hat\rho\in D_{\hat B}$, we have that :

\begin{equation*}
    \forall b\in\sigma(\hat B), \  \sum_{a\in\sigma(\hat A)} P(a) \sum_{\lambda\in\Lambda_{a}^{\hat A}\cap\Lambda_{b}^{\hat B}} Q_{\lambda | b}(\hat \rho) = \sum_{a\in\sigma(\hat A)} P(a) R_{(a,b)| b}(\hat \rho).
\end{equation*}
We fix $a'\in\sigma(\hat A)$. By Lagrange interpolation, there exists $P_{a'}$ such that $P_{a'}(a) = \delta_{a,a'}$ for all $a\in\sigma(\hat A)$ and  $\deg_{i}(P) \leqslant d_{i}^{\hat A}-1$ for all $i\in\IntEnt{1}{n}$. Thus, we have that, for all $\hat \rho\in D_{\hat B}$:
\begin{equation*}
    \forall b\in\sigma(\hat B), \sum_{\lambda\in\Lambda_{a'}^{\hat A}\cap\Lambda_{b}^{\hat B}} Q_{\lambda | b}(\hat \rho)  = R_{(a',b)| b}(\hat \rho).
\end{equation*}
Using Eq.~\eqref{eq:defQPRFrame} and Eq.~\eqref{eq:DefCondQPR} multiplied by $\Tr(\hat\rho \hat \Pi_b^{\hat B})$,
we obtain that, for all $(a,b)\in\sigma(\hat A)\times\sigma(\hat B)$ and for all $\hat \rho\in D_{\hat B}$
\begin{equation*}
    \mathrm{Tr}\left(\hat\rho \left(\sum_{\lambda\in\Lambda_{a}^{\hat A}\cap\Lambda_{b}^{\hat B}} S^{Q}_{\lambda}- S^{R}_{a,b}\right)^{\dagger} \right) = 0.
\end{equation*}
Consequently, we conclude that, for all $(a,b)\in\sigma(\hat A)\times\sigma(\hat B)$:
\begin{equation}\label{eq:Equality_frames}
    \sum_{\lambda\in\Lambda_{a}^{\hat A}\cap\Lambda_{b}^{\hat B}} S^{Q}_{\lambda}= S^{R}_{a,b}
\end{equation}
As $S^{R}_{a,b}\neq 0$, we have that for all $(a,b)\in\sigma(\hat A)\times\sigma(\hat B)$, $\Lambda_{a}^{\hat A}\cap\Lambda_{b}^{\hat B}\neq \emptyset$. Consequently, the map :
\begin{equation*}
    \Phi : \left\{\begin{array}{rcl}
         \Lambda& \to &\sigma(\hat A)\times\sigma(\hat B) \\
         \lambda& \mapsto & (Q_{\lambda}(\hat{A}),Q_{\lambda}(\hat{B}))
    \end{array}\right.
\end{equation*}
is surjective and as $\left|\Lambda\right| = \left|\sigma(\hat A)\times\sigma(\hat B)\right| = d^2$, $\Phi$ is a bijection. In particular, $\left|\Lambda_{a}^{\hat A}\cap\Lambda_{b}^{\hat B}\right|=1$ and this element is $\Phi^{-1}(a,b)$ for all $(a,b)\in \sigma(\hat A)\times\sigma(\hat B)$. We finally obtain, by Eq.~\eqref{eq:Equality_frames}, that for all $(a,b)\in\sigma(\hat A)\times\sigma(\hat B)$:
\begin{equation}
    S^{Q}_{\Phi^{-1}(a,b)} = S^{R}_{a,b}
\end{equation}
proving $(v)$. 
\end{proof}

This theorem shows that a quasiprobability representation of quantum mechanics on a Hilbert space $\mathcal H$ that is Born-compatible for two complementary CSCO  $\hat A$ and $\hat B$, is completely determined by its associated quantum conditional expectations of $\hat A ^p$ given $\hat B$ with $p\in\Gamma^{\hat A}$.  Combining Theorem \ref{thm:QPR give unique Econd} with Lemma~\ref{lem:ABsymmetric}, and setting  $(R,
\tilde R)=(Q^{\mathrm{KD},\ell}, \tilde Q^{\mathrm{KD},\ell})$, we immediately obtain the following theorem:

\begin{Theorem}\label{thm:StrongerMain}
Let $\hat A$ and $\hat B$ be complementary $\textrm{CSCO}$. Let $(Q,\tilde Q)$ be an $\hat A$ and $\hat B$-compatible quasiprobability representation of quantum mechanics defined on a set $\Lambda$ ($\left|\Lambda\right|=d^2)$. Then the following are equivalent:
\begin{itemize}
\item[(i)] $\forall \hat \rho\in D_{\hat B},\forall p\in\Gamma^{\hat A}, \quad \mathbb E^Q_{\hat \rho}(\hat A^p|\hat B)=\mathbb E^\ell_{\hat \rho}(\hat A^p|\hat B)$ 

\item[(ii)] $\forall \hat \rho\in D_{\hat A},\forall p\in\Gamma^{\hat B}, \quad  \mathbb E^Q_{\hat \rho}(\hat B^p|\hat A)=\mathbb E^{\mathrm{r}}_{\hat \rho}(\hat B^p|\hat A)$

\item[(iii)]   There exists a bijective map $\Phi :\Lambda\to \sigma(\hat A)\times \sigma(\hat B)$ such that for all $(a,b)\in\sigma(\hat A)\times \sigma(\hat B)$
\begin{equation*}
    \forall \hat \rho\in D(\mathcal{H}), Q_{\Phi^{-1}(a,b)}(\hat \rho)= Q^{\mathrm{KD},\ell}_ {a,b}(\hat \rho) \mathrm{ \ and \ } \forall \hat X\in \mathcal{L}(\mathcal{H}), \tilde{Q}_{\Phi^{-1}(a,b)}(\hat X)= \tilde{Q}^{\mathrm{KD},\ell}_ {a,b}(\hat X).
\end{equation*}
\end{itemize}
\end{Theorem}

Theorem~\ref{thm:StrongerMain} provides a direct proof of Theorem~\ref{thm:KDunique1_bis} as it is a strenghtened version thereof: indeed, Theorem~\ref{thm:StrongerMain} relies on the hypothesis that the conditional expectations coincides on powers of $\hat A$ or on powers of $\hat B$ whereas Theorem~\ref{thm:KDunique1_bis} relies on the equality of the conditional expectation on all operators in $\mathcal{L}(\mathcal{H}).$

The left Kirkwood-Dirac distribution is therefore the only quasiprobability representation of quantum mechanics that is $\hat A$ and $\hat B$ compatible and for which the quantum conditional expectation of $\hat A^p$ given $\hat B$ is given by the ``left weak-value'' formalism, meaning by the values of $\left(\frac{\mathrm{Tr}(\hat \Pi_{b}^{\hat B}\hat A^p\hat \rho)}{\mathrm{Tr}(\hat \Pi_{b}^{\hat B}\rho)}\right)_{b\in\sigma(\hat B),n\in\IntEnt{0}{d-1}}$. For the right Kirkwood-Dirac representation, it is characterized by the ``right weak-value formalism'', meaning by the values of $\left(\frac{\mathrm{Tr}(\hat \rho\hat A^p\hat \Pi_{b}^{\hat B})}{\mathrm{Tr}(\hat \Pi_{b}^{\hat B}\rho)}\right)_{b\in\sigma(\hat B),n\in\IntEnt{0}{d-1}}$.

Theorem~\ref{thm:QPR give unique Econd} can further be used to characterize any fixed $\hat A$ and $\hat B$ Born-compatible quasiprobability representation of quantum mechanics. If one briefly goes to the infinite dimensional setting for illustrative purposes, the Wigner representation is then uniquely determined by the associated conditional expectations of $\hat P^n$ given $\hat Q$ that will not coincide with the ones arising from the weak value formalism. We illustrate this fact in Appendix \ref{app:Wigner_KD} where we compare the KD and Wigner conditional expectations of $\hat P^n$ knowing $\hat Q$, for $n=1,2$, showing they are different when $n=2$. This also means that the conditional expectation given by the Wigner quasiprobability representation cannot be the one defined by minimization in Section~\ref{s:Q_cond_exp_min}.

\section{Quantum conditional expectation and parameter estimation}\label{s:Q_cond_exp_int}
In this section, we explore properties of the quantum conditional expectations $\mathbb E^{\ell/\mathrm r}_{\hat \rho}(\hat X|\hat Y)$ introduced in Section~\ref{s:Q_cond_exp_min} in relation to the following well-known parameter estimation problem. Let $\hat \rho$ be any quantum state and consider the state
\begin{equation}\label{eq:rho_theta}
\hat{\rho }_{\hat{X } } (\theta ):= e^{ - i \theta \hat{X } } \hat{\rho } e^{i \theta \hat{X }}
\end{equation}
evolved under the unitary flow generated by some observable $\hat X$. The problem is then to estimate the ``phase'' $\theta$ from  the measurement of an appropriately chosen observable  $\hat Y$ in $\hat \rho_{\hat X}(\theta)$~\cite{Giovaneti2004,hofmann2011,Giovanetti2011,Dressel_Jordan_2012_ImaginaryPartWeakValue,Smith2024}.

We first show that the variance of the imaginary part of the above quantum conditional expectation is equal to the classical Fisher information associated with this phase estimation problem (Proposition~\ref{prop:Fisher_imaginary}). This relation was previously established and elaborated upon in~\cite{hofmann2011, Dressel_Jordan_2012_ImaginaryPartWeakValue} within the context of weak value physics. We then apply the results of Section~\ref{s:Q_cond_exp_mean_var}, and in particular the law of iterated expectations and the law of total variance to relate the mean and the variance of $\hat X$ in the state  $\hat\rho_{\hat X}(\theta)$ to those of its best estimator $\mathbb E^{\ell/\mathrm r}_{\hat \rho_{\hat X}(\theta)}(\hat X|\hat Y)$ and of its conditional variance $\Delta^{\ell/\mathrm r,2}_{\hat\rho_X(\theta)}(\hat X|\hat Y)$ in the same state.
This sharpens a result of~\cite{hofmann2011} and corroborates the idea that $\mathbb E^{\ell/\mathrm r}_{\hat \rho}(\hat X|\hat Y)$ is a  ``best estimator'' of $\hat X$.

These developments allow us to introduce and then study, in
Section \ref{s:phase insensitivity}, the notion of phase insensitive state, which is a state for which the above Fisher information vanishes. We then show that, modulo some technical assumptions, KD-real states are phase insensitive (Proposition \ref{prop: phase insensitivity for KD real}). More precisely, we will observe that if the metrological protocol described above  involves an observable $\hat X$ and and a state $\hat \rho$ that are KD-real, with $\hat Y=\hat B$ or $\hat Y=\hat A$, then their Fisher information $\mathrm{I_F}(\hat Y;0)$ vanishes. Such  states and observables are therefore not useful for this particular phase estimation protocol. But, as noted in Section \ref{s:uncertainty}, they can provide simultaneous information on the Born probability distributions of non-commuting observables $\hat X,\hat Y$ through a weak measurement protocol.  This result further suggests that, if $\hat\rho$ and $\hat X$ are KD real, and if the associated quantum Fisher information is non-zero, then the logarithmic derivative $\hat L$ of $\hat \rho$ cannot be KD real.  We show this to be the case when $\hat \rho$ is pure and provide two families of examples where $\hat \rho$ is mixed and $\hat L$ is indeed not KD real.

\subsection{The imaginary part of \texorpdfstring{$\mathbb E^{\ell}_{\hat\rho}(\hat X|\hat Y)$}{E[...]}: Fisher information}\label{s:Fisherinformation}

If the observable $\hat{Y } $ is measured when the system is in the state $\hat \rho_{\hat X}(\theta)$, the outcome $y\in\sigma(\hat Y)$ is observed with probability  
\begin{equation} \label{eq:def_pytheta}
p(y; \theta ) := {\rm Tr } (\hat{\Pi }^{\hat{Y } } _y \hat{\rho }_{\hat{X } } (\theta) ).
\end{equation}
The Fisher information~\cite{fisher1925,braunsteincaves1994,helstrom1976,holevo1982}  of this $\theta$-dependent probability is  defined as the expected value of $(\partial _{\theta } \ln (p(y ; \theta ) ))^2 $:   
\begin{equation} \label{eq:FI1}   
\mathrm{I}_{\mathrm{F}} (\hat Y;\theta )   
:= \sum _y (\partial _{\theta } \ln (p(y ; \theta ))  )^2 p(y ; \theta ).     
\end{equation} 
Note that this Fisher information depends not only on $\hat Y$ and $\theta$, but also on $\hat \rho$ and $\hat X$, although we do not indicate this dependence in the notation. The Fisher information derives its importance principally from
the famous Cram\`er-Rao inequality which states that the variance $\mathbb V(\tilde\theta)$ of any unbiased estimator $\tilde{\theta } (\hat Y) $ of $\theta$ is bounded from below by the inverse of the Fisher information: 
\begin{equation}\label{eq:CRbound}
    \mathbb V(\tilde\theta)\geq (\mathrm{I_F}(\hat Y;\theta))^{-1}.
\end{equation}
The Fisher information $\mathrm{I}_{\mathrm{F}}(\hat Y;\theta)$ is related to the (imaginary part of) the conditional expectation $\mathbb E_{\hat\rho}^{\ell/\mathrm r}(\hat X|\hat Y)$, as follows~\cite{hofmann2011, Dressel_Jordan_2012_ImaginaryPartWeakValue}: 
\begin{proposition}\label{prop:Fisher_imaginary}  Let $\hat Y$ be a $\mathrm{CSCO}$, $\hat X^\dagger=\hat X\in\mathcal L(\mathcal H)$ and suppose that $\hat \rho_{\hat X}(\theta)\in D_{\hat Y}$. Then
    \begin{equation} \label{eq:FI}     
\mathrm{I}_{\mathrm{F}} (\hat Y;\theta ) = 
4 {\rm Tr } \left( \left( {\rm Im } \, \mathbb{E } ^{\ell/\mathrm r }_{\hat{\rho }_{\hat{X } } (\theta ) } (\hat{X } | \hat{Y } )\right) ^2 \hat{\rho }_{\hat{X } } (\theta ) \right) .      
\end{equation}
\end{proposition}
\begin{proof}
A straightforward computation yields~\cite{Dressel_Jordan_2012_ImaginaryPartWeakValue,hofmann2011}
\begin{equation}\label{eq:thetderivprob}
\partial _{\theta }\ln p (y ; \theta ) = \frac{- i {\rm Tr } ( \hat{\Pi }_y^{\hat Y} [\hat{X } , \hat{\rho }_{\hat{X } } (\theta )  ] ) }{{\rm Tr } (\hat{\Pi }_y^{\hat Y} \hat{\rho }_{\hat{X } } (\theta ) ) } = 2 {\rm Im } \left(\frac{{\rm Tr }  (\hat{\Pi }_y^{\hat Y} \hat{X } \hat{\rho }_{\hat{X } } (\theta ) ) }{{\rm Tr } (\hat{\Pi }_y^{\hat Y} \hat{\rho }_{\hat{X } } (\theta ) ) }\right),
\end{equation} 
where we used the self-adjointness of $\hat{X }$. 
This is the imaginary part of the coefficient of $\hat{\Pi }_y $ in $\mathbb{E }_{\hat{\rho}_{\hat X}(\theta) }^{\ell/\mathrm r} (\hat{X } | \hat{Y } ) $, see Definition~\ref{def: left/right Econd}.    
This proves Eq.~\eqref{eq:FI}. 
\end{proof}
Note that  $\mathbb E_{\hat\rho}(\mathrm{Im}(\mathbb E_{\hat\rho}^{\ell/r}(\hat X|\hat Y)))=0$. Hence  
Eq.~\eqref{eq:FI} states that the classical Fisher information of the quantum mechanical Born probabilities $p(y;\theta)$ equals (up to a factor $4$) the variance of the imaginary part of the (left/right) quantum conditional  expectation. This observation provides an operational meaning to the imaginary part of the conditional expectation $\mathbb E^{\ell}_{\hat \rho}(\hat X|\hat Y)$. Since in classical probability theory the conditional expectation of a real random variable is always real, the presence of an imaginary part to $\mathbb E_{\hat\rho}^{\ell}(\hat X|\hat Y)$ can be considered as a typical quantum feature.  This raises the question precisely when $\mathrm{Im} \mathbb{E}^{\ell}_{\hat{\rho}}(\hat{X}|\hat{Y})$ vanishes, depending on $\hat \rho,\hat X$, and $\hat Y$. From Eq.~\eqref{eq:left_cond_exp} one obtains that
\begin{equation}\label{eq:imag_cond_exp}
   \mathrm{Im} \mathbb{E}^{\ell}_{\hat{\rho}}(\hat{X}|\hat{Y})=-\frac{1}{2}\sum_{y\in \sigma(\hat{Y})}\frac{\langle \varphi_y^{\hat{Y}},[\hat{X},\hat{\rho}]\varphi_y^{\hat{Y}}\rangle}{\langle \varphi_y^{\hat{Y}},\hat{\rho}\varphi_y^{\hat{Y}}\rangle}\hat{\Pi}_y^{\hat{Y}}=-\mathrm{Im} \mathbb{E}^{\mathrm r}_{\hat{\rho}}(\hat{X}|\hat{Y}),
\end{equation}
Hence, if
\begin{equation}
\label{eq:commutators}
    [\hat\rho, \hat X]=0,\quad \mathrm{or}\ [\hat \rho, \hat \Pi^{\hat Y}_y]=0, \quad \mathrm{or}\
[\hat \Pi_y^{\hat Y},\hat X]=0, 
\end{equation}
then $\textrm{Im}\mathbb E^\mathrm{\ell/\textrm r}(\hat X|\hat Y)=0$. This, in turn, implies that the Fisher information $\mathrm{I_F}(\hat Y; 0)$ vanishes. One may therefore conclude that the absence of incompatibility between $\hat\rho, \hat X$ and $\hat Y$ -- which is a form of ``classicality'' --  is a sufficient condition for the absence of an imaginary part to the quantum conditional expectation and hence of the vanishing of the Fisher information $\mathrm{I_F}(\hat Y; 0)$. We will see below that it is not a  necessary condition, as the Fisher information can vanish in cases where the commutators of Eq.~\eqref{eq:commutators} are all non-zero. An example is given in Eq.~\eqref{example:anomalous_weak}. Generally, self-adjointness of the conditional expectation $\mathbb E_{\hat\rho}^{\ell}(\hat X|\hat Y)$ does not preclude the manifestation of other nonclassical features associated to the triplet $(\hat\rho, \hat X, \hat Y)$: as we will see in the next section, anomalous values of  $\mathbb E_{\hat\rho}^{\ell}(\hat X|\hat Y=y)$ point to the possibility of weak value amplification, which is a typical quantum phenomenon as well. 

In order to address the question under what circumstances the Fisher information vanishes, we introduce the notion of phase insensitive states:
\begin{definition}\label{def:phase_invariant}
   Given an observable $\hat X$ and a  $\mathrm{CSCO}$ $\hat Y$, we say that a state $\hat \rho\in D_{\hat Y}$ is phase insensitive for a measurement of $\hat Y$ if $\mathrm{I}_{\mathrm{F}} (\hat Y;0 )=0$.
\end{definition}  
Phase insensitivity of a state $\hat \rho$ depends both on the  generator $\hat X$ associated to the phase and on the observable $\hat Y$. We will simply say ``$\hat\rho$ is phase insensitive'' when the choice of $\hat X$ and $\hat Y$ is clear from the context; otherwise, we will say the triplet $(\hat \rho, \hat X, \hat Y)$ is phase-insensitive.  The operational meaning of phase insensitivity is discussed in some detail in Appendix~\ref{app:vanishingFI}. We recall there that the Fisher information can be viewed as the squared relative rate of change of the probabilities $p(y;\theta)$, when the phase $\theta$ varies by a small amount $\delta\theta$.
Comparing this quantity to the mean squared
relative error $\|\delta p\|^2_{\mathrm F}$ of the experimentally determined probabilities, one obtains an estimate on the minimal variation $\delta \theta_{\mathrm min}$ that is  experimentally detectable:
$$
\delta\theta^2\geq \delta\theta_{\mathrm min}^2:=\frac{\parallel \delta p\|^2_{\mathrm F}}{I_{\mathrm F}(\hat Y;\theta)}.
$$ 
For small, and a fortiori, for vanishing Fisher information, this minimal phase variation becomes very large, thereby justifying the definition of phase insensitivity. We refer to Appendix~\ref{app:vanishingFI} for more details. We will show how to identify  phase insensitive triplets $(\hat \rho, \hat X, \hat Y)$ using the notion of  KD-reality in Section~\ref{s:KD_phase_insensitive}, where we will also further explore the link with the quantum Fisher information.

\subsection{An additive uncertainty principle}\label{s:uncertainty}
As a direct application of the results of Section~\ref{s:Q_cond_exp_min}, and in particular of the quantum law of total variance derived there, we can now relate the variance of $\hat X$ to the variances of the real and imaginary parts of its best estimator $\mathbb E_{\hat\rho_{\hat X}(\theta)}^{\ell/r}(\hat X|\hat Y)$  and to the expected value of its conditional variance $\Delta^{\ell/\mathrm r,2}_{\hat \rho_{\hat X}(\theta)}(\hat X|\hat Y)$. In what follows, for $Z\in\mathcal L(\mathcal H)$, we write $|Z|=\sqrt{Z^\dagger Z}$. Since $\Delta_{\hat \rho_{\hat X}(\theta)}^2\hat X=\Delta_{\hat{\rho}}^2\hat X$ for any value of $\theta$, it follows from  the law of total variance in Proposition \ref{prop:variances_condexp1} combined with Eq~\eqref{eq: variance real and imaginary parts} and Proposition \ref{prop:Fisher_imaginary} that, for any $\hat X=\hat X^\dagger\in \mathcal L(\mathcal H)$ and $\mathrm{CSCO}$ $\hat Y$,
\begin{equation}\label{eq:Delta_IF_mixed_left}
    \Delta^{2}_{\hat \rho} \hat X=\Delta_{\hat \rho_{X}(\theta)}^{2}\mathbb E_{\hat \rho_{X}(\theta)}^{\mathrm{sa}}(\hat X | \hat Y)+\frac14 \mathrm{I}_{\mathrm{F}} (\hat Y;\theta )+ \Tr(\hat\rho_{\hat X}(\theta)\mid \hat X-\mathbb E^{\ell}_{\hat \rho_{X}(\theta)}(\hat X | \hat Y)\mid^2),
\end{equation}   
and
\begin{equation}\label{eq:Delta_IF_mixed_right}
    \Delta^{2}_{\hat \rho} \hat X=\Delta_{\hat \rho_{X}(\theta)}^{2}\mathbb E_{\hat \rho_{X}(\theta)}^{\mathrm{sa}}(\hat X | \hat Y)+\frac14 \mathrm{I}_{\mathrm{F}} (\hat Y;\theta )+ \Tr(\hat\rho_{\hat X}(\theta)\mid \hat X-\mathbb E^{\mathrm r}_{\hat \rho_{X}(\theta)}(\hat X | \hat Y)^\dagger\mid^2),
\end{equation}
where we recall the notation $\mathbb E^{\mathrm{sa}}_{\hat\rho_{\hat X}(\theta)}(\hat X|\hat Y):=\mathbb E^{\ell/\mathrm r,\mathrm{sa}}_{\hat\rho_{\hat X}(\theta)}(\hat X|\hat Y)$. 
 As we saw in Section~\ref{s:Q_cond_exp_mean_var}, the error term (last term) can be understood as the expected value of the conditional variance of $\hat X$. When $\hat\rho$ is pure, this error term vanishes and one finds
    \begin{equation}\label{eq:Delta_IF_pure}
    \Delta^{2}_{\hat \rho} \hat X=\Delta_{\hat \rho_{X}(\theta)}^{2}\mathbb E_{\hat \rho_{X}(\theta)}^{\mathrm{sa}}(\hat X | \hat Y)+\frac14\mathrm{I}_{\mathrm{F}} (\hat Y;\theta ).
\end{equation}
We therefore have, for all $\hat\rho$
\begin{equation}\label{eq:Delta_IF_uncrel}
     \Delta^{2}_{\hat \rho} \hat X\geq\Delta_{\hat \rho_{X}(\theta)}^{2}\mathbb E_{\hat \rho_{X}(\theta)}^{\mathrm{sa}}(\hat X | \hat Y)+\frac14 \mathrm{I}_{\mathrm{F}} (\hat Y;\theta ).
\end{equation}

Equation~\eqref{eq:Delta_IF_pure} and Eq.~\eqref{eq:Delta_IF_uncrel}  were first derived in~\cite{hofmann2011} (at $\theta=0$), using a formulation  in terms of weak values instead of in terms of conditional expectations. The identification of the error term in Eq.~\eqref{eq:Delta_IF_mixed_left}-\eqref{eq:Delta_IF_mixed_right}, the concurrent link with the quantum law of total variance, and the interpretation of the error term as the expected value of a quantum conditional variance are, to the best of our knowledge, new here.

It is  pointed out in~\cite{hofmann2011} that one can think of Eq.~\eqref{eq:Delta_IF_pure} (at $\theta=0$) as an (additive) uncertainty relation. Indeed, both terms in the right hand side depend on $\hat Y$, whereas the left hand side does not. So if one term is small, the other must be large, and vice versa; this result has the following interpretation.  If, for a given $\hat \rho$ and $\hat X$, $\hat Y$ can be chosen so that the Fisher information is large, meaning close to (four times) the variance of $\hat X$, then, according to the Cramer-Rao bound, the variance on the estimation of $\theta$ can be made small, meaning one can obtain a good estimate on the phase $\theta$ from measurements of $\hat Y$ in the state $\hat\rho_{\hat X}(\theta)$.  
In that case,   however, $\mathbb E_{\hat \rho_{\hat X}(\theta)}^{\mathrm{sa}}(\hat X|\hat Y)$ will have a small variance that  will be far from the variance of $\hat X$. So then the fluctuations in the conditional expectation $\mathbb E_{\hat \rho_{\hat X}(\theta)}^{\mathrm{sa}}(\hat X|\hat Y)$ do not provide much information on the probability distribution of $\hat X$ for the state $\hat\rho$. For pure states $\hat \rho$, the extreme case occurs when $\hat Y$ is taken to be equal to the symmetric logarithmic derivative $\hat L$ of $\hat \rho_{\hat X(\theta)}$ at $\theta=0$, because then the Fisher information is maximal.  This maximum Fisher information, which is the quantum Fisher information $I_{\rm QF } (0) $, then equals $4 \Delta _{\hat \rho } ^2 \hat X $ (See Appendix~\ref{app:Fisher} for details).  On the other hand, if $\hat Y$ is such that the Fisher information vanishes or is small, one can obtain only little information on the phase $\theta$ from measurements of $\hat Y$ in the state $\hat\rho_{\hat X}(\theta)$.  However, in that case,  the variance of $\mathbb E_{\hat{\rho }_ {X}(\theta)}^{\mathrm{sa}}(\hat X|\hat Y)$ will be close or equal to  the one of $\hat X$. Specifically, when the Fisher information vanishes and $\hat\rho$ is pure, the former  is equal to the latter (see Eq.~\eqref{eq:Delta_IF_pure}).  In conclusion, if $\hat\rho$ is pure, we find that 
\begin{equation}
    \min_{\hat Y} \Delta_{\hat \rho}^{2}\mathbb E_{\hat \rho}^{\mathrm{sa}}(\hat X | \hat Y)= \Delta_{\hat \rho}^{2}\mathbb E_{\hat \rho}^{\mathrm{sa}}(\hat X | \hat L)=\Delta^{2}_{\hat \rho} \hat X -\frac14I_{\mathrm{QF}}(0)=0,\quad  \Delta^{2}_{\hat \rho} \hat X= \max_{\hat Y} \Delta_{\hat \rho}^{2}\mathbb E_{\hat \rho}^{\mathrm{sa}}(\hat X | \hat Y),
\end{equation}
where the maximum on the right is reached on phase insensitive triplets $(\hat \rho, \hat X, \hat Y)$. When $\hat \rho$ is mixed, one has
\begin{equation}
    \min_{\hat Y} \Delta_{\hat \rho}^{2}\mathbb E_{\hat \rho}^{\mathrm{sa}}(\hat X | \hat Y) \leq \Delta_{\hat \rho}^{2}\mathbb E_{\hat \rho}^{\mathrm{sa}}(\hat X | \hat L)\leq\Delta^{2}_{\hat \rho} \hat X -\frac14I_{\mathrm{QF}}(0)\leq \Delta^{2}_{\hat \rho} \hat X= \max_{\hat Y} \Delta_{\hat \rho}^{2}\mathbb E_{\hat \rho}^{\mathrm{sa}}(\hat X | \hat Y),
\end{equation}
where the maximum is reached when $\hat Y=\hat X$, for example.

We note that the real and imaginary parts of the conditional expectations $\mathbb E_{\hat \rho}^{\ell, \mathrm r}(\hat X | \hat Y)$ can be experimentally accessed via a weak measurement procedure that we describe in Appendix~\ref{app:weakvalues}.
 We conclude that there is therefore a necessary trade-off between obtaining information on $\theta$ from $p(y;\theta)$, for which a large Fisher information is needed, or on the probability distribution of $\hat X$ in $\hat \rho$  through the weak measurement procedure of $\hat X$, conditioned on $\hat Y$, for which the Fisher information should optimally vanish, so that the triplet $(\hat\rho, \hat X, \hat Y)$ is phase insensitive. Phase insensitivity will be related to KD reality in Section~\ref{s:KD_phase_insensitive}.

We finally point out that, for all $\hat \rho$ and $\theta$ 
\begin{equation}\label{eq:IF_upperbound}
  0\leq  \mathrm{I}_{\mathrm{F}} (\hat Y;\theta )
\leq   I_{\mathrm{QF}}(\theta)\leq 4\Delta^{\ell,2}_{\hat \rho} \hat X\leq (x_{\max}-x_{\min})^2 ,  
\end{equation}   
where $x_{\max}$, respectively $x_{\min}$, is the largest, respectively smallest, eigenvalue of $\hat X$ and we used {\it Popoviciu's inequality}, which states that the variance of a random variable taking values in some bounded interval $[a , b ] $ is bounded from above by $\frac{1 }{4 } (b - a )^2 . $ See for example~\cite{Bhatia01042000}, where this inequality is derived as a consequence of a stronger inequality which states that 
$$  
\frac{1}{4}\mathrm{I}_{\mathrm{F}} (\hat Y;\theta ) \leq \Delta^{\ell,2}_{\hat \rho } \hat X\leq (x_{\max } - \mathbb{E }_{\hat{\rho } } (\hat{X } ) )(\mathbb{E }_{\hat{\rho } }(\hat{X }) - x_{\min } ) ,   
$$   
but for which one has to know the expected value of $\hat{X } . $   

\subsection{An application: KD reality implies  phase insensitivity}\label{s:KD_phase_insensitive}
\label{s:phase insensitivity}

 In this subsection, we exploit the link between the KD representations of quantum mechanics and the conditional expectations $\mathbb E^{\mathrm \ell/\mathrm r}_{\hat \rho}(\hat X|\hat Y)$ (Theorem \ref{thm:KDunique1_bis}) with $\hat Y=\hat A$ or $\hat Y=\hat B$ in order to construct phase-insensitive triplets using the KD-real sector of a KD representation of quantum mechanics.

 We have seen that, if any two of $\hat \rho, \hat X$, and $\hat Y$ are compatible (commute), then the associated Fisher information vanishes and hence $\hat\rho$ is phase insensitive.
 However, these conditions are very strong, and in fact rather trivial. For example, if $[\hat\rho, \hat X]=0$, then $\hat\rho_{\hat X}(\theta)=\hat\rho$  is independent of $\theta$ and if $[\hat X,\hat\Pi_y^{\hat Y}]=0$, the two observables involved are compatible. We will now show that non-trivial examples of phase insensitivity arise naturally within  the KD-real sector of KD representations.

As a first set of examples, consider the following situation. Let $\hat A$ and $\hat B$ be two complementary CSCO and suppose 
\begin{equation}
    \hat X=f_1(\hat A) +g_1(\hat B),\quad \hat \rho= f_2(\hat A)+g_2(\hat B),
\end{equation}
where $f_1, g_1$ are real-valued and $f_2, g_2$ are positive. Then 
\begin{equation}
    [\hat X,\hat \rho]=[f_1(\hat A), g_2(\hat B)]+[g_1(\hat B), f_2(\hat A)].
\end{equation}
Note that this commutator will not vanish if the functions $f_i, g_i$ are non-trivial. It nevertheless follows straightforwardly from Eq.~\eqref{eq:imag_cond_exp} that $\mathrm{Im}\mathbb E_{\hat\rho}^\ell(\hat X|\hat B)=0=\mathrm{Im}\mathbb E_{\hat\rho}^r(\hat X|\hat B)$ and hence that $\mathrm{I_F}(\hat B;0) = 0 $ and, similarly, $\mathrm{I_F}(\hat A;0) = 0$. In the following, we shall extend this family of examples by using the KD representation of quantum mechanics associated to $\hat A$ and $\hat B$.

For that purpose,  we define the space of left/right KD-real operators as 
\begin{equation}
    \VR^{\ell/\mathrm r}=\left\{\hat C\in\mathcal{L}(\mathcal H) \mid \forall (a,b)\in\sigma(\hat A)\times\sigma(\hat{B}), \ \tilde{Q}^{\mathrm{KD,\ell/\mathrm r}}_{a,b}(\hat C)\in\R\right\}.
\end{equation}
So, $\hat C\in \VR^{\ell/\mathrm r}$ if and only if $\hat C$ has a real left/right KD-symbol. In view of Eq.~\eqref{eq:KDleftright}, one has that $\hat C\in\VR^{\ell}$ if and only if $\hat C^\dagger\in\VR^{\mathrm r}$. Consequently
\begin{equation}
V_{\mathrm{KD},\mathrm{real}}^{\mathrm{sa}}=\VR^{\ell}\cap \mathcal L^{\textrm{sa}}(\mathcal H)=\VR^{\mathrm r}\cap \mathcal L^{\textrm{sa}}(\mathcal H).
\end{equation}
 In short, $V_{\mathrm{KD},\mathrm{real}}^{\mathrm{sa}}$ is the set of all observables with a real KD symbol. We will refer to $V_{\mathrm{KD},\mathrm{real}}^{\mathrm{sa}}$ as the KD-real sector of quantum mechanics on $\mathcal H$ and we will say a state $\hat \rho$ or an observable $\hat X$ is KD real when it belongs to  $V_{\mathrm{KD},\mathrm{real}}^{\mathrm{sa}}$.

  Before turning to the study of KD-real pairs $(\hat \rho, \hat X)$, let us take a brief look at the special case where $\hat \rho$ is KD positive and $\hat X$ is KD real. In analogy with the discussion in Section~\ref{s:Q_cond_exp_mean_var}, KD-positive states can be thought of as ``classical'', at least to the extent that they allow for a joint probability distribution for $\hat A$ and $\hat B$. Moreover, we know from Theorem~\ref{thm:KDunique1_bis} that the associated conditional expectations $\mathbb E_{\hat\rho}^{Q^\mathrm{KD}}(\hat X|\hat B)$ equal $\mathbb E_{\hat\rho}^{\ell}(\hat X|\hat B)$, for all $\hat X$ and can be thought as a classical conditional expectation (see Definition~\ref{def:Qconditionalexpectation_Sigma}). When in addition $\hat X=\hat A$  and $\hat \rho$ is KD positive, we are in the situation of Lemma~\ref{lem:nogo}: the KD distribution of $\hat \rho$ is then a joint probability for $\hat A$ and $\hat B$ and satisfies (ii) of Lemma~\ref{lem:nogo}, so that in particular, $\mathbb E^{\mathrm{sa}}_{\hat\rho}(\hat A|\hat B)$ does not have anomalous values. 
 In addition, such positive states, when injected into a KD-positivity preserving quantum circuit, can be efficiently simulated classically~\cite{thio2025,pashayan2015,Delfosse2015}. Note however the following. Even if $\hat\rho$ is KD-positive, and $\hat X$ is KD-real, so that $\mathbb{E}_{\hat\rho}^{Q^\mathrm{KD}}(\hat X|\hat B)=\mathbb{E}_{\hat\rho}^{\mathrm{sa}}(\hat X|\hat B)$  is self-adjoint, it only satisfies the operator bound
 \begin{equation}
    \min Q^{\mathrm{KD}}(\hat X)\leq \mathbb E_{\hat\rho}^{\mathrm{sa}}(\hat X|\hat B)\leq \max Q^{\mathrm{KD}}(\hat X);
\end{equation}
it is therefore still possible for $E_{\hat\rho}^{\mathrm{sa}}(\hat X|\hat B)$ to take on anomalous values, since nothing guarantees that $x_{\min}\leq \tilde Q(\hat X)\leq x_{\max}$.
 This is then still a form of nonclassicality, despite the positivity of the KD distribution of $\hat \rho$. In addition, there is, for general $\hat X$ no straightforward notion of ``marginals''. In conclusion, even if KD positive states have some classical features, they may still display quantum characteristics, depending on the problem considered.

 We now turn to KD-real pairs $(\hat\rho, \hat X)$. Our first result, which is an immediate consequence of Lemma~\ref{lem:ABsymmetric}, shows that the KD-real states (and not only the KD-positive ones) are phase-insensitive provided $\hat X$ is KD real and $\hat Y=\hat B$: 
 
\begin{proposition}
\label{prop: phase insensitivity for KD real}
    Let $\hat A, \hat B$ be two complementary $\mathrm{CSCO}$ and let $Q^{\mathrm{KD},\ell/\mathrm r}$ be the corresponding left/right KD quasiprobability representations. Let $\hat X=\hat X^\dagger$ belong to $V_{\mathrm{KD},\mathrm{real}}^{\mathrm{sa}}$ and $\hat\rho$ belong to $V_{\mathrm{KD},\mathrm{real}}^{\mathrm{sa}}\cap D_{\hat{B}}$. Then the left/right conditional expectations of $\hat X$, given $\hat B$, are self-adjoint so that: 
    \begin{equation}\label{eq:realsectorB}
    \mathbb E_{\hat\rho}^{\mathrm{sa}}(\hat X|\hat B)=\mathbb E^{\ell}_{\hat\rho}(\hat X|\hat B)=\mathbb E^{\mathrm r}_{\hat\rho}(\hat X|\hat B).
    \end{equation}
     Similarly  if $\hat\rho$ belongs to $V_{\mathrm{KD},\mathrm{real}}^{\mathrm{sa}}\cap D_{\hat{A}}$,
    \begin{equation}\label{eq:realsectorA}
    \mathbb E_{\hat\rho}^{\mathrm{sa}}(\hat X|\hat A)=\mathbb E^{\ell}_{\hat\rho}(\hat X|\hat A)=\mathbb E^{\mathrm r}_{\hat\rho}(\hat X|\hat A).
    \end{equation}
    Consequently, $\hat\rho$ is phase insensitive for the measurement of $\hat B$ and for the measurement of $\hat A$.
    \end{proposition}
\begin{proof}    $\hat\rho\in V_{\mathrm{KD},\mathrm{real}}^{\mathrm{sa}} $ implies that $Q_{a, b } ^{\mathrm{KD } , \ell } (\hat{\rho } ) $ and therefore also $Q_{a | b } ^{\mathrm{KD } , \ell } (\hat{\rho } ) := Q_{a, b } ^{\mathrm{KD } , \ell } (\hat{\rho } ) / \langle \varphi ^{\hat{B } } _b , \hat{\rho } \varphi ^{\hat{B } } _b \rangle $ are real. If $\hat{X } = \hat{X }^{\dagger } $ it follows from Lemma \ref{lem:ABsymmetric} that   
$$   
\mathbb E^{\ell}_{\hat\rho}(\hat X|\hat B) = \sum _b \left( \sum _{ {a} } \overline{\tilde Q_{a,b} ^{\mathrm{KD},\ell}(\hat X)} Q_{a | b } ^{\mathrm{KD } , \ell } (\hat{\rho } ) \right) \hat{\Pi}^{\hat{B } } _b . 
$$
 This operator is self-adjoint if $\hat{X } $ has real KD-symbol. Moreover, $\mathbb E^r _{\hat\rho}(\hat X|\hat B) = \mathbb E^{\ell}_{\hat\rho}(\hat X^{\dagger } |\hat B )^{\dagger } = \mathbb E^{\ell}_{\hat\rho}(\hat X|\hat B )$, proving Eq.~\eqref{eq:realsectorB}.
    Finally, Eq.~\eqref{eq:FI} implies that $\mathrm{I}_{\mathrm{F}}(\hat A; 0) = \mathrm{I}_{\mathrm{F}}(\hat B; 0)= 0$, meaning that $\hat\rho_{\hat X}(\theta)$ is phase insensitive for the measurement of $\hat B$ and of $\hat A$ at $\theta=0$.
\end{proof} 
The proposition provides an interpretation of the KD-real sector, as follows. If $\hat\rho$ and $\hat X$ belong to the KD-real sector, then the conditional expectations $\mathbb E_{\hat\rho}^{\ell/\mathrm r}(\hat X|\hat A)$ and $\mathbb E_{\hat\rho}^{\ell/\mathrm r}(\hat X|\hat B)$ are self-adjoint, which implies that the Fisher informations  $\mathrm{I_F}(\hat A;0)$ and  $\mathrm{I_F}(\hat B;0)$ vanish so that the state $\hat\rho$ is phase insensitive for measurements of $\hat A$ or of $\hat B$.

Consequently, in order to extract information about the phase $\theta$ from information about the family of states $\hat\rho_{\hat X}(\theta)$, one needs to perform a different measurement than the one of $\hat A$ or of $\hat B$. This result is of interest because the space $V_{\mathrm{KD},\mathrm{real}}^{\mathrm{sa}}$ can in a number of cases be described quite explicitly. As already pointed out above, the operators $f(\hat A)+ g(\hat B)$, with $f:\sigma(\hat A)\to\R$ and $g:\sigma(\hat B)\to\R$, always belong to  $V_{\mathrm{KD},\mathrm{real}}^{\mathrm{sa}}$. It was proven in~\cite{langrenezetal2024} that all operators in $V_{\mathrm{KD},\mathrm{real}}^{\mathrm{sa}}$ are of this form with probability one, when the eigenbases of $\hat A$ and $\hat B$  are chosen randomly from the uniform  (Haar) distribution. 
In specific cases, the space $V_{\mathrm{KD},\mathrm{real}}^{\mathrm{sa}}$ can be considerably larger, such as for the KD representation associated to the natural position-momentum Heisenberg representation of qdits  (with $d$ not a prime number) and of finite and certain second countable LCA groups~\cite{debievreetal2025a,spriet2025}. 
Let us finally note that  the space $V_{\mathrm{KD},\mathrm{real}}^{\mathrm{sa}}$ is particularly simple when the transition matrix between $\hat{A}$ and $\hat{B}$ only has real entries: $V_{\mathrm{KD},\mathrm{real}}^{\mathrm{sa}}$ is then the space of all real symmetric matrices~\cite{langrenez2023characterizing2}. This, occurs, for example, for $n$-qubit systems~\cite{thio2025}.  In fact, we note in passing that, quite generally, if $\hat X$ and $\hat \rho$ both have real symmetric matrices in the $\hat Y$ basis, then $\mathrm{I_F}(\hat Y;0)=0$, always, as can be seen directly from Eq.~\eqref{eq:imag_cond_exp}.

Suppose $\hat X$ and $\hat \rho$ are fixed and belong to the KD-real sector. To estimate $\theta$, a measurement different from $\hat A$ or $\hat B$ then has to be chosen. As is known, the measurement $\hat L$ yielding the largest Fisher information, referred to as the quantum Fisher information,
\begin{equation}
    {\rm{I}_{QF}}=\mathrm{I_F}(\hat L,0)= 4 {\rm Tr } \left( \left( {\rm Im } \, \mathbb{E } ^{\ell/\mathrm r }_{\hat{\rho }_{\hat{X } } (\theta ) } (\hat{X } | \hat{L } )\right) ^2 \hat{\rho }_{\hat{X } } (\theta ) \right),
\end{equation}
is provided by the symmetric logarithmic derivative $\hat L$ of $\hat \rho_{\hat X}(\theta)$ at $\theta=0$.  We refer to Appendix~\ref{app:Fisher} for a self-contained discussion on the quantum Fisher information and its properties; see in particular Corollary \ref{cor:Braunstein Caves}. 
In view of what precedes, it is tempting to conjecture that, if $\hat \rho$ and $\hat X$ are both KD real, the optimal observable $\hat L$ must not be KD-real itself, provided $\rm I_{QF}$ does not vanish. For pure states, this conjecture is true, as proved in the proposition below. To that effect, we first prove the following Lemma.

\begin{lemma}
    The following statements are equivalent:
\begin{enumerate}
    \item $I_{\mathrm{QF}}(0) = 0$,
    \item $\hat \rho$ commutes with $\hat X$.
\end{enumerate}
\end{lemma}
In other words, if $\hat \rho$ and $\hat X$ do not commute, then the triplet $(\hat \rho,\hat X, \hat L)$ is not phase insensitive.
\begin{proof}
    We denote by $\hat L$ the symmetric logarithmic derivative of $\hat\rho_X{(\theta)}$ at $\theta=0$.
    If $I_{\mathrm{QF}}(0) = 0$ then by definition, we have that
    \begin{equation}
        I_{\mathrm{QF}}(0) = \mathrm{Tr}(\hat\rho \hat L^2)=0.
    \end{equation}
    We thus obtain that  $0=\Tr(\hat \rho \hat L^2)=\Tr((\hat\rho^{1/2}\hat L)^{\dagger} \hat\rho^{1/2} \hat L)$. Thus $\hat \rho^{1/2}\hat L=0$, which implies $\hat \rho\hat L=0$. The same argument proves that $\hat L\hat \rho=0$. 
    Consequently, we obtain that :
    \begin{equation}
        \partial_{\theta}\hat\rho(\theta)|_{\theta=0} = \frac{1}{2}\{\hat L, \hat\rho(0)\}= \frac{1}{2}(\hat L \hat\rho+\hat\rho\hat L) =0
    \end{equation}
    and by definition of $\rho(\theta)$, we obtain that $-i[\hat X,\hat\rho]=\partial_{\theta}\hat\rho(\theta)|_{\theta=0}=0$ and thus, $\hat X$ and $\hat \rho$ commute.
    Conversely, if $[\hat X,\hat\rho]=0$ then we still have that
    \begin{equation}
        0= -i[\hat X,\hat\rho] = \partial_{\theta}\hat\rho(\theta)|_{\theta=0} = \frac{1}{2}(\hat L \hat\rho+\hat\rho\hat L) 
    \end{equation}
    and consequently,
    \begin{equation}
         I_{\mathrm{QF}}(0) = \mathrm{Tr}(\hat\rho \hat L^2) = \mathrm{Tr}\left( \frac{1}{2}\left(\hat\rho\hat L + \hat L \hat \rho\right)\hat L\right) = 0, 
    \end{equation}
    ending the proof.
\end{proof}

\begin{proposition}\label{prop:LnotKDreal}
    Let $\hat A$ and $\hat B$ be two complementary CSCO.
Suppose that $\hat \rho = \cket{\psi}\bra{\psi}\in D_{\hat B}$ is a left KD-real pure state  and that $\hat X = \hat X ^\dagger $ is a left KD-real observable
such that $I_{\mathrm{QF}}(0) > 0$ ( or equivalently, that $\cket{\psi}$ is not an eigenvector of $\hat X)$.  Let $\hat L$ be the symmetric logarithmic derivative of $\hat\rho_X{(\theta)}$ at $\theta=0$. 

Then the conditional expectation of $\hat L$ given $\hat B$, $\mathbb{E }_{\hat \rho }^{\ell} (\hat L | \hat B ) $, is non-zero and anti-Hermitian. In particular, $\hat L$ is not left KD real.
\end{proposition}
\begin{proof}
Suppose, by contradiction, that $\mathbb{E }_{\hat \rho }^{\ell} (\hat L | \hat B )=0$. Then, as $\hat \rho\in 
    D_{\hat B}$, we have that, for all $b\in\sigma(\hat B)$:
    \begin{equation*}
        \langle\varphi_{b}^{\hat B}, \hat L \psi\rangle=0,
    \end{equation*}
    meaning that $\hat L\psi =0$ as $\hat B$ is a CSCO.
    Moreover, by definition of $\hat L$, we have that:
    \begin{equation}\label{eq:defSLD}
        \frac{1}{2}\{\hat L,\hat\rho\} = \frac{1}{2}\{\hat L,\hat\rho_{\hat X}(0)\} = \partial_{\theta}\rho_{\hat X}(\theta)|_{\theta=0} = -i[\hat X,\hat \rho].
    \end{equation}
    Applying this equality with $\hat \rho=|\psi\rangle\langle \psi|$, we obtain 
    \begin{equation*}
    -i\left(\hat X\psi - \langle \psi, \hat X\psi\rangle\psi\right) = \frac{1}{2}\left(\hat L\psi + \langle \psi, \hat L\psi\rangle\psi\right) = 0, 
    \end{equation*}
    as $\hat L\psi =0$. This implies that 
    \begin{equation*}
        \hat X\psi = \langle \psi, \hat X\psi\rangle\psi,
    \end{equation*}
    proving that $\psi$ is an eigenvector of $\hat X$, which is a contradiction.
    Thus, $\mathbb{E }_{\hat \rho }^{\ell} (\hat L | \hat B ) \neq 0$.

According to Proposition~\ref{prop: phase insensitivity for KD real}, as $\hat \rho$ and $\hat X$ are left KD real, we have that $\mathbb{E}^{\ell}_{\hat\rho}(\hat X|\hat B)$ is self-adjoint, meaning that, for all $b\in\sigma(\hat B)$:
\begin{equation*}
    \mathrm{Im}(\mathrm{Tr}(\hat \Pi_{b}^{\hat B}\hat X\hat\rho))=0.
\end{equation*}
Moreover, from Eq~\eqref{eq:defSLD}, we compute that, for all $b\in\sigma(\hat B)$:
\begin{equation*}
    \mathrm{Re } \left(\mathrm{Tr} (\hat \Pi ^{\hat B } _b \hat{L } \hat \rho )\right) = {\rm Im } \, {\rm Tr } (\hat \Pi ^{\hat B } _b \hat X \hat \rho ).
\end{equation*}
This shows that $\mathbb{E}^{\mathrm{sa},\ell}_{\hat\rho}(\hat L|\hat B)=0$, proving that $\mathbb{E}^{\ell}_{\hat\rho}(\hat L|\hat B)$ is anti-hermitian.

Finally, as $\mathbb{E}^{\ell}_{\hat\rho}(\hat L|\hat B)$ is both non-zero and anti-hermitian, and as $\hat \rho$ is left KD real, Proposition~\ref{prop: phase insensitivity for KD real} implies that $\hat L$ cannot be left KD real. This concludes the proof.   
\end{proof}

We now show that this conjecture also holds for mixed states in two examples.

 For any state $\hat \rho$, there exists a CSCO $\hat A$ such that $\hat \rho=\rho(\hat A)$ for some $\rho:\sigma(\hat A)\to \mathbb{R}$. If $\hat \rho$ has simple eigenvalues, one can take $\hat{\rho}=\hat A$. Let $\hat B$ be a CSCO complementary to $\hat A$. $\hat \rho$ is then automatically KD real. For $a\in \sigma(\hat A)$, we denote by $\rho_a$ the eigenvalue of $\hat \rho$ associated to the vector $\varphi_a^{\hat A}$.
From Eq~\eqref{eq:log derivative}, one easily computes that a logarithmic derivative of $\hat \rho_{\hat X}(\theta)$ at $\theta=0$, $\hat L$, is given by
\begin{equation*}
    \hat L_{a,a'}=\begin{cases}
        i\frac{(\rho_a-\rho_{a'})}{\rho_a+\rho_{a'}}\langle \varphi_a^{\hat A},\hat X\varphi_{a'}^{\hat A}\rangle,\quad \text{if }\rho_a+\rho_{a'}\neq 0\\
        0 \quad \text{otherwise,}
    \end{cases}
\end{equation*}
where $\hat L_{a,a'}$ are the matrix coefficients of $\hat L$ in the eigenbasis of $\hat A$. The KD symbol of $\hat{L}$ is then given by
\begin{equation}
\label{eq:symbol L}
    \tilde{Q}^{\rm KD,\ell}_{a,b}(\hat L)=i\sum_{\substack{a',b'\\ \rho_a+\rho_{a'}\neq 0}}\frac{\rho_a-\rho_{a'}}{\rho_a+\rho_{a'}}\tilde{Q}_{a,b'}^{\rm KD,\ell}(\hat X)\frac{\langle \varphi_a^{\hat A},\varphi_{b'}^{\hat B}\rangle}{\langle \varphi_a^{\hat A},\varphi_b^{\hat B}\rangle}\langle \varphi_{b'}^{\hat B},\varphi_{a'}^{\hat{A}}\rangle\langle  \varphi_{a'}^{\hat A},\varphi_b^{\hat B}\rangle,
\end{equation}
$(a,b)\in \sigma(\hat A)\times \sigma(\hat B)$. From this expression we can give at least two situations in which the KD symbol of $\hat L$ is purely imaginary:
\begin{enumerate}
    \item If the transition matrix between $\hat A$ and $\hat B$ has real coefficients and $\hat X$ is KD real, then, the right hand side of Eq~\eqref{eq:symbol L} is purely imaginary. The condition that the transition matrix between $\hat A$ and $\hat B$ has real coefficients is for example fulfilled in systems of $n$ qbits, where the transition matrix between $\hat A$ and $\hat B$ is given by the matrix of the Fourier transform on the abelian group $\left(\mathbb{Z}/2\mathbb{Z}\right)^n$ \cite{debievreetal2025a, thio2025}.\\
    \item If $\hat X$ is of the form $\hat X= f(\hat A)+\hat \Pi_{b_0}^{\hat B}$, where $f:\sigma(\hat A)\to \mathbb{R}$ and $b_0\in \sigma(\hat B)$. It is clear in this case that $\hat X$ belongs to the KD real sector. Moreover it is easily computed in this case that for any $a\in \sigma(\hat A)$,
    \begin{equation*}
        \tilde{Q}_{a,b_0}^{\rm KD,\ell}(\hat L)=i\sum_{\substack{a'\\\rho_a+\rho_{a'}\neq 0}}\frac{\rho_a-\rho_{a'}}{\rho_a+\rho_{a'}}\Big|\langle \varphi_{b_0}^{\hat B},\varphi_{a'}^{\hat A}\rangle \Big|^2,
    \end{equation*}
    which is also purely imaginary.
\end{enumerate}

One may conclude from this analysis that, in order to test the phase sensitivity of a KD-real state $\hat \rho$ under the unitary flow generated by a KD-real observable $\hat X$, one cannot limit oneself to measuring the observable $\hat B$: indeed, under these hypotheses, the triplet $(\hat \rho, \hat X, \hat B)$ is phase-insensitive. This analysis does not rule out the possibility that some KD-real observable $\hat C$ exists for which $I_{\mathrm F}(\hat C, 0)>0$. Nevertheless, we proved that the optimal choice $\hat L$ of measurement observable is not KD real, provided $\hat \rho$ is pure. We leave the question whether this remains true for all mixed states open.
Another way to put the above result is to say that, if a pure state $\hat \rho$ and $\hat X$ do not commute,  then there does not exist a KD representation defined by two CSCO $\hat A$ and $\hat B$ for which $\hat L=g(\hat B)$ and for which both $\hat \rho$ and $\hat X$ are KD real.

We end this section with  a short remark on post-selected phase estimation, as studied in~\cite{Gladstein2022, Arvidsson-Shukur2020}. It is proven in those works that a quantum advantage can be obtained in that case, even though the state is KD real, but not KD positive. This does not contradict our findings for two reasons. First of all, the phase estimation problem considered there differs from the one considered here because of the post-selection process. As a result, the Fisher information is different as well. In addition, the relevant KD distribution used to obtain the results of~\cite{Gladstein2022, Arvidsson-Shukur2020} is different from the one considered here.

\section{Discussion and conclusions}

There is elegance in defining concepts through optimization: the principles of least action and of maximal entropy are eminent examples. There is also  convenience in defining concepts in quantum mechanics by analogy with concepts in classical mechanics and probability theory. We have combined both approaches in order to define the notion of conditional expectation of an observable $\hat X$ given an observable $\hat Y$ in quantum mechanics: it is  the best predictor of $\hat X$ by a function of $\hat Y$ in the sense that it is the function of $\hat Y$ that minimizes a quadratic error in a given state $\hat \rho$. In the first part of this paper, we have shown that, because of the non-commutativity between observables inherent in quantum mechanics, several different choices of quadratic error mimicking the classical choice exist, and lead to different definitions. We have explained how two of these choices are particularly natural. First, they are intimately linked to weak value physics, thereby sharpening the known links between weak values and quantum conditional expectations. Second, our approach allows us to establish that the quantum conditional expectations so defined are characterized fully by a ``pull out'' property, and obey a law of iterated expectations as well as a law of total variance, similar to the original notion in probability theory. In this manner, the definition of the quantum conditional expectation through minimization of a quadratic error provides a unified picture within which the essential properties of the quantum conditional expectation naturally emerge.  In addition, by working in analogy with classical probability theory, this approach allows us to highlight precisely under what circumstances the quantum object differs from its classical counterpart: in other words, it sheds light  on aspects of the classical-quantum boundary.  Notably, the conditional expectation of an observable is not necessarily an observable, but can fail to be self-adjoint, a feature that has no equivalent in the classical theory. Its imaginary part is thus a typical quantum feature. Its variance provides an extra term in the law of the total variance, absent in the classical theory. This imaginary part has an interpretation as the Fisher information associated to a phase estimation problem associated to the triplet $\hat\rho$, $\hat X$, $\hat Y$. We also show how the quantum conditional expectation can be used to establish a simple but efficient state-dependent no-go theorem for the existence of joint probabilities in quantum mechanics. 

We conclude from this analysis that the proposed approach to quantum conditional expectations via minimization of a quadratic error provides a satisfying unifying picture that allows to analyse efficiently its properties and to pinpoint its similarities and -- most importantly -- its differences with the classical notion. 

In the second part of the paper, we establish a  relation between these results and the theory of quasiprobability representations of quantum mechanics.  The origin of the difficulties with the notion of conditional expectations in quantum mechanics is well known to be that joint probabilities of incompatible observables cannot be naturally defined. Quasiprobability representations of quantum mechanics circumvent this difficulty by allowing for nonpositive and even nonreal distributions to represent quantum states. This has led to a definition of joint quasiprobabilities and, through an analog of Bayes' rule, to an alternative notion of conditional expectation from the one through minimization, that however depends very strongly on the choice of quasiprobability representation. 
Our main result in this paper is then that, among all possible quasiprobability representations that produce the correct Born probabilities associated to two chosen observables $\hat A$ and $\hat B$, only the Kirkwood-Dirac quasiprobability representations reproduce the above conditional expectations defined through miminization of a quadratic error. This provides therefore a defining feature for the Kirkwood-Dirac representations, singling them out among all Born-compatible quasiprobability representations. 

In the last part of the paper, we have applied our insights to provide an interpretation of the KD-real sector of quantum mechanics in terms of phase-insensitivity. 

Our work leaves a number of questions unanswered. First, whereas it is certainly elegant and natural to require a quantum conditional expectation to obey a ``pull-out'' formula, as we did, we did not find a direct operational meaning to such a pull-out formula. Such an interpretation would certainly strengthen the case for the approach promoted here. As we pointed out, other approaches exist; see, for example,~\cite{tsang2022, tsang2023,PRXQuantum.4.020334}, where no pull-out property is required and~\cite{umegaki1954}, where in a more algebraic context, on the contrary, both the left and right pull-out properties are required, thereby effectively restricting the states for which the notion makes sense. Note, however, that giving up the ``pull-out'' formula implies giving up the link with weak-value physics, as we proved. 
Second, although a quantum conditional variance  appears naturally within our approach, we did not identify an operational meaning for this quantity, contrary to what is known for the conditional expectations, which are linked to weak value physics. Indeed, although the expected value of this conditional variance plays an analogous role to the so-called ``unexplained variation'' familiar from statistics, no such interpretation seems to arise in an evident manner in quantum mechanics. It would be clearly of interest to develop such an interpretation. In addition, the definition of the conditional variance does not seem to be unique, several possibilities presenting themselves, that appear natural from different points of view.

Finally, to avoid technical complications and to allow for mathematical precision, we have restricted ourselves in this paper  to finite probability spaces $\Omega$ and to finite dimensional Hilbert spaces $\mathcal H$.  While many of our results extend, at least on a formal level, to the infinite dimensional setting, technical difficulties do then arise that we have not tried to address: see \cite{Brummelhuis2025a} for quantum conditional expectations given a self-adjoint operator when the underlying state  is pure.

\emph{Acknowledgments:} This work was supported in part by the CNRS through the MITI interdisciplinary programs. S. De Bièvre, C. Langrenez and M. Spriet acknowledge the support of the CDP C2EMPI, as well as the French State under the France-2030 programme, the University of Lille, the Initiative of Excellence of the University of Lille, the European Metropolis of Lille for their funding and support of the R-CDP-24-004-C2EMPI project. The authors thank David Arvidsson-Shukur, Justin Dressel, and Andrew Jordan for stimulating and insightful discussions on the subject of this paper. They also thank Arthur Parzygnat as well as two anonymous referees for their constructive comments on (a previous version of) this paper. 

\bibliographystyle{ieeetr}
\bibliography{Bibliography}

\appendix

\section{Conditioned measurements in von Neumann's measurement scheme}   \label{app:weakvalues}   

\medskip 
\subsection{Quantum measurements according to von Neumann} In von Neumann's model for the measurement of a quantum mechanical observable $\hat{X } = \hat{X } ^{\dagger } $ acting on some Hilbert space $\mathcal{H }$~\cite{Neumann18},     
$\hat{X } $ is coupled to a measuring device or meter with pointer position $q \in \mathbb{R }. $ Both the system and the meter are treated quantum mechanically, and we let $\hat{Q } $ and $\hat{P } = i^{-1 } \partial _q $ be the pointer's position and momentum operator, acting on $L^2 (\mathbb{R } ) $, where Planck's constant $\hbar $ is set equal to 1. If the system is in an eigenstate of $\hat{X } $ with eigenvalue $x \in \sigma (\hat{X } ) $, the measuring device is assumed to cause the pointer to undergo a displacement $\gamma x $, where the constant $\gamma $ is proportional to the interaction strength and to the duration of the system-meter interaction. The effect of the measurement is therefore modeled by the unitary transformation   
$$   
\hat{U }_{\gamma } := e^{- i \gamma \hat{X } \otimes \hat{P } } ,    
$$   
acting on the tensor product Hilbert space $\mathcal{H } \otimes L^2 (\mathbb{R } ) . $ Let $\hat{\rho } $ be the initial state of the system and assume that the initial state of the pointer is a pure state defined by a normalized wave-function $\varphi = \varphi (q) \in L^2 (\mathbb{R } ) . $ In the weak value literature, $\varphi $ is often taken to be a Gaussian, but we will allow it to be an arbitrary Schwartz function here, sometimes restricted to be real-valued to simplify formulas. (This condition can be weakened to requiring that $\varphi $ belong to the domain of suitable powers of $\hat P $ and $\hat{Q } $, depending on the order of approximation in $\gamma$.) 

The initial state of the compound system is the tensor product state $\hat{\rho } \otimes | \varphi \rangle \langle \varphi | $ and its state after the system-meter interaction is given by   
$$   
\hat{\sigma }_{\gamma } := \hat{\sigma }_{\gamma } (\hat{\rho } , \varphi ) := \hat{U }_{\gamma } (\hat{\rho } \otimes 
| \varphi \rangle \langle \varphi | ) \hat{U }_{\gamma }^{\dagger } .    
$$   
The quantum mechanical expectation of the pointer position $\hat{Q } $, post-interaction, is then found to be     
\begin{equation} \mathbb E_{\hat \sigma_\gamma}(\hat Q)=\label{eq:unconditional_vN}   
{\rm Tr } ( (I_{\mathcal{H } } \otimes \hat{Q } ) \hat{\sigma }_{\gamma } )
= \langle \varphi , \hat{Q } \varphi \rangle + \gamma {\rm Tr } (\hat{X } \hat{\rho } ) ,   
\end{equation}
and its variance post-interaction is given by 
\begin{equation} \label{eq:unconditional_Q_variance}
    {\rm var }_{\gamma } (\hat Q ) := \Delta _{\hat \sigma _{\gamma } } (\hat Q ) = {\rm var }_{\varphi } (\hat Q ) + \gamma ^2 {\rm var }_{\hat \rho } (\hat X )
\end{equation}
where ${\rm var }_{\varphi}$ is the variance associated to the state $\varphi$ (denoted $\Delta^{2}_{\varphi}$ in the main text),
 as a direct computation starting from Eq.~\eqref{eq:sigma_gamma} shows. The first term on the right is the initial mean position of the pointer, which we can always assume to be 0 by selecting the pointer's zero reading, and ${\rm Tr } (\hat{X } \hat{\rho } ) $ can in principle be inferred from measurements of $\hat{Q }$, provided $\gamma$ is known. In other words, a first point to make  here is that some information about the observable $\hat X$ of the system, while it is in the state $\hat \rho$, is obtained from a measurement of $\hat Q$ on the meter. Note however that, if the variance $\Delta \hat Q$ of $\hat Q$ in the state $\hat \rho$ is large, and  $\gamma {\rm Tr } (\hat{X } \hat{\rho })$ is small, then the signal-to-noise ratio will be unfavourable and, supposing $\gamma$ is known, an accurate determination of $\Tr(\hat X\hat \rho)$ will require a very large number of individual measurements.

 Similarly, if $\gamma$ is not known, but $\Tr(\hat X\hat\rho)$ is, this procedure allows for the determination of $\gamma$, but with the same caveat when $\gamma\Tr(\hat X\hat \rho)$ is small. In that case, a conditional measurement of $\hat Q$ can improve the signal-to-noise ratio, and lead to what is referred to as weak-value amplification, as we explain in the next subsection. The conditioning is based on the following procedure.
 
 If $\hat{Y } $ is a second observable on $\mathcal{H }$, then $\hat{Y } \otimes I_{L^2 (\mathbb{R } ) } $ and $\hat{Q } =   
I_{\mathcal{H } } \otimes \hat{Q } $ commute as operators on $\mathcal{H } \otimes L^2 (\mathbb{R } ) $, unlike (in general) $\hat{Y } $ and $\hat{X } $, and can be simultaneously measured using standard projective measurements. This then leads to a well-defined joint probability distribution of those two observables. In particular, simply writing $\hat{Q } $ for    
$I_{\mathcal{H } } \otimes \hat{Q } $ and $\hat{Z } $ for $\hat{Z } \otimes I_{L^2 (\mathbb{R } ) } $ for operators $\hat{Z } $ on $\mathcal{H } $, the {\it classical} conditional expectation, post interaction, $\hat{Q } = I_{\mathcal{H } } \otimes \hat{Q } $ given that $\hat{Y } = y $ with $y \in \sigma (\hat{Y } ) $  is, as we saw (see Remark \ref{rmk:QCE}(vi)), equal to the quantum conditional expectation $\mathbb E_{\hat\sigma_\gamma}(\hat Q|\hat Y)$ of $\hat Q$, given $\hat Y$, and is given by   
\begin{equation} \label{eq:CCE_gamma}
\mathbb{E } _{\hat\sigma_\gamma } (\hat{Q } | \hat{Y } = y ) = \frac{{\rm Tr } (\hat{\Pi }^{\hat{Y } } _y  \hat{Q } \hat{\sigma }_{\gamma } ) }{{\rm Tr } ( \hat{\Pi }^{\hat{Y } }_y \hat{\sigma }_{\gamma } ) }.   
\end{equation} 

This conditional expectation can be experimentally determined by doing successive measurements of $\hat{Q } $ and $\hat{Y }$ on the state $\hat\sigma_\gamma$ and by post-selecting on $\hat{Y } = y $ or, what is equivalent since we are dealing with  commuting operators, by only measuring $\hat{Q } $ when $\hat{Y } = y $, which can be implemented by passing the joint system through a filter which filters out the appropriate eigenstate of $\hat{Y } $ before reading the meter position. We show in the next subsection that in the limit of small $\gamma $ the real and imaginary parts of $\mathbb{E }_{\hat{\rho } } (\hat{X } | \hat{Y } = y ) $   determine the first order correction to the value of $\mathbb E_{\hat\sigma_\gamma}(\hat Q|\hat Y)$, for $\gamma = 0 $. When $\mathbb E_{\hat\rho}(\hat X|\hat Y)$ is anomalous, meaning that it lies outside the range of the spectrum of $\hat X$, this can enhance the signal-to-noise ratio referred to above, a phenomenon referred to as weak-value amplification~\cite{Aharonov88, Leggett_1989, Peres_1989,Duck89,Tamir2013,Jozsa2007,Dressel14,arvidssonshukur2024properties,dresselContextualvalueApproachGeneralized2012, Williams_Jordan_2008,Dressel_Jordan_2012,Dressel_Jordan_2012_ImaginaryPartWeakValue}.  This phenomenon is considered of typical quantum nature since such anomalous values cannot occur classically.

\subsection{Conditioned masurements in the weak interaction limit} 

In the analysis presented in this section only the left quantum conditional expectation intervenes. In order to lighten the notation, we therefore drop the ``$\ell$'' upper-script.
Let us define the quantum mechanical covariance of two observables $\hat C_1 $ and $\hat C_2 $ in a state $\hat{\sigma } $, all defined on some unspecified Hilbert space, by   
\begin{eqnarray}   
{\rm cov }_{\hat{\sigma }} (\hat{C }_1 , \hat{C }_2 ) &:= & {\rm Re } \, {\rm Tr } (\hat{C }_1 \hat{C }_2  \hat{\sigma } ) - {\rm Tr } (\hat{C }_1 \hat{\sigma } ) {\rm Tr } (\hat{C }_2 \hat{\sigma } ) \\   
&=& \frac{1}{2}{\rm Tr } (\{ \hat{C }_1 , \hat{C }_2 \} \hat{\sigma } ) - {\rm Tr } (\hat{C }_1 \hat{\sigma } ) {\rm Tr } (\hat{C }_2 \hat{\sigma } ) \nonumber   
\end{eqnarray}   
where we recall that $\{ \hat{C }_1 , \hat{C }_2 \} =  \hat{C }_1 \hat{C }_2 + \hat{C }_2 \hat{C }_1  $ is the anti-commutator. If $\hat{C }_1 = \hat{C }_2 = \hat{C } $, we write ${\rm cov }_{\hat{\sigma } } (\hat{C } , \hat{C } ) = {\rm var }_{\hat{\sigma } } (\hat{C } ) = \Delta^{2}_{\hat\sigma}(\hat C) $, the variance of $\hat{C }$. If the state is a pure state
$| \varphi \rangle \langle \varphi | $ we will simply write ${\rm cov }_{\varphi } $ instead of ${\rm cov }_{| \varphi \rangle \langle \varphi | } $, and similarly for the variance.   
\smallskip

\subsubsection{Small \texorpdfstring{$\gamma$}{E[...]} expansions of conditional expectations of pointer variables}  

We first compute the first order expansion in $\gamma$ of the conditional quantum expectations of $\hat Q$ and of $\hat P$. Recall our standing assumption that $\varphi $ is Schwartz class.   
\begin{proposition}\label{prop:EspGammaOrder1}
     As $\gamma \to 0 $, we have that:
     \begin{equation} \label{eq:thm_CCE_WVa_Q_2}
\mathbb{E } _{\hat\sigma_\gamma } (\hat{Q } | \hat{Y } = y ) = \langle \varphi, \hat{Q } \varphi \rangle+ \gamma \cdot \left({\rm Re } \, \mathbb{E }_{\hat{\rho } } (\hat{X} | \hat{Y } = y ) +    
2 {\rm Im } \mathbb{E }_{\hat{\rho } } (\hat{X } | \hat{Y } = y ) \cdot   
{\rm cov } _{\varphi } (\hat{Q } , \hat{P } ) \right) + O(\gamma ^2 ) ,   
\end{equation}
where ${\rm cov } _{\varphi } (\hat{Q } , \hat{P } ) := {\rm Re } \langle \hat{Q } \varphi , \hat{P } \varphi \rangle - \langle \varphi, \hat{Q } \varphi \rangle \langle \varphi , \hat{P } \varphi \rangle $ is the quantum mechanical covariance of the pointer position and momentum operators $\hat{Q } $ and $\hat{P } $ in the initial pointer state. Moreover, we also find that:

\begin{equation} \label{eq:thm_CCE_WVa_P_2}   
\mathbb{E }_{\hat\sigma_\gamma } (\hat{P } | \hat{Y } = y ) = \langle \varphi , \hat{P } \varphi \rangle + 2 {\rm Im } \, \mathbb{E }_{\hat{\rho } } (\hat{X } | \hat{Y }  = y ) {\rm var }_{\varphi } (\hat{P } ) \cdot \gamma  + O(\gamma ^2 ) ,   
\end{equation}   
where ${\rm var }_{\varphi } (\hat{P } ) = {\rm cov }_{\varphi } (\hat{P } , \hat{P } ) = \langle \hat{P } \varphi , \hat{P } \varphi \rangle - \langle \varphi , \hat{P } \varphi \rangle ^2$ is the quantum mechanical variance of the pointer momentum operator in the initial pointer state $\varphi$.
\end{proposition}

The proof will show that more generally, for functions of the position operator $\hat Q$, if $f $ is a real-valued differentiable function of at most polynomial growth, then     
\begin{eqnarray} \label{eq:thm_CCE_WV_bis_Q_2}    
\frac{d }{d \gamma }  \mathbb{E }_{\hat\sigma_\gamma } (f(Q ) | Y = y ) |_{\gamma = 0 }    
&=& {\rm Re } \, \mathbb{E }_{\hat{\rho } } (\hat{X} | \hat{Y } = y ) \cdot \langle  \varphi, f' (\hat{Q } ) \varphi \rangle \\   
&\ & + \, 2 {\rm Im } \mathbb{E }_{\hat{\rho } } (\hat{X } | \hat{Y } = y ) \cdot   
{\rm cov } _{\varphi } (f(\hat{Q } ) , \hat{P } ) . \nonumber   
\end{eqnarray}   
For the momentum operator $\hat P$, we find more generally that   
\begin{equation}\label{eq:thm_CCE_WV_bis_P_2}   
\frac{d }{d\gamma } \mathbb{E }_{\hat\sigma_\gamma } (f(\hat{P } ) | \hat{Y } = y ) |_{\gamma = 0 } = 2 {\rm Im } \, \mathbb{E }_{\hat{\rho } } (\hat{X } | \hat{Y }  = y ) \cdot {\rm cov }_{\varphi } (f (\hat{P } ) , \hat{P } ) .      
\end{equation}       

\begin{proof} For the proof of Eq.~\eqref{eq:thm_CCE_WVa_Q_2}, we will simply compute the derivative at $\gamma = 0 $ and show that   
\begin{equation} \label{eq:thm_CCE_WV}   
\frac{d }{d \gamma } \mathbb{E } _{\hat\sigma_\gamma } (\hat{Q } | \hat{Y } = y ) |_{\gamma = 0 } = {\rm Re } \, \mathbb{E }_{\hat{\rho } } (\hat{X} | \hat{Y } = y ) +    
2 {\rm Im } \mathbb{E }_{\hat{\rho } } (\hat{X } | \hat{Y } = y ) \cdot   
{\rm cov } _{\varphi } (\hat{Q } , \hat{P } ) .   
\end{equation}
If $p_{\gamma } (q, y ) $ is the joint probability (density) of having pointer position $\hat{Q } = q $ while $\hat{Y } = y $, then the conditional pointer position density is given by     
\begin{equation} \label{eq:cond_pointer_density}   
p_{\gamma } (q | y ) := \frac{p_{\gamma } (q , y ) }{p_{\gamma } (y) } , \, p_{\gamma } (y ) = \int _{\mathbb{R } } p_\gamma(q, y ) dq ,    
\end{equation}   
and   
\begin{equation} \label{eq:proof_CCE_WV1}   
\mathbb{E } _{\hat\sigma_\gamma} (\hat{Q } | \hat{Y } = y ) = \int _{\mathbb{R } } q p_{\gamma } (q | y ) dq .          
\end{equation}   
We compute $p_{\gamma } (q, y ) $ and its derivative with respect to $\gamma . $ Since $e^{- i \gamma \hat{X } \otimes \hat{P }  } = \sum _{x \in \sigma (\hat{X } ) } e^{ - i \gamma x \hat P } \hat{\Pi }^{\hat{X } } _x $, we have   
\begin{equation} \label{eq:sigma_gamma}   
\hat \sigma _{\gamma } = \sum _{x, x' \in \sigma (\hat{X } ) } \hat{\Pi }  ^{\hat {X } } _x \hat{\rho } \hat{\Pi }^{\hat{X } } _{x' } \otimes | \varphi ( \cdot - \gamma x ) \rangle \langle \varphi (\cdot - \gamma x' ) | .   
\end{equation}   
It follows that   
\begin{equation} \label{eq:proof_CCE_WV2}   
p_{\gamma } (q, y ) = \sum _{x , x' } {\rm Tr } (\hat{\Pi }^{\hat Y } _y \hat{\Pi }^{\hat{X } } _x \hat{\rho } \hat{\Pi }^{\hat{X } } _{x' } )  \varphi (q - \gamma x ) \overline{\varphi (q - \gamma x' ) } .   
\end{equation}   
The derivative with respect to $\gamma $ in $\gamma = 0 $ is   
\begin{eqnarray*}   
\partial _{\gamma } p_{\gamma } (q, y ) |_{\gamma = 0 } 
&=& - \sum _{x, x' \in \sigma (\hat{X } ) }\left( x {\rm Tr } (\hat{\Pi }^{\hat{Y } } _y \hat{\Pi }^{\hat{X } } _x \hat{\rho } \hat{\Pi }^{\hat{X } } _{x' } )  \partial _q \varphi (q ) \overline{\varphi (q ) } + x' \, {\rm Tr } (\hat{\Pi }^{\hat{Y } } _y  \hat{\Pi }^{\hat{X } } _x \hat{\rho }  \hat{\Pi }^{\hat{X } } _{x' } )  \varphi (q ) \overline{\partial _q \varphi (q ) }\right) \\      
&=& - {\rm Tr } (\hat{\Pi }^{\hat{Y } } _y \hat{X } \hat{\rho } ) \, \partial _q \varphi (q ) \overline{\varphi (q ) } - {\rm Tr } (\hat{\Pi }^{\hat{Y } } _y \hat{\rho } \hat{X } ) \varphi (q ) \, \overline{\partial _q \varphi (q ) } \\   
&=& -  2 {\rm Re } \, \left( {\rm Tr } (\hat{\Pi }^{\hat{Y } } _y \hat{X } \hat{\rho } ) \, \overline{\varphi (q) } \partial _q \varphi (q) \right) ,       
\end{eqnarray*}  
and since $\overline{\varphi (q) } \partial _q \varphi (q)   = \frac{1 }{2 } \partial _q |\varphi (q) |^2 + i {\rm Re } \, \overline{\varphi (q) } \hat P\varphi (q) $, this equals
\begin{equation}   \label{eq:partial_gamma_pgamma}
\partial _{\gamma } p_{\gamma } (q, y ) |_{\gamma = 0 } = - {\rm Re } ( {\rm Tr } (\hat{\Pi }^{\hat{Y } } _y \hat{X } \hat{\rho } ) ) \, \partial _q |\varphi (q) |^2 + 2 {\rm Im } ({\rm Tr } (\hat{\Pi }^{\hat{Y  } } _y \hat{X } \hat{\rho } ) ) \, {\rm Re } ( \overline{\varphi (q) } \hat P\varphi (q) ) .   
\end{equation}   
Integrating with respect to $q $,   
\begin{equation} \label{eq:post_weak_interaction_y}   
\partial _{\gamma } p_{\gamma } (y) = \int \partial _{\gamma } p (q, y ) dq = 2 {\rm Im } ({\rm Tr } (\hat{\Pi }^{\hat{Y  } } _y \hat{X } \hat{\rho } ) ) \langle \varphi ,\hat P \varphi \rangle ,   
\end{equation}
Finally, differentiating (\ref{eq:cond_pointer_density}) while using that $p_0 (q, y )  = |\varphi (q) |^2 {\rm Tr } (\hat{\Pi }^{\hat{Y } } _y \hat{\rho } ) $ and $p_0 (y) = {\rm Tr } (\hat{\Pi }^{\hat{Y } } _y \hat{\rho } ) $, we find   
\begin{eqnarray} \label{eq:proof_CCE_WV3}   
\partial _{\gamma }  p(q | y ) |_{\gamma = 0 } &=& - {\rm Re } \, \mathbb{E }_{\hat{\rho } } (\hat{X} | \hat{Y } = y ) \partial _q |\varphi |^2 +    
2 {\rm Im } \mathbb{E }_{\hat{\rho } } (\hat{X } | \hat{Y } = y ) {\rm Re } (\overline{\varphi } \hat P \varphi ) \\   
&\ & - 2 {\rm Im } \mathbb{E }_{\hat{\rho } } (\hat{X } | \hat{Y } = y ) \langle \varphi , \hat{P } \varphi \rangle |\varphi |^2 . \nonumber   
\end{eqnarray}   
and (\ref{eq:thm_CCE_WV}) follows from (\ref{eq:proof_CCE_WV1}) after an integration by parts. Equation (\ref{eq:thm_CCE_WV_bis_P_2}) is proven similarly by differentiating $\mathbb{E }_{\hat\sigma_\gamma } ( f(\hat{Q } ) | \hat{Y } = y ) = \int f(q) p_{\gamma } (q | y ) dq . $    

We next turn to the proof of Eq.~\eqref{eq:thm_CCE_WVa_P_2}. This equation is shown as before by computing the derivative with respect to $\gamma$, a computation that we will carry in momentum representation. If we (somewhat idiosyncratically, but the letter $p $ is already used to designate probability densities) denote the conjugate (or dual) variable of $q $ by $k $ and let $\hat{\varphi } (k ) $ be the Fourier transform of $\varphi $ (normalized so as to be unitary on $L^2 (\mathbb{R } )$), then the probability density that $\hat{P } = k $ while $\hat{Y } = y $ is given by   
$$   
p_{\gamma } (k , y ) = 
\sum _{x, x' } {\rm Tr } (\hat{\Pi }^{\hat{Y } } _y \hat{\Pi }^{\hat{X } } _x \hat{\rho } \hat{\Pi }^{\hat{X } } _{x' }) e^{- i \gamma (x - x' ) k } |\hat{\varphi } (k ) |^2 ,   
$$   
and   
$$   
\partial _{\gamma } p_{\gamma } (k , y ) |_{\gamma = 0 } = 2 {\rm Im } {\rm Tr } (\hat{\Pi }^{\hat{Y } } _y \hat{X } \hat{\rho } ) \, k |\hat{\varphi } (k ) |^2 . $$   
If $p_{\gamma } (k | y ) = p_{\gamma } (k , y ) / p_{\gamma } (y ) $, with $p_{\gamma } (y) = \int p_{\gamma } (k , y ) dk $, then we find that  $$   
\partial _{\gamma } p_{\gamma } (k | y ) = 
{\rm Im } \, \mathbb{E }_{\hat{\rho } } (\hat{X } | \hat{Y } = y ) \bigl( k |\hat{\varphi } (k ) |^2 - \langle \varphi , \hat{P } \varphi \rangle |\hat{\varphi } (k ) |^2)   
$$   
(note that $\partial _{\gamma } p_{\gamma } (y) |_{\gamma = 0 } = 2 {\rm Im } {\rm Tr } (\hat{\Pi }^{\hat{Y } } _y \hat{X } \hat{\rho } ) \int k |\hat{\varphi } (k ) |^2 dk $). The theorem follows by differentiating 
$\mathbb{E }_{\hat\sigma_\gamma } (\hat{P } | \hat{Y } = y ) = \int k p_\gamma(k | y ) dk $ under the integral sign. Similarly for (\ref{eq:thm_CCE_WV_bis_P_2}) by differentiating $\int f(k) p_\gamma(k | y ) dk . $
\end{proof}
As we saw, therefore, the imaginary part of the weak values/quantum conditional expectations manifest themselves in the conditional $\hat Q$ and $\hat{P }$ measurements of the meter, provided the latter is prepared in appropriate states~\cite{Dressel_Jordan_2012_ImaginaryPartWeakValue}.

Equation~(\ref{eq:thm_CCE_WVa_Q_2}), with 
$\mathbb{E }_{\hat{\rho } } (\hat{X } | \hat{Y } = y ) $ replaced by its explicit expression ${\rm Tr } (\hat{\Pi }^{\hat{Y } } _y \hat{X } \hat{\rho } ) / {\rm Tr } (\hat{\Pi }^{\hat{Y } } _y \hat{\rho } ) $ 
goes back to the paper by Aharonov {\it et al.}~\cite{Aharonov88}, where the term  weak value is introduced. Equation~\eqref{eq:thm_CCE_WVa_P_2} has also been noted before~\cite{Dressel_Jordan_2012_ImaginaryPartWeakValue}. The measurements of $\hat Q$ or of $\hat P$ are said to be ``weak'' when the post-interaction state $\hat\sigma_\gamma$ is not much perturbed from the initial state. This will clearly happen in the extreme case when $|\varphi\rangle=|p=0\rangle$ is the momentum eigenstate with zero momentum, independently of how large $\gamma$ is. Since the momentum eigenstates are not physical, in practice, one mostly considers a centered Gaussian state with a large $\hat Q$-variance $\Delta$, which is therefore sharply localized in momentum about $p=0$. This leads then to a small perturbation of the initial state $\hat \rho\otimes |\varphi\rangle\langle \varphi|$ of the system plus probe under $\hat U_\gamma$, provided $\gamma$ is not too large, depending on $\Delta$. For a detailed analysis of the conditions needed, we refer to~\cite{Duck89}: a rule of thumb is that one needs $\gamma(x_{\max}-x_{\min})\ll \Delta$. For our purposes here, we only consider Taylor expansions of the relevant observers, to leading order, without explicitly controlling the error terms.

The definition of weak value as introduced in the literature depends on the measurement protocol described above. It has, \emph{a posteriori}, been interpreted as a conditional expectation because of its appearance in the low order expansions in $\gamma$ of the \emph{classical} conditional expectations associated to the joint probability distribution of the commuting observables $\hat Q$ and $\hat Y$. Our perspective is different: we have defined the quantum conditional expectation  $\mathbb E_{\hat\rho}(\hat X|\hat Y)$ as the solution to an appropriate minimization problem, in direct analogy with the conditional expectation in classical probability.  We then showed that the conditional expectation so defined has the natural properties expected from a conditional expectation. Finally, Eq.~\eqref{eq:thm_CCE_WVa_Q_2} and~\eqref{eq:thm_CCE_WV_bis_P_2}  give an operational meaning to this abstract definition.

The ``weak measurement'' terminology is somewhat confusing since the two measurements that are performed, of the meter position $\hat Q$ and of the system variable $\hat Y$ are both projective, hence  ``strong'' measurements, but it is the interaction between system and meter that is weak. Also, as a result, the information about the probability distribution $p(x)=\Tr(\hat\Pi_x^{\hat X}\hat\rho)$ obtained in this manner is partial, and the effect of the measurement on the system limited, or ``weak''.

We can go one step further and compute the term of order $\gamma^2$ for the conditional $\hat Q$ measurements. We define :
\begin{equation}\label{eq:def_E_cal_2}
    \mathcal{E }_2 (\hat X ; \hat \rho , y ) := \frac{{\rm Tr } \bigl( \hat{\Pi }^{\hat{Y } } _y \hat{X } \hat{\rho } \hat{X } \bigr) }{{\rm Tr } (\hat{\Pi }^{\hat{Y } } _y \hat{\rho } ) }\geq 0.
\end{equation}
Note that $ \mathcal{E }_2 (\hat X ; \hat \rho , y )$ already appears in~\cite{Ogawa_2021}.  
Its physical/probabilistic interpretation is not clear, except in the two following examples: first, when $[\hat X,\hat\rho]=0$ or when $[\hat X,\hat Y]=0$, then
\begin{equation*}
    \mathcal{E }_2 (\hat X ; \hat \rho , y ) = \mathbb{E}_{\hat\rho}(\hat X^2 | Y=y)
\end{equation*}
Second, when $\hat \rho = |\psi \rangle \langle \psi | $ is a pure state, and $\hat{Y } $ is a CSCO. In that case, $\hat \Pi ^{\hat Y } _y $ is the projection onto the joint eigenvector $\varphi _y $, and      
\begin{equation} \label{eq:E_cal_pure_2}   
\mathcal{E }_2 (\hat X ; \hat \rho , y ) = 
\frac{{\rm Tr } \,  \bigl( | \varphi _y \rangle \langle \varphi _y | \hat X | \psi \rangle \langle \psi | \hat X \bigr) }{| \langle \varphi _y , \psi \rangle |^2 } =   
\frac{| \langle \varphi _y , \hat X \psi \rangle |^2 }{|\langle \varphi _y , \psi \rangle |^2 } = \left | \mathbb{E }_{| \psi \rangle \langle \psi | } (\hat X | \hat Y = y ) \right |^2 ,  
\end{equation}   
as already observed in~\cite{Ogawa_2021}. Note that this is no longer true if $\hat \Pi ^{\hat Y } _y $ is of rank 2 or larger.
\begin{proposition} \label{prop:2nd_order_exp_E1}   
    If $\varphi$ is real-valued and as $\gamma \to 0$, we have that:
    \begin{eqnarray}\label{eq:SecondOrderQCE_Q}
    \mathbb{E }_{\hat\sigma_\gamma } (\hat Q | \hat Y = y ) &=& \langle \varphi , \hat Q \varphi \rangle + {\rm Re }( \mathbb{E}_{\hat\rho}(\hat X | \hat Y = y)) \cdot \gamma \\
    && \nonumber-\gamma^2\cdot\mathrm{Re}\left(\mathbb{E}_{\hat\rho}(\hat X^2 | \hat Y = y)-\mathcal{E }_2 (\hat X ; \hat \rho , y )\right) \mathrm{cov}_{\varphi}(\hat P ^2,\hat Q) + O(\gamma ^3)
\end{eqnarray}
\end{proposition}
Note that, when $[\hat X,\hat\rho]=0$ or when $[\hat X,\hat Y]=0$, the second order coefficient in  Eq.~\eqref{eq:SecondOrderQCE_Q} is zero, meaning that  Eq.~\eqref{eq:thm_CCE_WVa_Q_2} is in fact a second order approximation in this case. We also note that, when $\hat \rho = |\psi \rangle \langle \psi | $ is a pure state and $\hat{Y } $ is a CSCO, the second order coefficient in Eq.~\eqref{eq:SecondOrderQCE_Q} can be interpreted as the real part of a conditional variance, see Eq.~\eqref{eq:alternative cond variance}.

\begin{proof}
    We have already seen that for real-valued $\varphi $,   
    \begin{equation}
        \partial_{\gamma} p_{\gamma}(q,y)|_{\gamma = 0 } = -\mathrm{Re}\left(\mathrm{Tr}(\Pi_{y}^{\hat Y}\hat X\hat \rho)\right)\partial_{q}\varphi^2(q) = -p_{0}(y)\mathrm{Re}\left(E_{1}\right)\partial_{q}\varphi^2(q) ,   
    \end{equation}   
    since $p_0 (y) = {\rm Tr } (\hat{\Pi }^{\hat{Y } } _y \hat{\rho } ) . $ 
   Differentiating (\ref{eq:proof_CCE_WV2}) twice gives 
    \begin{eqnarray} \nonumber   
        \partial_{\gamma}^{2} p_{\gamma}(q,y)|_{\gamma = 0 } &=& 2\left(\mathrm{Re}\left(\mathrm{Tr}(\Pi_{y}^{\hat Y}\hat X^2\hat \rho)\right)\varphi(q)\partial_{q}^2\varphi(q)+\mathrm{Tr}(\Pi_{y}^{\hat Y}\hat X\hat \rho\hat X) \left(\partial_{q}\varphi(q)\right)^2 \right) \\   
        &=& 2p_{0}(y)\left(\mathrm{Re}\left(E_{2}\right)\varphi(q)\partial_{q}^2\varphi(q)+\mathcal{E}\left(\partial_{q}\varphi(q)\right)^2 \right) 
    \end{eqnarray}   
    where $E_{\nu } = \mathbb{E }_{\hat \rho } (\hat X ^{\nu } | \hat Y = y ) $ and $\mathcal{E } = \mathcal{E }_2 (\hat X ; \hat \rho , y )$. The second order Taylor expansion in $\gamma $ of $p_{\gamma}(q,y)$ at $\gamma=0$
    is thus given by:
    \begin{equation} \label{eq:2nd_order_p_gamma}   
        p_{\gamma}(q,y) = p_{0}(y)\left(\varphi^2(q)-\gamma\cdot\mathrm{Re}\left(E_{1}\right)\partial_{q}\varphi^2(q) + \gamma^2\left(\mathrm{Re}\left(E_{2}\right)\varphi(q)\partial_{q}^2\varphi(q)+\mathcal{E}\left(\partial_{q}\varphi(q)\right)^2 \right)\right) + O(\gamma^3).
    \end{equation}
    Integrating this expression with respect to $q$, we find   
    \begin{equation} \label{eq:2nd_order_p_gamma(y)}
         p_{\gamma}(y) = \int_{\R}  p_{\gamma}(q,y)dq =  p_{0}(y)\left(1 + \gamma^2\left(\mathcal{E}-\mathrm{Re}\left(E_{2}\right) \right)\langle \varphi,P^2\varphi\rangle\right) + O(\gamma^3) ,   
    \end{equation}  
    since   
    \begin{equation*}
        \int_{\R} \partial _q \varphi^2(q)dq=0 \mathrm{\ and \ } \int_{\R} \varphi(q)\partial_{q}^2\varphi(q)dq = - \int_{\R} (\partial_{q}\varphi(q))^2dq = -\langle P\varphi,P\varphi\rangle = -\langle \varphi,P^2\varphi\rangle.
    \end{equation*}
    We next compute the Taylor expansion of $\int_{\R}  qp_{\gamma}(q,y)dq$:
    \begin{equation}
        \int_{\R}  qp_{\gamma}(q,y)dq = p_{0}(y)\left(\langle \varphi, Q\varphi\rangle + \gamma\cdot\mathrm{Re}(E_1)+\gamma^2\cdot(\mathcal{E}-\mathrm{Re}(E_2)\langle\varphi,PQP\varphi\rangle\right)+O(\gamma^3)
    \end{equation}
    where we used, as before, that   
    \begin{equation*}
        \int_{\R} q\partial_{q}\varphi^2(q)dq=-\int_{\R} \varphi^2(q)dq = -1 ,
    \end{equation*}
   while   
    \begin{equation*}
    \int_{\R} q\varphi(q)\partial_{q}^2\varphi(q)dq = - \int_{\R} q(\partial_{q}\varphi(q))^2dq = -\langle P\varphi,QP\varphi\rangle = -\langle \varphi,PQP\varphi\rangle.
    \end{equation*}
    Putting everything together, we finally conclude that:
    \begin{eqnarray*}
        \mathbb{E } _{\hat\sigma_\gamma} (\hat{Q } | \hat{Y } = y )&=& \frac{\langle \varphi, Q\varphi\rangle + \gamma\cdot\mathrm{Re}(E_1)+\gamma^2\cdot(\mathcal{E}-\mathrm{Re}(E_2)\langle\varphi,PQP\varphi\rangle+O(\gamma^3)}{1 + \gamma^2\left(\mathcal{E}-\mathrm{Re}\left(E_{2}\right) \right)\langle \varphi,P^2\varphi\rangle + O(\gamma^3)} \\
        &=& \langle \varphi, Q\varphi\rangle + \gamma\cdot\mathrm{Re}(E_1) + \gamma^2\left(\mathcal{E}-\mathrm{Re}\left(E_{2}\right)\right)\left(\langle \varphi,PQP\varphi\rangle-\langle \varphi,P^2\varphi\rangle\langle \varphi,Q\varphi\rangle\right)\\
        &&+ O(\gamma^3) \\
        &=& \langle \varphi, Q\varphi\rangle + \gamma\cdot\mathrm{Re}(E_1) + \gamma^2\left(\mathcal{E}-\mathrm{Re}\left(E_{2}\right)\right)\mathrm{cov}_{\varphi}(P^2,Q)+ O(\gamma^3)
    \end{eqnarray*}
    where we used that :
    \begin{equation*}
        \langle \varphi,PQP\varphi\rangle = \frac{1}{2}\left(\langle \varphi,(P^2Q+QP^2)\varphi\rangle - \langle \varphi,[P,[P,Q]]\varphi\rangle\right) = \mathrm{Re}\left(\langle\varphi,P^2Q\varphi\rangle\right).
    \end{equation*}
    This ends the proof.
    \end{proof}

We note that if $\varphi $ is even, and in particular if it is a centered Gaussian, ${\rm cov } _{\varphi } (\hat P ^2 , \hat Q ) = 0 $ and the approximation of the mean pointer position by $\gamma {\rm Re } \, \mathbb{E }_{\hat \rho } (\hat X | \hat Y = y ) $ is of third order in $\gamma . $

\subsubsection{Small \texorpdfstring{$\gamma$}{E[...]} behavior of the conditional variance of the pointer position}   
   
We examine the small-$\gamma $ behaviour of the conditional variance of $\hat{Q } $, which is given by       
\begin{equation} \label{eq:def_cond_variance_gamma}   
{\rm var }_{\hat\sigma_\gamma } (\hat{Q } | \hat{Y } = y ) := \mathbb{E }_{\hat\sigma_\gamma } (\hat{Q }^2 | \hat{Y } = y ) - \mathbb{E }_{\hat\sigma_\gamma } (\hat{Q } | \hat{Y } = y ) ^2 ,   
\end{equation}     
Equation~\eqref{eq:thm_CCE_WV_bis_Q_2} allows us to quickly compute the first order expansion in $\gamma $ of this conditional variance.       
Its derivative   
$$   
\frac{d }{d \gamma } {\rm var }_{\hat\sigma_\gamma } (\hat{Q } | \hat{Y } = y ) = \frac{d }{d \gamma } \mathbb{E }_{\hat\sigma_\gamma } (\hat{Q }^2 | \hat{Y } = y ) - 2 \mathbb{E }_{\hat\sigma_\gamma } (\hat{Q } | \hat{Y } = y ) \frac{d }{d \gamma } \mathbb{E }_{\hat\sigma_\gamma } (\hat{Q } | \hat{Y } = y )   
$$   
in $\gamma = 0 $ becomes, on using~\eqref{eq:thm_CCE_WV_bis_Q_2},  
\begin{eqnarray*}   
&{\rm Re }(E_1) \cdot 2 \langle \varphi , \hat{Q } \varphi \rangle + 2 {\rm Im }(E_1) \cdot {\rm cov }_{\varphi } (\hat{Q }^2 , \hat{P } ) - 2 \underbrace{\langle \varphi , \hat{Q } \varphi \rangle }_{\mathbb{E }_{\hat\sigma_0} (\hat{Q } | \hat{Y } = y ) } \left( {\rm Re }(E_1)  + 2 {\rm Im }(E_1 ){\rm cov }_{\varphi } (\hat{Q } , \hat{P } ) \right) \\   
&= 2 {\rm Im }(E_1) \cdot \left( {\rm cov }_{\varphi } (\hat{Q }^2 , \hat{P } ) - 2 \langle \varphi , \hat{Q } \, \varphi \rangle {\rm cov }_{\varphi } (\hat{Q } , \hat{P } ) \right),   
\end{eqnarray*}
where we recall that $E_{1}$ stands for for the quantum conditional expectation  $\mathbb{E } (\hat X| \hat Y = y ) $.
We therefore have the following corollary, providing yet another interpretation for the imaginary part of the weak value:   
   
\begin{proposition}\label{cor:VarGammaOrder1} Assuming for simplicity that the initial pointer position $\langle \varphi , \hat{Q } \varphi \rangle = 0 $, the pointer's conditional variance has the weak-interaction expansion     
\begin{equation}   
{\rm var }_{\hat\sigma_\gamma } (\hat{Q } | \hat{Y } = y ) = \langle \varphi , \hat{Q }^2 \varphi \rangle + 2 \gamma {\rm Im } \, \mathbb{E }_{\hat{\rho } } (\hat{X } | \hat{Y } = y ) \cdot {\rm cov }_{\varphi } (\hat{Q }^2 , \hat{P } ) + O(\gamma ^2 )   
\end{equation}   
\end{proposition}   

The corresponding expression when $\langle \varphi , \hat{Q } \varphi \rangle \neq 0 $ follows on replacing $\hat{Q } $ by $\hat{Q } - \langle \varphi , \hat{Q } \varphi \rangle . $   
This corollary is interesting in that the (empirical) conditional variance of $\hat{Q } $ can be directly computed from experimental pointer position data alongside with the conditional expectation of $\hat{Q } $ itself. It provides an estimator for the imaginary part of the weak value, provided that ${\rm cov }_{\varphi } (\hat{Q }^2 , \hat{P } ) \neq 0 . $ This covariance is 0 when $\varphi $ is real but, perhaps somewhat  surprisingly, also for all Gaussian wave functions $\varphi $, even if these are complex, as a direct computation shows. In fact, if $\varphi = R e^{iS } $ then   
\begin{eqnarray*}
{\rm cov }_{\varphi } (\hat{Q }^2 , \hat{P } ) &=& \int q^2 {\rm Re } (\overline{\varphi } P \varphi ) dq - \int q^2 |\varphi |^2 dq \int {\rm Re } (\overline{\varphi } P \varphi ) \\   
&=& \int q^2 \left(\partial _q S \right)R^2 dq -\int q^2 R^2 ds \int \left( \partial _q S \right) R^2 dq ,   
\end{eqnarray*}   
since ${\rm Re } (\overline{\varphi } P \varphi ) = {\rm Im } (\overline{\varphi } \partial _q \varphi ) = (\partial _q S) R^2 . $ This is clearly 0 if $\varphi $ is real-valued or if $S $ is linear and $\partial _q S $ is a constant. A complex Gaussian wave-function can be written in the form   
$$   
\varphi = C e^{- \frac{1 }{2 } \alpha (q - q_0 )^2 + i p _0 q } ,   
$$   
with ${\rm Re } \, \alpha > 0 $, where $C = C_{{\rm Re } \alpha } $ is a normalization constant and $q_0 = \langle \varphi , \hat{Q } \varphi \rangle $ and $p _0 = \langle \varphi , \hat{P } \varphi \rangle $ are the mean position respectively momentum. 
If $q_0 = 0 $ then $\partial _q S = - {\rm Im } \, \alpha \, q + p_0 $, $R^2 = e^{- {\rm Re } \alpha q^2 } $ and ${\rm cov }_{\varphi } (\hat{Q }^2 , \hat{P } ) $ is easily found to be 0. We note in passing, with reference to (\ref{eq:thm_CCE_WVa_Q_2}), that ${\rm cov }_{\varphi } (\hat{Q } , \hat{P } ) = - {\rm Im } \alpha / 2 {\rm Re }  \alpha $ for such a Gaussian.   
\medskip   

If ${\rm cov }_{\varphi } (\hat{Q }^2 , \hat{P } ) = 0 $, then ${\rm var }_{\hat\sigma_\gamma } (\hat{Q } | \hat{Y } = y ) = {\rm var }_{\varphi} (\hat{Q } ) + O(\gamma ^2 ) $, showing that  to first order the system-meter interaction does not alter the variance of the position measurements $\hat{Q } $; in particular, it does not increase it and the $O(\gamma ) $ shift of the conditional mean of $\hat{Q } $ is not accompanied by a decrease in  precision. 
   
If we can prepare the meter in an initial state for which ${\rm cov }_{\varphi } (\hat{Q }^2 , \hat{P } ) < 0 $ then the conditional variance of $\hat{Q } $ would, to first order, {\it decrease} with $\gamma $, making the  measurement of $\hat{Q } $ (marginally) more accurate after the interaction than before. 

Using the same method and Equation~\eqref{eq:thm_CCE_WV_bis_Q_2}, one can derive the  small-$\gamma $ asymptotics of order $1$ for the conditional variance ${\rm var }_{\hat\sigma_\gamma }  (\hat{P } | \hat{Y } = y )$:  we leave the details to the reader.  
\medskip   

 We will now compute the second order correction to the conditional variance of the position operator $\hat Q$.   We allow general $\varphi $ but will assume these to be real, to simplify the calculations somewhat. As we have seen, this implies in particular that ${\rm cov }_{\varphi } (\hat{Q }^2 , \hat{P } ) = 0 $.

\begin{proposition} \label{prop:2nd_order_exp_variance}   
\label{thm:var_gamma_order2} Suppose $\varphi $ is real-valued. Then   
$$   
{\rm var }_{\hat\sigma_\gamma } (\hat Q | \hat Y = y ) = {\rm var }_{\varphi } (\hat Q ) + c_2 \gamma ^2 + O(\gamma ^3 ) 
$$   
with   
\begin{eqnarray} \label{eq:var_gamma_order2}   
c_2 &=& \ \ \mathcal{E }_2 (\hat X ; \hat \rho , y ) - {\rm Re } (\mathbb{E }_{\hat \rho } (\hat X | \hat Y = y ) )^2 \ + \\     &\ & \Bigl( \mathcal{E }_2 (\hat X; \hat \rho , y ) - {\rm Re } (\mathbb{E }_{\hat \rho } (\hat X ^2 | \hat Y = y  ) ) \Bigr) \cdot\left( {\rm cov }_{\varphi }  (\hat P ^2 , \hat Q^2 )-2\mathrm{cov}_{\varphi }  (P^2,Q)\langle\varphi,Q\varphi\rangle \right) \nonumber  
\end{eqnarray}

\end{proposition}   

\begin{proof}   

The second order Taylor expansion (\ref{eq:2nd_order_p_gamma}) of $p_{\gamma } (q, y ) $ implies that   
\begin{eqnarray}   
\int q^2 p_{\gamma } (q, y ) dq &=& p_0 (y) \left( \langle \varphi, \hat{Q }^2 \varphi \rangle + 2 \langle \varphi , \hat Q \varphi \rangle {\rm Re }(E_1 ) \cdot \gamma \right. \\   
&\ & \left . + ({\rm Re }(E_2 ) + (\mathcal{E } - {\rm Re }(E_2 )) ||\hat{Q } \hat{P } \varphi ||^2 ) \cdot \gamma ^2 + O(\gamma ^3 ) \right) , \nonumber   
\end{eqnarray}   
where we used that   
$$   
\int q^2 \partial _q \varphi ^2 dq = - 2 \int q \varphi ^2 dq = - 2 \langle \varphi , \hat Q \varphi \rangle ,   
$$   
and   
$$   
\int q^2 \varphi \partial _q ^2 \varphi dq = - \int 2 q \varphi \partial _q \varphi dq - \int q^2 (\partial _q \varphi )^2 dq = \int \varphi ^2 dq - || q \partial _q \varphi ||^2 = 1 - || \hat{Q } \hat{P } \varphi ||^2 .      
$$   
Dividing by $p_{\gamma } (y) = p_0 (y) \bigl( 1 + (\mathcal{E } - {\rm Re }(E_2 ) ) \, || \hat P \varphi ||^2 \cdot \gamma ^2 + O(\gamma ^3 ) \bigr) $ (cf. (\ref{eq:2nd_order_p_gamma(y)}) we see that   
\begin{eqnarray*}   
\mathbb{E }_{\hat\sigma_\gamma } (\hat{Q }^2 | \hat{Y } = y ) &=& \langle \varphi, \hat{Q }^2 \varphi \rangle + + 2 \langle \varphi , \hat Q \varphi \rangle {\rm Re }(E_1 ) \cdot \gamma + \\   
&\ & \left( {\rm Re }(E_2 ) + (\mathcal{E } - {\rm Re }(E_2 ) )( ||\hat{Q } \hat{P } \varphi ||^2 -  \langle \varphi, \hat{Q }^2 \varphi \rangle \langle \varphi , \hat{P }^2 \varphi \rangle ) \right) \cdot \gamma ^2 \\   
&\ & + O(\gamma ^3 ) .   
\end{eqnarray*}   
Since $|| \hat Q \hat P \varphi ||^2 = \langle \varphi, \hat P \hat Q ^2 \hat P \varphi \rangle $ and    
\begin{eqnarray*}   
\frac{1 }{2 } (P^2 Q^2 + Q^2 P^2 ) &=& \frac{1 }{2 } \left( P (Q^2 P + [P, Q^2 ] ) + (P Q^2 P + [Q^2 , P ] ) P \right) = P Q^2 P + \frac{1 }{2 } [P, [P, Q^2 ] ] \\   
&=& P Q^2 P - 1   
\end{eqnarray*}   
we find that the coefficient of $\gamma ^2 $ can be re-written as $\mathcal{E } + (\mathcal{E } - {\rm Re } (E_2 ) ) {\rm cov } _{\varphi } (\hat P ^2 , \hat Q ^2 ) . $ 
The proposition now follows from the definition (\ref{eq:def_cond_variance_gamma}) of conditional variance and the second order expansion of $\mathbb{E }_{\hat \sigma _{\gamma } } (\hat X | \hat Y = y ) $ of proposition \ref{prop:2nd_order_exp_E1}.   
\end{proof}   
   
Ogawa {\it et al.} in~\cite{Ogawa_2021} computed the second order correction to 
${\rm var }_{\hat\sigma_\gamma } (\cos \alpha \, \hat Q + \sin \alpha \, \hat P ) $ for arbitrary $\alpha $, but for Gaussian meter states $\varphi $ only.  
To make the connection with~\cite{Ogawa_2021}, let us again write      
$E_{\nu } $ for $\mathbb{E } (\hat X ^{\nu } | \hat Y = y ) $ and $\mathcal{E } $ for $\mathcal{E } _2 (\hat X ; \hat \rho , y) $, and put      
\begin{equation}   
\Gamma _{\varphi } := {\rm cov }_{\varphi }  (\hat P ^2 , \hat Q^2 )-2\mathrm{cov}_{\varphi }  (P^2,Q)\langle\varphi,Q\varphi\rangle   
\end{equation}  
Ogawa {\it et al.} define the {\it weak variance of $\hat{X } $} in the post-selected state $\varphi _y $ and the preselected state $\hat \rho $, by    
\begin{equation} \label{eq:Ogawa_wv}   
\sigma ^2 _w  := \sigma ^2 _{w, \hat \rho } (\hat X | \hat Y = y ) := \mathbb{E }^{\ell } _{\hat \rho } (\hat X ^2 | \hat Y = y ) - \mathbb{E }^{\ell } _{\hat \rho } (\hat X | \hat Y = y ) ^2 = E_2 - E_1 ^2 .    
\end{equation}   
The coefficient $c_2 $ of the second order term can then be re-written in terms of this weak variance as
\begin{equation}
    c_2 = \left( \mathcal{E } - ({\rm Re } \, E_1 )^2 \right) (\Gamma _{\varphi } + 1 ) + ({\rm Im } E_1 )^2 \Gamma _{\varphi } - {\rm Re } (\sigma _w ^2 )  \Gamma _{\varphi } 
\end{equation}   
If $\hat \rho $ is pure, $\mathcal{E } = |E_1 |^2 $ by (\ref{eq:E_cal_pure_2}) and   
\begin{equation}   
c_2 = {\rm Im } (E_1 )^2 \, (2 \Gamma _{\varphi } + 1 ) - {\rm Re } (\sigma _w ^2 ) \Gamma _{\varphi } .   
\end{equation}   
If we now specialize to $\varphi (q) = \pi ^{-1/4 } e^{- q^2 / 2 } $, a centered Gaussian, then $\Gamma _{\varphi } = {\rm cov }_{\varphi } (\hat P ^2 , \hat Q ^2 ) = - \frac{1 }{2 } $ by direct computation, and we recover \cite{Ogawa_2021}'s result that   
\begin{equation}\label{eq:OgawaResult}
    {\rm var }_{\hat\sigma_\gamma } (\hat Q | \hat Y = y ) = {\rm var }_{\varphi } (\hat Q ) + \frac{1 }{2 } {\rm Re } \, \sigma ^2 _w (\hat X ; \hat \rho , y ) \ \gamma ^2 + 0(\gamma ^3 ) .   
\end{equation}   
In fact, since $\Gamma _{\varphi } $ is invariant under dilations $\varphi (q)  \to \alpha ^{1/4 } \varphi (\alpha q ) $, $\alpha > 0 $ and can be written as   
$\Gamma _{\varphi } = {\rm cov }_{\varphi } \bigl( \hat P ^2 , (\hat Q - \langle \varphi , \hat Q \varphi \rangle )^2 \bigr) $, the result holds for arbitrary Gaussians $\varphi (q) = \pi ^{-1/4 } \alpha ^{1/4 } e^{- \alpha (q - \beta )^2 / 2 } $, where we can even allow complex $\alpha $ in the right half plane, ${\rm Re } \, \alpha > 0 $, by analytic continuation.    
The weak variance (\refeq{eq:Ogawa_wv}) does not satisfy the law of total variance of section 3: if we introduce a "weak variance operator" by   
\begin{equation}   
\hat \sigma ^2_{w, \hat \rho } (\hat X | \hat Y ) = \sum _{y \in \sigma (\hat Y ) } \sigma ^2 _{w, \hat \rho } (\hat X | \hat Y = y ) \, \hat \Pi ^{\hat Y } _y ,   
\end{equation}
then its expectation will in general not be equal to ${\rm var }(X) - {\rm var }_{\hat \rho } (\mathbb{E }_{\hat \rho } (\hat X | \hat Y ) $ (except when $\mathbb{E }_{\hat \rho } (\hat X | \hat Y ) $ happens to be Hermitian), unlike the conditional variances of section 3. It is natural to try and express the coefficient $c_2 $ in terms of the latter. We cannot use $\Delta ^{2, \ell } _{\hat \rho } (\hat X | \hat Y ) $, since this is 0 for pure states, but this is easily done in terms of $\tilde \Delta _{\hat \rho } ^{2, \ell } (\hat X | \hat Y ) . $  Dropping the $\ell $ from the notation, and writing    
\begin{equation}   
\tilde \Delta ^2 := \tilde{\Delta }^2 _{\hat \rho } (\hat X | \hat Y = y ) = \mathbb{E }_{\hat \rho } (\hat X ^2 | \hat Y = y ) - | \mathbb{E }_{\hat \rho } (\hat X | \hat Y = y ) | ^2 = E_2 - |E_1 |^2 .    
\end{equation}     
we find on replacing ${\rm Re }(E_2) $ by ${\rm Re }(\tilde \Delta ) + |E_1 |^2 $ in (\ref{eq:var_gamma_order2}) that   
\begin{equation}   
c_2 = (\mathcal{E } - {\rm Re } (E_1 ) ^2 ) (\Gamma _{\varphi } + 1 ) - {\rm Im } (E_1 )^2 \Gamma _{\varphi } - {\rm Re } (\tilde \Delta ^2 ) .   
\end{equation}   
In particular, for pure states we arrive at the following corollary of proposition \ref{prop:2nd_order_exp_variance}:   

\begin{corollary} \label{cor:alt_expr_c2} If $\hat \rho $ is pure, then   
\begin{equation} \nonumber   
    {\rm var }_{\hat\sigma_\gamma } (\hat Q | \hat Y = y ) = {\rm var }_{\varphi } (\hat Q ) + c_2 \gamma ^2 + O(\gamma ^3 ) ,   
\end{equation}   
with   
\begin{equation} \label{eq:alt_expr_c2}   
    c_2 = c_2(\hat X , \hat \rho , y ) = ({\rm Im } \, \mathbb{E }^{\ell } _{\hat \rho } (\hat X | \hat Y = y ) )^2 - {\rm Re } \bigl( 
    \tilde \Delta ^2 _{\hat \rho } (\hat X | \hat Y = y ) \bigr) \cdot \Gamma _{\varphi }
\end{equation}
\end{corollary}    
\medskip   
   
An interesting class of examples generalizing the Gaussians is that of the Hermite functions $h_n = h_n (q) $, which are the eigenfunctions of the Harmonic oscillator $\hat P ^2 + \hat Q ^2 . $ Using the well-known recurrence relations for the derivative of $h_n $ and for the multiplication of $h_n $ by $q $, 
$$   
h_n ' = \sqrt{\frac{n }{2 } } h_{n - 1 } - \sqrt{\frac{n + 1 }{2 } } h_{n + 1 } , \ \ q h_n = \sqrt{\frac{n }{2 } } h_{n - 1 } + \sqrt{\frac{n + 1 }{2 } } h_{n + 1 } ,   
$$   
together with their orthonormality, one easily shows that $\langle h_n , \hat Q h_n \rangle = 0 $, $\langle h_n , \hat Q ^2 h_n \rangle = \langle \varphi, \hat P ^2 \varphi \rangle = n + \frac{1 }{2 } $ and      
$$   
\Gamma _{h_n } = {\rm cov }_{h_n } (\hat P ^2 , \hat Q ^2 ) = - \frac{1 }{2 } (n^2 + n + 1 ) .        
$$   
This implies that, for pure states again,     
\begin{eqnarray*}      
    {\rm var }_{\hat\sigma_\gamma } (\hat Q | \hat Y = y ) &=&  n + \frac{1 }{2 } + \left( {\rm Im } (E_1 )^2  + \frac{1 }{2 } (n^2 + n + 1 ) \, {\rm Re } \, \tilde \Delta _2 \right) \cdot \gamma ^2 + O(\gamma ^3 ) .       
\end{eqnarray*}

\subsection{Conditioned measurements in the case of strong system-meter interactions}

We turn to the large $\gamma $- or strong interaction limit of the conditional expectations with respect to $\sigma _{\gamma } . $ It turns out that these are still given by quantum conditional expectations, but with respect to a new state, which is the post-interaction system state:  the system not only acts on the meter, but there is also a back-action of the meter on the system. We assume for simplicity that the initial reader-wave function is compactly supported, though this is by no means necessary.   
\begin{proposition} Suppose that $\varphi $ is compactly supported. Then for sufficiently large $\gamma $,   
\begin{equation} \label{eq:_LI_CE}
\mathbb{E }_{\hat\sigma_\gamma } (\hat{Q } | \hat{Y } = y ) = \langle \varphi, \hat{Q } \varphi \rangle + \gamma \mathbb{E }_{\hat{\rho } ^{\hat{X } } } (\hat{X } | \hat{Y } = y ) ,   
\end{equation}   
where   
\begin{equation}   
\hat{\rho }^{\hat{X } } := \sum _{x \in \sigma (\hat{X } ) } \hat{\Pi }^{\hat{X } } _x \hat{\rho } \hat{\Pi }^{\hat{X } }_x = {\rm Tr }_{L^2 (\mathbb{R } ) } (\hat{\sigma } _{\gamma } ),    
\end{equation}      
is the reduced post-interaction state of the system.   
\end{proposition}   

\noindent Note that $\mathbb{E }_{\hat{\rho } ^{\hat{X } } } (\hat{X } | \hat{Y } = y ) $ is real since $\hat{X } $ and $\hat{\rho } ^{\hat{X } } $ commute. This theorem shows that the weak value as usually defined can also naturally appear in scenarios where the interaction is not weak and corresponds to a quantum conditional expectation, as in the weak regime.  
 
\begin{proof} Since $\varphi $ is compactly supported, the supports of $\varphi (q - \gamma x ) $ and $\varphi (q - \gamma x' ) $ for $x \neq x' $ will be disjoint if $\gamma $ is sufficiently large, and therefore   
\begin{equation} \label{eq:LI_CCE_WV}   
p_{\gamma } (q, y ) = \sum _{x \in \sigma (\hat{X } ) } {\rm Tr } (\hat{\Pi }^{\hat Y } _y \hat{\Pi }^{\hat{X } } _x \hat{\rho } \hat{\Pi }^{\hat{X } } _x )  |\varphi (q - \gamma x ) |^2   
\end{equation}   
Integrating over $q $,   
$$   
p_{\gamma } (y) = \sum _{x \in \sigma (\hat{X } ) } {\rm Tr } (\hat{\Pi }^{\hat Y } _y \hat{\Pi }^{\hat{X } } _x \hat{\rho } \hat{\Pi }^{\hat{X } } _x ) = {\rm Tr } (\hat{\Pi } ^{\hat{Y } } _y \hat{\rho } ^{\hat{X } } )   
$$   
while   
\begin{eqnarray*}   
\int q \, p_{\gamma } (q, y ) dq &=& \sum _{x \in \sigma (\hat{X } ) } {\rm Tr } (\hat{\Pi }^{\hat Y } _y \hat{\Pi }^{\hat{X } } _x \hat{\rho } \hat{\Pi }^{\hat{X } } _x) \int (q + \gamma x ) |\varphi (q) |^2 dq \\   
&=& {\rm Tr } (\hat{\Pi } ^{\hat{Y } } _y \hat{\rho } ^{\hat{X } } ) \langle \varphi, \hat{Q } \varphi \rangle + \gamma \sum _x x {\rm Tr } (\hat{\Pi }^{\hat Y } _y \hat{\Pi }^{\hat{X } } _x \hat{\rho } \hat{\Pi }^{\hat{X } } _x ) \\   
&=& {\rm Tr } (\hat{\Pi } ^{\hat{Y } } _y \hat{\rho } ^{\hat{X } } ) \langle \varphi, \hat{Q } \varphi \rangle + \gamma {\rm Tr } (\hat{\Pi }^{\hat Y } _y \hat{X } \hat{\rho }^{\hat{X } } ) 
\end{eqnarray*}   
Taking the quotient proves (\ref{eq:_LI_CE}).   
Concerning the interpretation of $\hat{\rho }^{\hat{X } } $, formula (\ref{eq:sigma_gamma}) for $\hat{\sigma }_{\gamma } $ implies that for sufficiently large $\gamma $,    
\begin{eqnarray*}
{\rm Tr } _{L^2 (\mathbb{R } ) }(\hat \sigma_{\gamma}) &=& \sum _{x, x' \in \sigma (\hat{X } ) } \hat{\Pi }  ^{\hat {X } } _x \hat{\rho } \hat{\Pi }^{\hat{X } } _{x' } \int _{\mathbb{R } } \varphi ( q- \gamma x ) \overline{\varphi (q - \gamma x' ) } dq \\   
&=& \sum _{x \in \sigma (\hat{X }) }  \hat{\Pi }  ^{\hat {X } } _x \hat{\rho } \hat{\Pi }^{\hat{X } } _x \int _{\mathbb{R } } |\varphi ( q- \gamma x ) |^2 dq \\   
&=& \sum _{x \in \sigma (\hat{X }) }  \hat{\Pi }  ^{\hat {X } } _x \hat{\rho } \hat{\Pi }^{\hat{X } } _x   
\end{eqnarray*}   
\end{proof}   

If $\varphi $ is not compactly supported, but if $\varphi (q) $ tends to 0 at $\pm \infty $  while $\varphi (q) $ and $q \varphi (q) $ are bounded, the theorem remains true in the form      
$$   
\lim _{\gamma \to \infty } \frac{ \mathbb{E }_{\hat\sigma_\gamma } (\hat{Q } | \hat{Y } = y ) - \langle \varphi, \hat{Q } \varphi \rangle}{\gamma } = \mathbb{E }_{\hat{\rho }_{\hat{X } } } (\hat{X } | \hat{Y } = y ) .     
$$

We also mention two other interpretations of $\hat{\rho }^{\hat{X } } $: it is Umegaki conditional expectation of $\hat{\rho } $ relative to $\mathcal{F }_{\mathbb{C } , \hat{X } } $ and also the conditional expectation of $\hat{\rho } $ given $\hat{X } $ with respect to the maximally mixed state.   
\medskip   

The large $\gamma $-behaviour of the pointer momentum $\hat{P } $ is very different: if $\gamma $ is sufficiently large, then   
\begin{equation}
\mathbb{E }_{\hat\sigma_\gamma } (\hat{P } | \hat{Y } = y ) = \langle \varphi , \hat{P } \varphi \rangle ,   
\end{equation}   
independently of $\gamma $, $\hat{X } $ and $\hat{Y } $, 
as follows by taking the trace of $\hat{P } \hat{\Pi }^{\hat{Y} } \hat{\sigma }_{\gamma } $ with $\hat{\sigma }_{\gamma } $ given by (\ref{eq:sigma_gamma}) and using the disjointedness of the supports of the different translates of $\varphi . $ This result is the same as for unconditional $\hat{P } $-measurements, for which $\mathbb{E }_{\hat\sigma_\gamma } (\hat{P } ) = \langle \varphi , \hat{P } \varphi \rangle = \mathbb{E }_{\sigma_0} (\hat{P } ) $, indicating that there is no net momentum transfer from the system to the meter and vice versa.

\section{Quantum conditional expectation \textit{versus} von Neumann algebraic conditional expectation}\label{app:vonNeumann}

 According to the theory of von Neumann algebras, following Umegaki~\cite{umegaki1954}, a conditional expectation of $\mathcal{L}(\mathcal{H } ) $ given a subalgebra $\mathcal{A } $ is a trace-preserving projection map $\mathcal{E }_{\mcl{A}} : \mathcal{L}(\mathcal{H } ) \to \mathcal{A } $ which satisfies $\mathcal{E } _{\mcl{A}} (A_1 X A_2 ) = A_1 \mathcal{E }_{\mcl{A}}  (X) A_2 $ for all $A_1, A_2 \in \mathcal{A }$ (this notion is in fact introduced for type II von Neumann algebras in~\cite{umegaki1954}). The map $\mcl{E}_{\mcl{A}} $ is unique and, in the finite dimensional case, given by the orthogonal projection with respect to the Hilbert-Schmidt inner product, as can be easily verified. If $\mathcal{A } = \CompFuncObs{\hat Y}$ for a CSCO $\hat Y $, it is given by:
\begin{equation}\label{eq:VNexpectation}
     \mathcal{E}_{\hat Y}(\hat X):=\mcl{E}_{\CompFuncObs{\hat Y}} (\hat X)= \sum _{y \in \sigma (\hat{Y } ) } \hat{\Pi } ^{\hat{Y } } _y \hat{X } \hat{\Pi }^{\hat{Y } } _y.
\end{equation}

The left quantum conditional expectation $\mathbb{E }^{\ell}_{\hat{\rho } } (\hat{X } | \hat{Y } ) $ can be expressed in terms of the von Neumann algebraic conditional expectation $\mathcal{E }_{\hat{Y } } : \mathcal{L}(\mathcal{H } ) \to \mathcal{F }_{\mathbb{C } ,\hat{Y } } $ as   
\begin{equation}\label{eq:leftvsVNexp}   
\mathbb{E }^{\ell}_{\hat{\rho } } (\hat{X } | \hat{Y } ) = \mathcal{E }_{\hat{Y } } ( \hat{X } \hat{\rho } ) \mathcal{E }_{\hat{Y } } (\hat{\rho } )^{-1 }   
\end{equation}   
where $\mathcal{E }_{\hat{Y } }(\hat{\rho}) $ is invertible since $\hat{\rho } \in D_{\hat{Y } } $ (cf. \cite{Brummelhuis2025a} for an analogous formula for $\mathbb{E }^{\rm sa }_{\hat{\rho } } (\hat{X } | \hat{Y } ) $ when $\hat{X } $ is self-adjoint). This expression can be verified by direct computation or by using Theorem \ref{thm:UniqueEcond}: the right hand side  of Eq.\eqref{eq:leftvsVNexp} clearly satisfies Eq.~\eqref{eq:left invariance}, and Eq.~\eqref{eq:iterated expectation} is checked as follows:    
\begin{align*}   
{\rm Tr } (\mathcal{E }_{\hat{Y } } ( \hat{X } \hat{\rho } ) \mathcal{E }_{\hat{Y } } (\hat{\rho } )^{-1 } \hat{\rho } ) &=   
{\rm Tr } (\mathcal{E }_{\hat{Y } }  \bigl(\mathcal{E }_{\hat{Y } } ( \hat{X } \hat{\rho } ) \mathcal{E }_{\hat{Y } } (\hat{\rho } )^{-1 } \hat{\rho } ) \bigr) \\
&= {\rm Tr } \bigl( \mathcal{E }_{\hat{Y } } ( \hat{X } \hat{\rho } ) \mathcal{E }_{\hat{Y } } (\hat{\rho } )^{-1 } \mathcal{E }_{\hat Y} (\hat{\rho } ) \bigr) \\
&= {\rm Tr } \bigl( \mathcal{E }_{\hat{Y } } ( \hat{X } \hat{\rho } ) \bigr) \\
&= {\rm Tr } (\hat{X } \hat{\rho } ). 
\end{align*}  
One similarly checks that $\mathbb{E}^r_{\hat \rho}(\hat X|\hat Y)=\mathcal{E}_{\hat{Y}}(\hat \rho \hat X)\mathcal{E}_{\hat Y}(\hat \rho)^{-1}$. 

It is legitimate to wonder if, given a CSCO $\hat Y$ and a state $\hat \rho$, it is possible that the left and right conditional expectations coincide, for all $\hat X$. The answer  is provided by the following proposition. 

\begin{proposition}\label{prop:leftequalsright}  Let $\hat Y$ be a $\mathrm{CSCO}$ and let $\hat\rho\in D_{\hat Y}$. 
 Then the following are equivalent:   
\begin{enumerate}  
 \item [(i)] $\mathbb{E }^{\ell } _{\hat{\rho } } (\hat{X } | \hat{Y } ) = \mathbb{E }^r _{\hat{\rho } } (\hat{X } | \hat{Y } ) $, for all $\hat{X } $;
\item [(ii)] The left quantum conditional expectation $\mathbb{E }^{\ell } _{\hat{\rho } } (\cdot | \hat{Y } ) $ has the "right pull-out property" given by Eq.~\eqref{eq:right invariance}; 
\item [(iii)] $\hat{Y } $ and $\hat{\rho } $ commute.  
\end{enumerate}   
In that case, one has
\begin{equation} \label{eq:left_right_commute}
\mathbb{E } ^{\ell } _{\hat{\rho } } 
(\hat{X } | \hat{Y } ) = \mathcal{E}_{\hat Y}(\hat X)= \sum _{y \in \sigma (\hat{Y } ) } \hat{\Pi } ^{\hat{Y } } _y \hat{X } \hat{\Pi }^{\hat{Y } } _y=  \sum _{y \in \sigma (\hat{Y } ) } \langle \varphi_y, \hat{X }\varphi_y\rangle \hat{\Pi }^{\hat{Y } } _y.
\end{equation}
\end{proposition} 
We point out that Eq.~\eqref{eq:left_right_commute} is independent of the state $\hat \rho\in D_{\hat Y}$; it is in particular the conditional expectation of the maximally mixed state $\hat{\rho } = \frac{1 }{d }\id{d}$, where $d$ is the dimension of $\mathcal{H} . $ 
\begin{proof}
That (i) implies (ii) is immediate. We prove that (ii) implies (iii). If $\mathbb{E }_{\hat{\rho } }^{\ell} (\hat{X } | \hat{Y } ) $ satisfies (\ref{eq:right invariance}) for all $\hat{X } $ and $f $, then it is equal to the right quantum conditional expectation, by the characterization of Theorem~\ref{thm:UniqueEcond}. This means that for all $y $ and all $\hat{X } $,   
$$   
{\rm Tr } ( \hat{X } \hat{\rho } \hat{\Pi }^{\hat{Y } } _y ) = {\rm Tr } (\hat{X } \hat{\Pi } ^{\hat{Y } } _y \hat{\rho } )    
$$   
which implies that $\hat{\rho } \hat{\Pi }^{\hat{Y } } _y = \hat{\Pi } ^{\hat{Y } } _y \hat{\rho } $, i.e. $\hat{\rho } $ commutes with all spectral projectors of $\hat{Y } $ and therefore with $\hat{Y }$ . To show that (iii) implies (i) it suffices to note that  if $\hat{\rho } $ and $\hat{Y } $ commute, then the left- and right quantum conditional expectations coincide.

Finally, to show Eq.~\eqref{eq:left_right_commute}, we note that, as $\hat Y$ and $\hat \rho$ commute, $(\varphi ^{\hat{Y } } _y)_{y\in\sigma(\hat{Y})} $ are also eigenfunctions of $\hat{\rho } $ and their associated eigenvalues are not $0$ as $\hat{\rho } \in D_{\hat{Y } } $. Consequently, for all $y\in\sigma(\hat Y)$:  
\begin{equation} \label{eq:left_right_nondegenerate}
\frac{\langle \varphi ^{\hat{Y } } _y , \hat{X } \hat{\rho } \varphi ^{\hat{Y } } _y  \rangle}{\langle \varphi ^{\hat{Y } } _y, \hat{\rho } \varphi ^{\hat{Y } } _y \rangle } =  \langle \varphi ^{\hat{Y } } _y , \hat{X }  \varphi ^{\hat{Y } } _y  \rangle\frac{\langle \varphi ^{\hat{Y } } _y, \hat{\rho } \varphi ^{\hat{Y } } _y \rangle}{\langle \varphi ^{\hat{Y } } _y, \hat{\rho } \varphi ^{\hat{Y } } _y \rangle } =
\langle \varphi ^{\hat{Y } } _y , \hat{X } \varphi ^{\hat{Y } } _y \rangle .
\end{equation}
Thus, we obtain that :
\begin{equation*}
    \mathbb{E } ^{\ell } _{\hat{\rho } } (\hat X | \hat{Y } ) = \sum_{y\in\sigma(\hat Y)} \frac{\langle \varphi ^{\hat{Y } } _y , \hat{X } \hat{\rho } \varphi ^{\hat{Y } } _y  \rangle}{\langle \varphi ^{\hat{Y } } _y, \hat{\rho } \varphi ^{\hat{Y } } _y \rangle } \hat \Pi_{y}^{\hat Y} = \sum_{y\in\sigma(\hat Y)} \langle \varphi ^{\hat{Y } } _y , \hat{X } \varphi ^{\hat{Y } } _y \rangle \hat \Pi_{y}^{\hat Y} = \sum _{y \in \sigma (\hat{Y } ) } \hat{\Pi } ^{\hat{Y } } _y \hat{X } \hat{\Pi }^{\hat{Y } } _y.
\end{equation*}
\end{proof}
We finish by mentioning that, as a consequence of this result, the only state $\hat \rho$ that satisfies 
\begin{equation*}
    \mathbb{E}^{\ell}_{\hat \rho}(\hat X|\hat Y)=\mathbb{E}^{r}_{\hat \rho}(\hat X|\hat Y), \quad \forall \hat X\in \mathcal{L}(\mathcal{H}),\quad \forall \rm \hat Y \quad \rm CSCO,
\end{equation*}
is the maximally mixed state $\hat \rho=\frac{1}{d}\mathbb{I}_d . $ Indeed from Proposition \ref{prop:leftequalsright}, such a state must commute with all CSCOs, hence it must commute with all rank one projectors, thus with all of $\mathcal{L}(\mathcal{H})$. This implies that it is a multiple of the identity.

\section{Comparing the Wigner and KD distributions}\label{app:Wigner_KD}
 To illustrate our findings, and in particular the distinguishing features of the KD representations, we compare here the KD representation and the Gross-Wigner representation of a single qutrit with Hilbert space $\mathbb{C}^3$. 
 
 For that purpose we introduce  the ``position operator'' $\hat{Q} = 0\ket{0}\bra{0} + 1\ket{1}\bra{1} + 2\ket{2}\bra{2}$
    with eigenbasis  
   \begin{equation*}
       \ket{0}=
       \begin{pmatrix}
           1\\
           0\\
           0
       \end{pmatrix}
       ,\:\ket{1}=
       \begin{pmatrix}
           0\\
           1\\
           0
       \end{pmatrix}
       ,
       \:\ket{2}=
       \begin{pmatrix}
           0\\
           0\\
           1
       \end{pmatrix}
       ,
   \end{equation*}
   as well as the conjugate  ``momentum'' operator $\hat{P} = 0\ket{\hat{0}}\bra{\hat{0}} + 1\ket{\hat{1}}\bra{\hat{1}} + 2\ket{\hat{2}}\bra{\hat{2}}$, with eigenbasis
   \begin{equation*}
       \ket{\hat{0}}=\frac{1}{\sqrt{3}}
       \begin{pmatrix}
           1\\
           1\\
           1
       \end{pmatrix}
       ,\:\ket{\hat{1}}=\frac{1}{\sqrt{3}}
       \begin{pmatrix}
           1\\
           j\\
           j^2
       \end{pmatrix}
       ,
       \:\ket{\hat{2}}=\frac{1}{\sqrt{3}}
       \begin{pmatrix}
           1\\
           j^2\\
           j
       \end{pmatrix}
       ,\end{equation*}
   where $j=e^{i\frac{2\pi}{3}}$.   
   The transition matrix between these two bases is provided by the discrete Fourier transform in $3$ dimensions:
    \begin{equation*}
        U=\frac{1}{\sqrt{3}}
       \begin{pmatrix}
            1 & 1 & 1\\
            1 & j & j^2\\
            1 & j^2 & j
        \end{pmatrix}
        .\end{equation*}
We will consider the KD representation of quantum mechanics on $\mathcal H=\mathbb C^3$ with $\hat A=\hat P$ and $\hat B=\hat Q$. 
It was proven in~\cite{langrenez2023characterizing2} that in this case $\hat X\in V_{\mathrm{KD},\mathrm{real}}^{\mathrm{sa}}$ if and only if $\hat X=f(\hat P)+g(\hat Q)$. In other words, any such $\hat X$ can be written 
\begin{equation}
    \hat X =\sum_{i=1}^3 f_i|\hat i\rangle\langle \hat i|+\sum_{j=1}^3 g_j|j\rangle\langle j|, \quad  f_i, g_j\in\R.
\end{equation}
The same holds true for qudits, when $d$ is a prime number~\cite{langrenez2023characterizing2, debievreetal2025a}.

On the other hand, we recall the definition of the Gross-Wigner representation, defined in \cite{gross2006}. Let $\Lambda=\sigma(\hat P)\times \sigma(\hat Q)=\{0,1,2\}^2$. Consider the operators
\begin{equation}
   \label{eq:Gross Wigner frames}
       S_{p,q}^{\textrm W}=\frac{1}{3}\sum_{0\leq p',q'\leq 2}e^{i\frac{2\pi}{3}(pq'-qp')}e^{i\frac{\pi}{3}p'q'}\hat{z}(p')\hat{x}(q'),\quad 0\leq p,q\leq 2,
   \end{equation}
   where $\hat{z}(p),\hat{x}(q)$ are generalized Pauli operators for $(p,q)\in \Lambda$, satisfying
   \begin{equation*}
       \hat{x}(q)\ket{q'}=\ket{q'+q},\quad \hat{z}(p)\ket{q'}=e^{i\frac{2\pi}{3}pq'}\ket{q'},
   \end{equation*}
   for all $q'\in \{0,1,2\}$. It is readily checked that the family $(S_{p,q}^{\rm W})_{0\leq p,q\leq 2}$ is a frame, whose dual frame is given by
   \begin{equation}
       \quad T_{p,q}^{\rm W}=3S_{p,q}^{\rm W}, \quad 0\leq p,q\leq 2.
   \end{equation}
The Gross-Wigner function of a state $\hat{\rho}$ is the associated quasiprobability, given by
   \begin{equation*}
       W_{p,q}(\hat{\rho})=\Tr\left(\hat{\rho}\left(S_{p,q}^{\rm W}\right)^{\dagger}\right)=\frac{1}{3}\sum_{0 \leq\xi\leq2}e^{-i\frac{2\pi}{3}\xi p}\hat{\rho}_{q+\frac{\xi}{2},q-\frac{\xi}{2}}, \quad 0\leq p,q\leq 2.
   \end{equation*}
   The associated symbol of an operator $\hat X\in \mathcal{L}(\mathcal{H})$ is given by 
   \begin{equation*}
       \tilde{W}_{p,q}(\hat X)=\Tr\left(\hat{X}\left(T_{p,q}^{\rm W}\right)^{\dagger}\right)=\sum_{\xi}e^{-i\frac{2\pi}{3}\xi p}\hat{X}_{q+\frac{\xi}{2},q-\frac{\xi}{2}}=3W_{p,q}(\hat{X}), \quad 0\leq p,q \leq 2.
   \end{equation*}

 Theorem~\ref{thm:main_hatY} implies that the conditional expectation $\mathbb E_{\hat\rho}^{\mathrm W}(\hat X|\hat Q)$ does not coincide with 
$\mathbb E_{\hat\rho}^{\ell/\mathrm r}(\hat X|\hat Q)$ for all $\hat X$ and $\hat \rho$. Indeed, the frame of Eq.\eqref{eq:Gross Wigner frames} is not composed of rank one operators. For example, $S^{\rm W}_{0,0}$ is the parity operator satisfying $S_{0,0}^{\rm W}\ket{q}=\ket{-q}$, and is therefore of full rank. We now illustrate this fact with a specific example. We have 
   \begin{equation}
       \mathbb{E}_{\hat{\rho}}^{\textrm W}(\hat{X}|\hat{Q})=\sum_{p,q=0}^2\overline{\tilde{W}_{p,q}(\hat{X}^{\dagger})}\frac{W_{p,q}(\hat{\rho})}{\bra{q}\hat{\rho}\ket{q}}\hat{\Pi}_q^{\hat{Q}}.
   \end{equation}
    It is easily checked that the Gross-Wigner symbol satisfies $\tilde{W}_{p,q}(X^{\dagger})=\overline{\tilde{W}_{p,q}(X)}$ for any $X\in \mathcal{L}(\mathcal{H})$ and any $p,q\in \Lambda$. 
   It follows that
   \begin{equation}
       \mathbb{E}_{\hat{\rho}}^{\textrm W}(\hat{X}|\hat{Q})=\sum_{p,q=0}^2\tilde{W}_{p,q}(\hat{X})\frac{W_{p,q}(\hat{\rho})}{\bra{q}\hat{\rho}\ket{q}}\hat{\Pi}_q^{\hat{Q}}.
   \end{equation}
  Consider the state $\hat{\rho}$ given in the position basis by
   \begin{equation*}
       \hat{\rho}= \frac{1}{4}
       \begin{pmatrix}
           1 & 1 & 0\\
           1 & 2 & 0\\
           0 & 0 & 1
       \end{pmatrix}.
   \end{equation*}
One readily checks that $\hat \rho\in D_{\hat Q}$. The Gross-Wigner function of $\hat{\rho}$ is then given by
   \begin{equation*}
       (W_{p,q}(\hat{\rho}))_{0\leq p,q\leq 2}=\frac{1}{12}
       \begin{pmatrix}
           1 & 2 & 3\\
           1 & 2 & 0\\
           1 & 2 & 0
       \end{pmatrix}.
   \end{equation*}
The KD distribution of $\hat \rho$ is
     \begin{equation*}
       (Q^{\rm KD,\ell}_{a,b}(\hat \rho))_{0\leq a,b\leq 2}=\frac{1}{12}
       \begin{pmatrix}
           2 & 3&1\\
           1+j^2 & 2+j & 1\\
           1+j & 2+j^2 & 1
       \end{pmatrix}.
   \end{equation*} 
Note that $\hat \rho$ is Wigner-positive; in fact, it can be checked that $\hat\rho$ is  a convex combination of stabilizer states  but not KD-positive. We recall that the pure Wigner-positive  states are the so-called \textit{stabilizer} states and that the convex hull of the stabilizer states does not exhaust the set of all Gross-Wigner-positive states~\cite{gross2006}. We note that, in this particular case, the set of KD-positive states is the convex hull of the basis states~\cite{debievreetal2025a} and thus all KD-positive states are convex combinations of stabilizer states. Note that this is  no longer true for any system of $n$ qubits~\cite{thio2025}, for $n\geq 2$. Thus, the state $\hat \rho$ we consider is in the convex hull of the stabilizer states but outside of the convex-hull of pure KD-positive states.

 Now let $\hat X=f(\hat P)$ for $f:\sigma( \hat{P})=\{0,1,2\}\to \mathbb{C}$. Then its Gross-Wigner function is  
\begin{equation}\label{eq:wignerexample}
       (\tilde W_{p,q}(f(\hat{P})))_{0\leq p,q\leq 2}=
       \begin{pmatrix}
           f(0) & f(0) & f(0)\\
           f(1) & f(1) & f(1)\\
           f(2) & f(2) & f(2)
       \end{pmatrix}.
   \end{equation}
It follows from Eq.~\eqref{eq:wignerexample} that
\begin{equation}
\label{eq:GW Econd for f(p)}
    \mathbb{E}^{\rm W}_{\hat \rho}(f(\hat P)|\hat Q)=\frac{1}{3}(f(0)+f(1)+f(2))\left( \proj{0}+\proj{1}\right) + f(0)\proj{2}.
\end{equation}
Note that, when $f$ is real-valued, this is a self-adjoint operator. In fact, as the Gross-Wigner function of self-adjoint operators is real, the Gross-Wigner conditional expectation $\mathbb{E}_{\hat \rho}^{\rm W}(\hat X|\hat Q)$ is self-adjoint when $\hat{X}$ is self-adjoint. On the other hand, 
it is easily computed that
\begin{align}
    \mathbb{E}_{\hat{\rho}}^{\ell}(f(\hat P)|\hat Q)=&\frac{1}{3} (2f(0)+(1+j^2)f(1) + (1+j)f(2))\proj{0} \nonumber\\
    +& \frac{1}{6}(3f(0)+(2+j)f(1)+(2+j^2)f(2))\proj{1} \nonumber\\
    +& \frac{1}{3}(f(0)+f(1)+f(2))\proj{2}.\label{eq:example_condex}
\end{align}
Hence the Gross-Wigner-conditional expectation is indeed different from the left conditional expectation and therefore also from the left KD-conditional expectation. In particular, the latter  may have a non-zero imaginary part even if $f$ is a real valued function as is readily seen from Eq.~\eqref{eq:example_condex}.

To illustrate that the Gross-Wigner conditional expectation doesn't satisfy the pull-through property, we consider the function $f(p)=\sin\left(\frac{2\pi}{3}p\right)$ and the operator $\hat{X}=f(\hat P)$. From Eq.\eqref{eq:GW Econd for f(p)}, one has
\begin{equation*}
    \mathbb{E}^{\rm W}_{\hat \rho}(\hat X|\hat Q)=0.
\end{equation*}
On the other hand, one computes
\begin{equation*}
    (\tilde W_{p,q}(\hat{Q}\hat{X}))_{0\leq p,q\leq 2}=
       \frac{i\sqrt{3}}{4}
       \begin{pmatrix}
           -2 & 1 & 1\\
           1 & j^2 & j\\
           1 & j & j^2
       \end{pmatrix}
       ,
\end{equation*}
from which follows
\begin{equation*}
    \mathbb{E}^W_{\hat{\rho}}(\hat{Q}\hat{X}|\hat{Q})=\frac{i\sqrt{3}}{4}\proj{2}.
\end{equation*}
Hence
\begin{equation*}
    \mathbb{E}^W_{\hat{\rho}}(\hat{Q}\hat{X}|\hat{Q})\neq \hat{Q} \mathbb{E}_{\hat{\rho}}^W(\hat{X}|\hat{Q}),
\end{equation*}
which shows that the left pull-through property is not satisfied by the Gross-Wigner conditional expectation. One could similarly check that the right pull-through property isn't satisfied either.
 While the conditional expectations associated to the Wigner and Kirkwood-Dirac representations are different, they may share similarities for specific observables or states. To illustrate this fact, let us briefly go to the infinite dimensional setting, with $\mathcal{H}=L^2(\mathbb{R})$. Let us denote by $W$ the standard Wigner distribution and by $Q^{\mathrm{KD}}$ the Kirkwood-Dirac distribution associated to the position and momentum operators $\hat Q$ and $\hat P$. It is easily computed that for any pure state $\psi\in L^2(\mathbb{R})$ and any $q\in \mathbb{R}$,
\begin{equation*}
    \mathbb{E}_{\psi}^{W}(\hat P|\hat Q=q)=\frac{\mathrm{Re}\left(i\psi(q)\overline{\psi'(q)}\right)}{|\psi(q)|^2}, \quad \mathbb{E}_{\psi}^{Q^{\mathrm{KD}}}(\hat P|\hat Q=q)=\frac{i\psi(q)\overline{\psi'(q)}}{|\psi(q)|^2},
\end{equation*}
so that the Wigner conditional expectation of $\hat P$ knowing $\hat Q$ is exactly is the self-adjoint part of the Kirkwood-Dirac conditional expectation. This property does not hold in general, as
\begin{equation*}
    \mathbb{E}_{\psi}^{W}(\hat P^2|\hat Q=q)=\frac{|\psi'(q)|^2-\mathrm{Re}\left(\psi''(q)\overline{\psi(q)}\right)}{2|\psi(q)|^2}, \quad \mathbb{E}_{\psi}^{Q^{\mathrm{KD}}}(\hat P^2|\hat Q=q)=\frac{-\psi''(q)\overline{\psi(q)}}{|\psi(q)|^2},
\end{equation*}
hence $\mathbb{E}_{\psi}^W(\hat P^2|\hat Q)\neq \mathrm{Re}\left(\mathbb{E}_{\psi}^{Q^{\mathrm{KD}}}(\hat P^2|\hat Q)\right)$ in general.

\section{Interpretation of phase insensitivity}\label{app:vanishingFI}
In this Appendix we briefly explore the physical and operational meaning of a vanishing Fisher information.  We will in particular give an  alternative interpretation of the Fisher information that is independent of parameter estimation and the Cramer-Rao bound and as such useful in the following discussions.

 First, remark that the Cramer-Rao bound seems to indicate that a vanishing Fisher information implies that the variance of any estimator of $\theta$ is infinite. This however, is not possible since the set $\sigma(\hat Y)$ is finite, so that any random variable on $\sigma(\hat Y)$ necessarily has a finite variance. In fact, $\mathrm{I}_{\mathrm F}(\hat Y;0)=0$ is equivalent to $\partial_\theta p(y, 0)=0$, for all $y\in\sigma(\hat Y)$. But, for an unbiased estimator $\tilde\theta(\hat Y)$, one has
\begin{equation}
    \sum_y \tilde\theta(y)p(y;\theta) =\theta,
\end{equation}
for all $\theta$ in a neighbourhood of $\theta=0$. Taking a derivative with respect to $\theta$ on both sides of this equation, one immediately concludes that such an unbiased estimator does not exist if $\partial_\theta p(y, 0)=0$, for all $y\in\sigma(\hat Y)$. In other words, one concludes that repeated  measurements of $\hat Y$ cannot be used to produce an unbiased estimate on the phase $\theta$ of the states $\hat\rho_{\hat X}(\theta)$.

A different interpretation of the Fisher information, not directly related to parameter estimation, sheds light on this situation, and further justifies the above definition of phase-insensitive state. We consider the set of all probabilities $(p_y)_{y\in\sigma(\hat Y)}$ on $\sigma(\hat Y)$, which form a convex polytope. We will suppose $p_y\neq 0$ for all $y$, meaning that $p_y$ does not lie on any of the facets of the probability polytope. When such a probability is experimentally determined, it is  natural to  think of the experimental errors $(\delta p_y)_{y\in\sigma(\hat Y)}$ as forming a tangent vector to the set of probability distributions on the spectrum of $\hat Y$, at the probability distribution $(p_y)_{y\in\sigma(\hat Y)}$. Given two such tangent vectors $(\delta p_y)_{y\in\sigma(\hat Y)}, (\delta p_y')_{y\in\sigma(\hat Y)}$ , one may consider the Riemannian metric ({\it i.e.} the symmetric positive definite form)
\begin{equation}
    \langle(\delta p), (\delta p')\rangle_{\mathrm{F}}=\sum_{y\in\sigma(\hat Y)}\delta p_y \frac1{p(y)}\delta p_y',
\end{equation}
known as the Fisher metric. This transforms the probability polytope into an Riemannian manifold. The square of the Fisher norm of $\delta p_y$ is naturally interpreted as the mean squared relative error of the experimentally determined probabilities:
\begin{equation}
     \|\delta p\|^2_{\textrm{F}}=\langle \delta p,\delta p\rangle_{\mathrm{F}}=\sum_{y\in\sigma(\hat Y)} \left(\frac{\delta p_y}{p(y)}\right)^2 p(y).
\end{equation}
Next, we consider the curve $\theta\mapsto (p(y;\theta))_y$, where $p(y;\theta)=\Tr(\hat\rho_{\hat X}(\theta) \hat \Pi^{\hat Y}_y)$. One then sees that 
\begin{equation} 
\mathrm{I}_{\mathrm{F}}(\hat Y;\theta)=\left| \frac{\mathrm d p}{\mathrm d \theta}\right|^2_{\mathrm{F}}.
\end{equation}
This interprets the Fisher information as the squared relative rate of change of the probabilities $p(y;\theta)$, as a result of 
a variation in the phase $\theta$. As such, $\mathrm{I}_{\mathrm{F}}(\hat Y; \theta)$ measures the sensitivity of $p(y;\theta)$ to variations of $\theta$. In fact, for small $\delta\theta$, one has that 
\begin{equation}
\|p(\theta+\delta\theta)-p(\theta)\|_{\mathrm F}\simeq \mathrm{I_F}(\hat Y;\theta)^{1/2}\delta\theta.
\end{equation}
One concludes that, for a variation $\delta \theta$ to induce an experimentally noticeable change in the relative changes of the probabilities $p(y;\theta)$, this variation must exceed the mean square relative experimental error, meaning that 
\begin{equation}
\delta \theta^2 \mathrm{I}_{\mathrm{F}}(\hat Y;\theta)\geq  \|\delta p\|^2_{\mathrm F},\quad \textrm{or}\quad \delta \theta^2 \geq \delta\theta_{\textrm{min}}^2:= \frac{\|\delta p\|^2_{\mathrm F}}{ \mathrm{I}_{\mathrm{F}}(\hat Y;\theta)} .
\end{equation}
In other words, only when $\delta\theta$ exceeds $\delta\theta_{\min}$ is one sure that the probability distribution $p(y,\theta+\delta\theta)$ leaves the ball of Fisher radius $\|\delta p\|_{\mathrm F}$ around $p(y;\theta)$, thereby making the change in $\theta$ experimentally detectable. From this point of view, the probability distributions $p(y;\theta)$ are indeed phase-insensitive when the Fisher information is small: $\delta\theta_{\min}$ is then large. This interpretation does not refer to the estimation of $\theta$.

Conversely, one may also conclude that, given the experimental precision $\delta p_y$,  it is not possible to infer $\theta$ with a precision smaller than $\delta\theta_{\mathrm{min}}$. In other words, when the Fisher information is large, a small change in the phase $\theta$ will be experimentally noticeable in variations of the Born probabilities $p(y;\theta)$. Under these same circumstances, one can also expect to determine $\theta$ with high accuracy: this is the idea behind the Cramer-Rao bound. 

If $\mathrm{I_F}(\hat Y;\theta)=0$, on the other hand, $\partial_\theta p(y,0)$ vanishes for all $y$: variations of $\theta$ have then little effect on the measurement outcome probabilities $p(y,\theta)$. A second order Taylor expansion of $p(y,\theta) $ then yields analogously 
\begin{equation}
    \delta\theta^2\geq \delta \theta_{\textrm{min}}^2:=\frac{2\|\delta p\|_{\mathrm F}}{J(\hat Y;\theta)},
\end{equation}
where 
\begin{equation}
    J(\hat Y;\theta)=\sum_{y\in\sigma(\hat Y)}\partial_\theta^2 p(y,\theta)\frac1{p(y,\theta)}\partial_\theta^2 p(y,\theta),
\end{equation}
assuming not all $\partial_\theta^2 p(y,\theta)$ vanish. Note that in this case $\delta\theta_{\mathrm{min}}^2$ behaves like $\|\delta p\|_{\mathrm F}$ and is therefore larger than when the Fisher information does not vanish. The state is therefore less sensitive to variations in $\theta$ when the Fisher information vanishes than when it does not.

It should be noted that, in what precedes, we have only invoked the classical Fisher information of the $\theta$-dependent probabilities $p(\theta, y)$. We briefly elaborate on the link with the quantum Fisher information of the states $\hat\rho_{\hat X}(\theta)$ in Appendix~\ref{app:Fisher} where we will use the techniques developed here concerning the quantum conditional expectation to give a simple proof of the result of Braunstein and Caves~\cite{braunsteincaves1994} linking the definition of quantum Fisher information of Helstrom and Holevo~\cite{helstrom1976, holevo1982} to the optimization of classical Fisher informations over all possible quantum measurements.

\section{Quantum Fisher information} \label{app:Fisher}

Quantum conditional expectations can be used to give a quick proof of the Braunstein and Caves theorem expressing the Quantum Fisher information as a supremum of classical Fisher informations. To be able to do so in full generality we will have to extend the definition of the Fisher information $\mathrm{I}_{\mathrm{F}} (\hat Y;\theta ) $ to cases when $\hat{\rho } \notin D_{\hat{Y } } $ and also revisit the quantum conditional expectations for such $\hat{\rho } . $

First of all, if $\hat \rho (\theta ) $ is an arbitrary $\mathcal{C}^1 $-family of density operators on $\mathcal{H } $, not necessarily of the form considered in section 3, and if $\hat{Y } $ is an arbitrary observable, we define the Fisher information by   
\begin{equation}   
\mathrm{I}_{\mathrm{F}} (\hat Y;\theta ) := \sum _{y : p(y ; \theta ) \neq 0 } \frac{ (\partial _{\theta } p (y ; \theta ) )^2 }{p(y ; \theta ) }  
\end{equation}   
where $p(y; \theta ) := {\rm Tr } (\hat{\Pi }^{\hat{Y } } _y \hat{\rho } (\theta ) ) $ and the sum extends over elements $y $ of the spectrum $\sigma (\hat{Y }) $, here and below.      
   
\begin{rmk} That this is the correct definition of the Fisher information when $\hat{\rho } \notin D_{\hat{ Y } } $ follows from the continued validity of the Cram\`er-Rao inequality for the variance of an unbiased estimator $\theta (\hat{Y } ) $ for $\theta $, which itself is a consequence of the identity $\sum _{y: p(y; \theta ) \neq 0 } \bigl( \theta (y) - \theta \bigr) \partial _{\theta } p(y; \theta ) = 1 . $ The latter follows as usual from differentiating the unbiasedness condition $\sum _y \theta (y) p(y ; \theta ) = \theta $ and using that $\sum _y \partial _{\theta } p (y ; \theta ) = 0 . $ Since $p(y, \theta ) = 0 $ implies that $\partial _{\theta } p(y ; \theta ) = 0 $, the sum extends only over $y $'s for which $p(y ; \theta ) \neq 0 . $   
\end{rmk}   

Turning to quantum conditional expectations, as noted in Remark~\ref{rmk:QCE}, minimizers $f_0 (y) $ of ~\eqref{eq:Q_min_prob_left}) (or minimizers of \eqref{eq:Q_min_prob_right}) 
are no longer unique (all the minimizers are given in \eqref{eq:QCE_degenerate}). However, $\mathbb{E } _{\hat{\rho } } ( |f_0 (\hat{Y } ) |^2 ) $ is independent of the 
choice of minimizer and given by      
\begin{equation}   
\mathbb{E }_{\hat{\rho } } \bigl( |\mathbb{E }^{\ell , 0 } _{\hat{\rho } } (\hat{X } | \hat{Y } ) |^2 ) = \sum _{{\rm Tr } (\hat{\Pi }^{\hat{Y } } _y \hat{\rho } ) \neq 0 } \frac{ |{\rm Tr } (\hat{\Pi }^{\hat{Y } } _y \hat{X } \hat{\rho } ) |^2 }{ {\rm Tr } (\hat{\Pi }^{\hat{Y } } _y \hat{\rho } ) } ,      
\end{equation}    
where it is convenient to introduce a distinguished minimizer   
\begin{equation}   
\mathbb{E }^{\ell , 0 } _{\hat{\rho } } (\hat{X } | \hat{Y } ) := \sum _{{\rm Tr } (\hat{\Pi }^{\hat{Y } } _y \hat{\rho } ) \neq 0 } \frac{ {\rm Tr } (\hat{\Pi }^{\hat{Y } } _y \hat{X } \hat{\rho } ) }{{\rm Tr } (\hat{\Pi }^{\hat{Y } } _y \hat{\rho } ) } \hat{\Pi }^{\hat{Y } } _y ,   
\end{equation}   
which is the minimizer with minimal Hilbert-Schmidt norm. We have that $\mathbb{E }_{\hat{\rho } } \bigl( |\mathbb{E }^{\ell , 0 } _{\hat{\rho } } (\hat{X } | \hat{Y } ) |^2 ) \leq \mathbb{E }_{\hat{\rho }} (\hat{X }^{\dagger } \hat{X } ) $ with equality if $\hat{Y } = \hat{X } . $ Similarly, if we define   
$\mathbb{E }^{{\rm sa }, 0 } _{\hat{\rho } } (\hat{X } | \hat{Y } ) = {\rm Re } \, \mathbb{E }^{\ell , 0 } _{\hat{\rho } } (\hat{X } | \hat{Y } ) $, then for self-adjoint $\hat{X } $, $\mathbb{E }_{\hat{\rho } } \bigl( |\mathbb{E }^{{\rm sa } , 0 } _{\hat{\rho } } (\hat{X } | \hat{Y } ) |^2 )\leq \mathbb{E }_{\hat{\rho } } (\hat{X } ^2 ) $, with equality if $\hat{Y } = \hat{X } . $ This will be used below for the proof of Braunstein-Caves.

To finally make the connection with the quantum Fisher information, recall that a {\it symmetric logarithmic derivate} of a $C^1 $-family $\hat{\rho } (\theta ) $ is defined to be any self-adjoint operator $\hat{L } (\theta ) $ satisfying     
$$   
\partial _{\theta } \hat{\rho } (\theta ) = \frac{1 }{2 } \{ \hat{L } (\theta )    
, \hat{\rho } (\theta ) \} ,      
$$   
where $\{ \hat{C } , \hat{D } \} = \hat{C } \hat{D } + \hat{D } \hat{C} $ is the anti-commutator~\cite{helstrom1976,holevo1982,braunsteincaves1994,paris2009}. In general, if $\hat{R } \geq 0 $ and $\hat{S } $ are self-adjoint operators on $\mathcal{H } $, the operator (or matrix) equation $\hat{L } \hat{R } + \hat{R } \hat{L } = \hat{S } $ has a solution iff $\hat{\Pi }^{\hat{R } } _0 \hat{S } \hat{\Pi }^{\hat{R } } _0 = 0 $, where $\hat{\Pi }^{\hat{R } } _0 $ is the orthogonal projection onto the kernel of $\hat{R } . $ If $|\psi _j\rangle $ is an orthogonal basis of eigenvectors of $\hat R $ with eigenvalues $r_j $ (repeated according to their multiplicity), then the general solution is given by   
\begin{equation}
\label{eq:log derivative}
\hat{L } = \sum _{r_j + r_k > 0 } \frac{ \langle \psi _j , \hat{S } \psi _k \rangle }{r_j + r_k } | \psi _j \rangle \langle \psi _k | + \hat{\Pi }^{\hat{R } } _0 \hat{C } \hat{\Pi }^{\hat{R } } _0 ,   
\end{equation}
where $\hat{C} \in \mathcal{L}(\mathcal{H } ) $ is arbitrary. Since $\hat{S } $ is self-adjoint, $\hat{L } $ is self-adjoint iff the final term on the right is, and $\mathbb{E }(\hat{L }^2 \hat{R } ) $ is independent of the choice of self-adjoint solution $\hat{L }$. 

\medskip   
   
We can apply this with $\hat{R } = \hat{\rho } (\theta ) $ and $\hat{S } = \partial _{\theta } \hat{\rho } (\theta ) . $ The necessary and sufficient condition for solvability of the defining equation for $\hat{L } (\theta ) $ is automatically satisfied, and a consequence of the positivity of $\hat{\rho } (\theta) $: suppose that ${\rm Ker } (\hat{\rho } (\theta _0 ) ) \neq 0 $, and let $\hat{\Pi } _0 (\theta _0 ) := \hat{\Pi } ^{\hat{\rho } (\theta _0 ) } _0 $ be the corresponding orthogonal projector. Then the positive (= non-negative) function $\langle v , \hat{\Pi }_0 (\theta _0 ) \hat{\rho } (\theta ) \hat{\Pi }_0 (\theta _0 ) v \rangle $ vanishes in $\theta = \theta _0 $, for any vector $v \in \mathcal{H } $, which implies that its derivative in $\theta _0 $ vanishes :    $$   
\langle v , \hat{\Pi } _0 (\theta _0 ) \partial _{\theta } \hat{\rho } (\theta_0 ) \hat{\Pi} _0  (\theta _0 ) v \rangle = 0 ,   
$$   
Since this holds for arbitrary $v $,  $\hat{\Pi }_0 (\theta _0 ) \partial _{\theta } \hat{\rho } (\theta _0 ) \hat{\Pi } _0 (\theta _0 ) = 0 $, by the usual polarisation argument.      

We can now, following Helstrom \cite{helstrom1976}, define the {\it quantum Fisher information} or QFI by   
\begin{equation}   
\mathrm{I}_{\mathrm{QF}} (\theta ) = \mathbb{E }_{\hat{\rho }(\theta) } (\hat{L } (\theta )^2 ) ,        
\end{equation}
where the expectation on the right is, as we have seen, independent of the choice of symmetric logarithmic derivative. Helmstrom introduced the notion of symmetric logarithmic derivative, and used it to prove the quantum Cram\`er Rao bound that the variance of any unbiased estimator $\theta (\hat{Y } ) $ of $\theta $, constructed from (observations of) some observable $\hat{Y } $, is bounded from below by $\mathrm{I}_{\mathrm{QF}} (\theta )^{-1 } $, by adapting the proof of the classical Cram\`er-Rao bound: see also \cite{holevo1982}.

Symmetric logarithmic derivatives can now be used to express the classical Fisher information $\mathrm{I}_{\mathrm{F}} ( \hat{Y } , \theta  ) $ in term of the quantum conditional expectations.  

\begin{proposition} If $\hat{L }(\theta ) $ is a symmetric logarithmic derivative of the $C^1 $-family of states $\hat{\rho } (\theta ) $, then        
\begin{equation} \label{eq:FI_SLD}     
\mathrm{I}_{\mathrm{F}} ( \hat{Y } , \theta  ) = \mathbb{E }_{\hat{\rho }(\theta) } \bigl(  \mathbb{E }_{\hat{\rho }(\theta) } ^{{\rm sa }, 0 } (\hat{L } (\theta ) | \hat{Y } )^2 \bigr)   
\end{equation}   
\end{proposition}   

\begin{proof} If $p(y, \theta ) \neq 0 $, then         
$$   
\partial _{\theta } \log p (y ; \theta ) = \frac{{\rm Tr } (\hat{\Pi }^{\hat{Y } } _y \partial _{\theta } \hat{\rho }(\theta) ) }{{\rm Tr } (\hat{\Pi }^{\hat{Y } } _y \hat{\rho }(\theta) ) } =   
\frac{{\rm Re } \, {\rm Tr } (\hat{\Pi }^{\hat{Y } } _y \hat{L }(\theta)\hat{\rho }(\theta) ) }{{\rm Tr } (\hat{\Pi }^{\hat{Y } } _y \hat{\rho }(\theta) ) } 
= \mathbb{E }_{\hat{\rho }(\theta) } ^{{\rm sa } , 0 } (\hat{L }(\theta) | \hat{Y } = y ),      
$$   
from which the formula follows after squaring, multiplying by $p(y; \theta ) $ and summing over the $y $'s for which $p(y; \theta ) \neq  0 . $
\end{proof}   

Since the expectation on the right hand side of (\ref{eq:FI_SLD}) is maximized if $\hat{Y } = \hat{L }(\theta) $, we obtain the Braunstein and Caves theorem as an immediate corollary:   

\begin{corollary} 
\label{cor:Braunstein Caves}(\cite{braunsteincaves1994}) If $\hat{\rho }(\theta) $ is a $C^1 $-family of non-singular states, then     
\begin{equation}   
\max _{\hat{Y } \, {\rm observ. } } \mathrm{I}_{\mathrm{F}} ( \hat{Y } ;\theta ) = \mathbb{E }_{\hat{\rho } (\theta ) } \bigl( \hat{L } (\theta )^2 \bigr) = \mathrm{I}_{\mathrm{QF}} (\theta ) .     
\end{equation}     
\end{corollary}   

The reason we cannot just work with quantum conditional expectations for which $\hat{\rho } \in D_{\hat{Y } } $ is that it is not necessarily true that $\hat{\rho }(\theta ) \in D_{\hat{L } (\theta ) } $, preventing us to then take $\hat{Y } = \hat{L }(\theta ) . $ For example, if $\hat{\rho } (\theta ) = | \psi (\theta ) \rangle \langle \psi (\theta ) | $ is a family of pure states with $\psi (\theta ) $ a $\mathcal{C}^1 $-family of unit vectors, then $\hat{\rho } (\theta ) ^2 = \hat{\rho } (\theta ) $ implies that $\hat{L }(\theta) = 2\partial _{\theta } \hat{\rho } (\theta ) = 2 ( | \partial _{\theta } \psi (\theta ) \rangle \langle \psi (\theta ) | + | \psi (\theta ) \rangle \langle \partial _{\theta } \psi (\theta ) |) $ is a symmetric logarithmic derivative, and if $\chi \perp \psi (\theta ) , \partial _{\theta } \psi (\theta ) $, then $\chi $ is an eigenvector of $\hat L(\theta ) $ with eigenvalue 0, for which $\langle \chi , \hat{\rho } (\theta ) \chi \rangle = 0 $, so  $\hat{\rho }(\theta ) \notin D_{\hat{L } (\theta ) } . $   
\medskip   
   
We note for completeness, that a short computation shows that the quantum Fisher information of such a family of pure states is given by the well-known formula         
$$   
\mathrm{I}_{\mathrm{QF}} (\theta ) = 4 (|| \partial _{\theta } \psi ||^2 - | \langle \psi , \partial _{\theta } \psi\rangle |^2 ) ,   
$$   
with the right hand side evaluated in $\theta $, and where we used that $\langle \psi , \partial _{\theta } \psi \rangle = - \langle \partial _{\theta } \psi , \psi \rangle . $

If we take $\hat{\rho }(\theta) $ to be the Born state associated to an observable $\hat{Y } $,   
$$   
\hat{\rho }(\theta) = \sum _y p(y; \theta ) \hat{\Pi }^{\hat{Y } } _y ,   
$$  
then the quantum Fisher information associated to this state is precisely the classical Fisher information, $\mathrm{I}_{\mathrm{F}} ( \hat{Y } ;\theta ) . $ We examine what happens near an isolated point $\theta _0 $ where this state is singular: $p(y; \theta _0 ) = 0  $ for certain values of $y$. Assuming that $p(y ; \theta ) $ is analytic (or at least not flat at $\theta _0 $), positivity implies that the first non-zero derivative in $\theta _0 $ is of even order, say of order $2k_y $, $k_y \in \mathbb{N }^* . $ Letting $\pi _{2k } (y) := \partial _{\theta } ^{2 k_y } p(y ; \theta ) $, Taylor expansion of the numerator and denominator shows that   
$$   
\frac{(\partial _{\theta } p(y ; \theta ) )^2 }{p(y ; \theta ) } = \frac{(2k )! }{(2k - 1 )! ^2 } (\theta - \theta _0 )^{2k - 2 } + \cdots ,   
$$   
whose  limit as $\theta \to \theta _0 $ is 0 unless $k = 1 $, in which case it is $2 \pi _2 . $ We therefore have the curious result that   
$$   
\lim _{\theta \to \theta _0 } \mathrm{I}_{\mathrm{F}} ( \hat{Y } ;\theta ) = \mathrm{I}_{\mathrm{F}} ( \hat{Y } ;\theta _0 ) + \sum _{y: p(y; \theta _0 ) = 0 } 
2 \partial _{\theta } ^2 p (y ; \theta _0 ) ,
$$   
showing that the Fisher information is in general not continuous as function of the parameter, unless all maxima and minima of the probabilities $p(y; \theta ) $ are degenerate (non-Morse). Doing a similar analysis of the multi-parameter case would be more complicated, the limit of the quotient of two quadratic forms does not exist.  

\section{Variational characterisations of KD}\label{app:variational_char}   
The identification of the conditional expectations defined by the left or right KD distribution with the corresponding quantum conditional expectations leads to a variational characterisation of the KD quasiprobabilities.   
In the proof of Theorem \ref{thm:Uniqueminimizers} the minimisation is done for each $y \in \sigma (\hat{Y } ) $ separately (see Eq.~\eqref{eq:quadform}) and the proof shows in fact that    
\begin{equation} \label{eq:var_char_QCE}
\mathbb{E }_{\hat{\rho } } (\hat{X } | \hat{Y } = y ) = \frac{{\rm Tr } (\hat{\Pi }^{\hat{Y } } _y \hat{X } \hat{\rho } ) }{{\rm Tr } (\hat{\Pi }^{\hat{Y } } _y \hat{\rho } ) } =   
{\rm Argmin }_{\lambda \in \mathbb{C } } \mathbb{E }_{\hat{\rho } } \bigl( | \hat{X } - \lambda \hat{\Pi }^{\hat{Y } } _y |^2 \bigr) . 
\end{equation}   
Applying this with $\hat{X } = \hat{\Pi }^{\hat{A } } _a $ and $\hat{Y } = \hat{B } $, where $\hat{A } $ and $\hat{B } $ are two complementary CSCOs,   
this translates into   
\begin{equation} \label{eq:var_char_KD}   
\frac{Q ^{{\rm KD } , \ell } _{a, b } }{\langle \varphi ^{\hat{B } } _b , \hat{\rho } \varphi ^{\hat{B } } _b \rangle } =   
{\rm Argmin } _{ \lambda  \in \mathbb{C } } \mathbb{E }_{\hat{\rho } } \bigl( |\hat{\Pi }^{\hat{A } } _a  - \lambda \hat{\Pi }^{\hat{B } } _b |^2 \bigr) 
\end{equation}   
Now if a function $f(\lambda ) $ of a complex variable $\lambda $ attains its minimum in $\lambda _0 $ and if $\beta \in \mathbb{C } $, then the function $f_{\beta } (\lambda ) := f(\lambda / \beta ) $ has its minimum in $\beta \lambda _0 $, since $f(\lambda / \beta ) \geq f(\lambda _0 ) = f_{\beta } (\beta \lambda _0 ) $, and this remains of course true for the function $\beta f_{\beta } (\lambda ) $ if $\beta > 0 . $  Applying this to the function which is minimized    
in the right hand side of (\ref{eq:var_char_KD}) with $\beta =   
\langle \varphi ^{\hat{B } } _b , \hat{\rho } \varphi ^{\hat{B } } _b \rangle $ we find:   
   
\begin{proposition}   
\begin{equation} \label{eq:var_char_KD2}   
Q^{{\rm KD } , \ell } _{a, b } (\hat{\rho } ) = {\rm Argmin }_{\lambda \in \mathbb{C } } \ \mathbb{E }_{\hat{\rho } } \left( \bigl | \langle \varphi ^{\hat{B } } _b , \hat{\rho } \varphi ^{\hat{B } } _b \rangle  \hat \Pi ^{\hat{A } } _a - \lambda \hat \Pi ^{\hat{B } } _b \bigr |^2 \right) .   
\end{equation}   
\end{proposition}   
   
\noindent In words, $Q^{{\rm KD } , \ell } _{a, b } (\hat{\rho }) \hat{\Pi }^{\hat{B } } _b $ is the best approximation of the spectral projector $\hat{\Pi }^{\hat{A } } _a $ of $a \in \sigma (\hat{A } ) $, weighted by the Born probability of $b \in \sigma (\hat{B } ) $ , by a multiple of the spectral projector $\hat{\Pi }^{\hat{B } } _b . $     
\medskip   
   
There is an second variational characterisation which follows from the observation that if $\hat{\rho }_0 := d^{-1 } \id{d} $ is the maximally mixed state ($d $ being the dimension of $\mathcal{H } $), then 
\begin{equation}   
Q^{{\rm KD } , \ell } _{a, b } (\hat{\rho } )= {\rm Tr } ( \hat{\Pi } ^{\hat{B } } _b \hat{\Pi } ^{\hat{A } }_a \hat{\rho } ) =   
\mathbb{E }_{\hat{\rho }_0 } \bigl(\hat{\Pi } ^{\hat{A } }_a \hat{\rho } | \hat{B } = b ) ,   
\end{equation}   
${\rm Tr } ( \hat{\Pi } ^{\hat{B } } _b \hat{\rho }_0 ) $ being 1 since $\hat{B } $ is a CSCO. Formula (\ref{eq:var_char_QCE}) then immediately implies:   

\begin{proposition}   
\begin{equation} \label{eq:var_char_KD3}   
Q^{{\rm KD } , \ell } _{a, b } (\hat{\rho } ) = {\rm Argmin }_{\lambda \in \mathbb{C } } \ {\rm Tr } \bigl( | \hat{\Pi } ^{\hat{A } } _a  \hat{\rho } -  \lambda \hat{\Pi } ^{\hat{B } } _b |^2 \bigr)   
\end{equation}   
\end{proposition}   

Another way to derive (\ref{eq:var_char_KD3}) is by observing that   
\begin{equation} \nonumber   
\sum _{b \in \sigma (\hat{B } ) } Q ^{{\rm KD } , \ell } _{a, b } (\hat \rho)\hat{\Pi }^{\hat{B } } _b = \mathcal{E }_{\hat{B } } \bigl( \hat{\Pi } ^{\hat{A } }_a \hat{\rho } \bigr) ,   
\end{equation}     
where $\mathcal{E }_{\hat{B } } $ is defined in (\ref{eq:VNexpectation}). Since $\mathcal{E }_{\hat{B } } $ is the orthogonal projection with respect to the Hilbert-Schmidt norm of $\mathcal{L}(\mathcal{H } ) $ onto $\mathcal{F }_{\hat{B } , \mathbb{C } } $, this implies (\ref{eq:var_char_KD3}).   

We finally note that formulas (\ref{eq:var_char_KD2}) and (\ref{eq:var_char_KD3}) can also easily be verified directly.

\end{document}